\documentclass[sigconf,
]{acmart}
\AtBeginDocument{%
  }

\setcopyright{acmlicensed}
\copyrightyear{2018}
\acmYear{2018}
\acmDOI{XXXXXXX.XXXXXXX}
\acmConference[Conference acronym 'XX]{Make sure to enter the correct
  conference title from your rights confirmation email}{June 03--05,
  2018}{Woodstock, NY}
\acmISBN{978-1-4503-XXXX-X/2018/06}

\newcommand{\macrospath}{./Macros}

\usepackage{ifthen}
 \usepackage{xspace}

\newboolean{talk}
\setboolean{talk}{false}
\newboolean{paper}
\setboolean{paper}{false}

\newboolean{IEEEstyle}
\setboolean{IEEEstyle}{false}
\newboolean{lipicsstyle}
\setboolean{lipicsstyle}{false}
\newboolean{eptcsstyle}
\setboolean{eptcsstyle}{false}
\newboolean{entcsstyle}
\setboolean{entcsstyle}{false}
\newboolean{sigplanstyle}
\setboolean{sigplanstyle}{false}
\newboolean{easychairstyle}
\setboolean{easychairstyle}{false}
\newboolean{scrartclstyle}
\setboolean{scrartclstyle}{false}
\newboolean{lncsstyle}
\setboolean{lncsstyle}{false}
\newboolean{acmartstyle}
\setboolean{acmartstyle}{false}

\newboolean{needstheorems}
\setboolean{needstheorems}{false}
\newboolean{withimages}
\setboolean{withimages}{false}
\newboolean{withproofs}
\setboolean{withproofs}{true}
\newboolean{complexity}
\setboolean{complexity}{false}

\newboolean{techr}
\setboolean{techr}{false}

\newboolean{french}
\setboolean{french}{false}

\newboolean{colored}
\setboolean{paper}{false}

\setboolean{colored}{false}
\setboolean{withimages}{true}
\setboolean{acmartstyle}{true}
  \ifthenelse{\boolean{lipicsstyle}}{\usepackage{mathabx}}{}
\ifthenelse{\boolean{talk}}{}{\usepackage[utf8]{inputenc}}
\ifthenelse{\boolean{french}}{\usepackage[french]{babel}}{
  \ifthenelse{\boolean{lipicsstyle}}{}{\usepackage[english]{babel}}}

\ifthenelse{\boolean{acmartstyle}}{}{
\usepackage{amssymb}
\usepackage{amsmath}
}

\makeatletter
\@ifpackageloaded{stix}{%
}{%
  \DeclareFontEncoding{LS2}{}{\noaccents@}
  \DeclareFontSubstitution{LS2}{stix}{m}{n}
  \DeclareSymbolFont{stix@largesymbols}{LS2}{stixex}{m}{n}
  \SetSymbolFont{stix@largesymbols}{bold}{LS2}{stixex}{b}{n}
  \DeclareMathDelimiter{\lBrace}{\mathopen} {stix@largesymbols}{"E8}%
                                            {stix@largesymbols}{"0E}
  \DeclareMathDelimiter{\rBrace}{\mathclose}{stix@largesymbols}{"E9}%
                                            {stix@largesymbols}{"0F}
}
\makeatother

\usepackage{bussproofs}
\usepackage{stmaryrd}
\usepackage[c]{esvect} 
\usepackage{appendix}
\usepackage[normalem]{ulem} 
\usepackage{actuarialangle} 
\usepackage{wasysym} 

\ifthenelse{\boolean{talk}}{\usepackage{color}}{
 	\ifthenelse{\boolean{lipicsstyle}}{\usepackage{color}}{
   		\ifthenelse{\boolean{acmartstyle}}{\usepackage{color}}{
			\usepackage[usenames]{xcolor}
		}
	}
}
 \ifthenelse{\boolean{IEEEstyle}}{
 \usepackage[bookmarks={false}]{hyperref}}{
 \usepackage{hyperref}}
 \usepackage{cleveref}
  \usepackage{hhline}
 
\ifthenelse{\boolean{IEEEstyle}}{
\setboolean{needstheorems}{true}}{}

\ifthenelse{\boolean{eptcsstyle}}{
\setboolean{needstheorems}{true}}{}

\ifthenelse{\boolean{scrartclstyle}}{
\setboolean{needstheorems}{true}}{}

\ifthenelse{\boolean{sigplanstyle}}{
\setboolean{needstheorems}{true}}{}

\ifthenelse{\boolean{easychairstyle}}{
\setboolean{needstheorems}{true}}{}

\ifthenelse{\boolean{lipicsstyle}}{
    }{}

\ifthenelse{\boolean{needstheorems}}{
\input{\macrospath/ben-macros-thm-environments}}{}

\newcommand{\myproof}[1]{
\ifthenelse{\boolean{withproofs}}{#1}{}
}

\usepackage{graphicx}

\newcommand{\la}[1]{\lambda #1.}

\newcommand{\tm}{t}
\newcommand{\tmtwo}{u}
\newcommand{\tmthree}{s}
\newcommand{\tmfour}{r}

\newcommand{\var}{x}
\newcommand{\vartwo}{y}
\newcommand{\varthree}{z}
\newcommand{\varfour}{w}

\newcommand{\dom}[1]{\symfont{dom}(#1)}

\newcommand{\rootRew}[1]{\mapsto_{#1}}
\newcommand{\Rew}[1]{\rightarrow_{#1}}

\renewcommand{\to}{\Rew{}}

\newcommand{\tostrat}{\Rew{x}}

\newcommand{\symfont}[1]{\mathtt{#1}}

\newcommand{\ovsym}{oh}

\newcommand{\nsym}{\symfont{n}}

\newcommand{\cbv}{{CbV}\xspace}

\newcommand{\val}{v}
\newcommand{\valtwo}{\val'}
\newcommand{\valthree}{\val''}

\newcommand{\ctxholep}[1]{\langle #1\rangle}
\newcommand{\ctxhole}{\ctxholep{\cdot}}

\newcommand{\ctx}{C}

\newcommand{\evctx}{C}
\newcommand{\evctxtwo}{\evctx'}

\newcommand{\evctxp}[1]{\evctx\ctxholep{#1}}

\newcommand{\evctxtwop}[1]{\evctxtwo\ctxholep{#1}}

\newcommand{\nbvctxtwo}[1]{\nbvctxtwo{#1}}

\newcommand{\defeq}{:=}

\newcommand{\grameq}{::=}

\newcommand{\esub}[2]{[#1{\shortleftarrow}#2]}
\newcommand{\isub}[2]{\{#1{\shortleftarrow}#2\}}

\ifthenelse{\boolean{complexity}}{ }{
  
  }

\newcommand{\llbrace}{\{ \kern -0.27em \vert}
\newcommand{\rrbrace}{\vert \kern -0.27em \}}

\renewcommand{\l}{\lambda}
\newcommand{\ie}{i.e.\xspace}
\newcommand{\eg}{e.g.\xspace}
\newcommand{\ih}{{\textit{i.h.}}\xspace}
\newcommand{\fv}[1]{\symfont{fv}(#1)}

\ifthenelse{\boolean{entcsstyle}}{}{
	}

\newcommand{\ben}[1]{{\color{red} {#1}}}

\newcommand{\ignore}[1]{}

\newcommand{\myinput}[1]{\ifthenelse{\boolean{withimages}}{\input{#1}}{}}
\newcommand{\inputproof}[1]{\ifthenelse{\boolean{withproofs}}{\input{#1}}{}}

\newcommand{\reflemma}[1]{Lemma~\ref{l:#1}}
\newcommand{\reflemmap}[2]{Lemma~\ref{l:#1}.\ref{p:#1-#2}}

\newcommand{\reflemmaeq}[1]{{L.\ref{l:#1}}}
\newcommand{\reflemmaeqp}[2]{{L.\ref{l:#1}.\ref{p:#1-#2}}}

\newcommand{\refthm}[1]{Theorem~\ref{thm:#1}}
\newcommand{\refthmp}[2]{Theorem~\ref{thm:#1}.\ref{p:#1-#2}}

\newcommand{\refprop}[1]{Proposition~\ref{prop:#1}}
\newcommand{\refpropeq}[1]{Prop.~\ref{prop:#1}}
\newcommand{\refsect}[1]{Sect.~\ref{sect:#1}}

\ifthenelse{\boolean{easychairstyle}}{}{
  \ifthenelse{\boolean{lncsstyle}}{}{	
	
  }
}
\newcommand{\reffig}[1]{Fig.~\ref{fig:#1}}

\newcommand{\refdef}[1]{Definition~\ref{def:#1}}
\newcommand{\refdefeq}[1]{Def.~\ref{def:#1}}

\newcommand{\refpoint}[1]{Point~\ref{p:#1}}

\newcommand{\set}[1]{\{#1\}}

\newcommand{\nat}{\mathbb{N}}

\newcommand{\betav}{\beta_\val}

\newcommand{\tobv}{\Rew{\betav}}

\newcommand{\rtobv}{\rootRew{\betav}}

\newcommand{\size}[1]{|#1|}
\newcommand{\sizep}[2]{|#1|_{#2}}

\newcommand{\clos}{c}
\newcommand{\clostwo}{\clos'}

\newcommand{\env}{E}
\newcommand{\envtwo}{\env'}

\newcommand{\stack}{S}
\newcommand{\stacktwo}{\stack'}

\ifthenelse{\boolean{acmartstyle}}{
	\renewcommand{\state}{q}
	}{
	\newcommand{\state}{s}
}
\newcommand{\statetwo}{\state'}
\newcommand{\statethree}{\state''}

\newcommand{\pairstate}[2]{(#1\mid#2)}
\newcommand{\fourstate}[4]{(#1\mid#2\mid#3\mid#4)}

\newcommand{\exec}{\rho}
\newcommand{\exectwo}{\sigma}

\newcommand{\decode}[1]{\underline{#1}}
\newcommand{\decodep}[2]{\decode{#1}\ctxholep{#2}}

\newcommand{\emptystack}{\epsilon}
\newcommand{\cons}{:}

\newcommand{\deriv}{\ensuremath{e}}
\newcommand{\derivtwo}{\deriv'}
\newcommand{\derivthree}{\deriv''}

\newcommand{\rename}[1]{#1^\alpha}

\newcommand{\lab}{l}

\newcommand{\pair}[2]{(#1,#2)}

\newcommand{\mach}{\symfont{M}}

\newcommand{\tomachhole}[1]{\leadsto_{#1}}
\newcommand{\tomach}{\tomachhole{}}

\newcommand{\compil}[1]{#1^\circ}

\newcommand{\subsym}{sub}
\newcommand{\seasym}{sea}
\newcommand{\prsym}{pr}
\newcommand{\tomachine}{\tomachhole\mach}
\newcommand{\tomachpr}{\tomachhole{pr}}
\newcommand{\tomachbeta}{\tomachhole{\betav}}
\newcommand{\tomachbetaev}{\tomachhole{\evsym\betav}}

\newcommand{\tomachproj}{\tomachhole{\pi}}
\newcommand{\tomachprojev}{\tomachhole{\evsym\pi}}

\newcommand{\tomachsub}{\tomachhole{\subsym}}
\newcommand{\tomachsubnv}{\tomachhole{\nsubsym}}

\newcommand{\subvsym}{\subsym_\val}
\newcommand{\subwsym}{\subsym_w}
\newcommand{\nsubsym}{\nvsym\subsym}
\newcommand{\nsubvsym}{\nvsym\subvsym}
\newcommand{\nsubwsym}{\nvsym\subwsym}
\newcommand{\tomachsubvev}{\tomachhole{\nsubvsym}}
\newcommand{\tomachsubwev}{\tomachhole{\nsubwsym}}

\newcommand{\tomachseaone}{\tomachhole{\seasym_1}}
\newcommand{\tomachseatwo}{\tomachhole{\seasym_2}}

\newcommand{\tomachseaoneev}{\tomachhole{\evsym\seasym_{1}}}
\newcommand{\tomachseaonenv}{\tomachhole{\nvsym\seasym_{1}}}
\newcommand{\tomachseatwonv}{\tomachhole{\nvsym\seasym_{2}}}
\newcommand{\tomachseathreeev}{\tomachhole{\evsym\seasym_{3}}}
\newcommand{\tomachseathreenv}{\tomachhole{\nvsym\seasym_{3}}}
\newcommand{\tomachseafournv}{\tomachhole{\nvsym\seasym_{4}}}
\newcommand{\tomachseafivenv}{\tomachhole{\nvsym\seasym_{5}}}

\newcommand{\tomachseasixev}{\tomachhole{\evsym\seasym_{6}}}

\newcommand{\tomachseatupev}{\tomachhole{\evsym\seasym_6}}

\newcommand{\tomachseasevenev}{\tomachhole{\evsym\seasym_{7}}}

\newcommand{\tomachseaactev}{\tomachhole{\evsym\seasym_{7}}}

\newcommand{\sizepr}[1]{\sizep{#1}{\prsym}}
\newcommand{\sizebeta}[1]{\sizep{#1}{\betav}}

\newcommand{\midd}{\mid}

\newcommand{\withproofs}[1]{\ifthenelse{\boolean{withproofs}}{#1}{}}

\newcommand{\withoutproofs}[1]{\ifthenelse{\boolean{withproofs}}{}{#1}}

\ifthenelse{\boolean{entcsstyle}}{}{
  \ifthenelse{\boolean{lncsstyle}}{}{	
    
  }
 }

\newcounter{numberone}

\newcommand{\bigo}{{\mathcal{O}}}

\newcommand{\emptyenv}{\epsilon}

\newcommand{\tmtwop}{\tmtwo'}

\ifthenelse{\boolean{acmartstyle}}{}{
	\newcommand{\terms}{T}
}

\definecolor{dgreen}{rgb}{0.0, 0.5, 0.0}

\newcommand{\Id}{\symfont{I}}

\newcommand{\tomacho}{\tomachov}
\newcommand{\tomachov}{\tomachhole{\ovsym}}
\newcommand{\length}[1]{|#1|}

\makeatletter
\newcommand{\exder}{%
  \def\exderW[##1]{\triangleright_{##1}\ }%
  \def\exderWO{\triangleright\ }%
  \@ifnextchar[\exderW\exderWO%
  }
\makeatother

\newcommand\Deribbase[5]{{#3}\ {\pmb\vdash}_{#2}^{#1} {#4}\  {:}\  {#5}}

\makeatletter

\newcommand{\Deribase}[1]{%
  \def\DeribW[##1]{\Deribbase{##1}{#1}}%
  \def\DeribWO{\Deribbase{}{#1}}%
  \@ifnextchar[\DeribW\DeribWO%
  }

  \newcommand{\Deri}{%
  \def\DeriW_##1{\Deribase{##1}}%
  \def\DeriWO{\Deribase{}}%
  \@ifnextchar_\DeriW\DeriWO%
  }
\makeatother

\newcommand{\app}{\symfont{app}}

\makeatletter
\newcommand{\appresult}{%
  \def\appresultW<##1>{\app_\result^{##1}}%
  \def\appresultWO{\app_\result}%
  \@ifnextchar<\appresultW\appresultWO%
  }
\makeatother

\newcommand\mydots{\hbox to .6em{.\hss.}}

\newcommand{\cameratech}[2]{\ifthenelse{\boolean{techr}}{#2}{#1}}

	\ifthenelse{\boolean{colored}}{
\newcommand{\tuple}[1]{\orange{\llparenthesis}#1\orange{\rrparenthesis}}
}{
\newcommand{\tuple}[1]{\llparenthesis#1\rrparenthesis}
}

	\ifthenelse{\boolean{colored}}{
\newcommand{\mycyan}[1]{\cyan{#1}}
}{
\newcommand{\mycyan}[1]{#1}
}

\newcommand{\proj}{\pi}

\ifthenelse{\boolean{colored}}{
	\newcommand{\nvsym}{\red\bullet}
	\newcommand{\evsym}{\blue\bullet}
}{
	\newcommand{\nvsym}{{\circ}}
	\newcommand{\evsym}{{\bullet}}
}

\newcommand{\bareclos}{c}
\newcommand{\laxclos}{\flag\bareclos}
\newcommand{\nvlab}[1]{\nvsym{#1}}
\newcommand{\evlab}[1]{\evsym{#1}}
\renewcommand{\clos}{\evlab\bareclos}
\renewcommand{\clostwo}{\evlab{\bareclos'}}
\newcommand{\nvclos}{\nvlab\bareclos}

\newcommand{\clc}[3]{\overline{#1}^{#2,#3}}
\newcommand{\lali}[1]{\overline{#1}}

\newcommand{\projsym}{\pi}
\newcommand{\rtoproj}{\rootRew{\projsym}}
\newcommand{\toproj}{\Rew{\projsym}}

\newcommand{\cbvsym}{\symfont{cbv}}
\newcommand{\intprefix}{\symfont{i}}
\newcommand{\tarprefix}{\symfont{t}}
\newcommand{\rtollbv}{\rootRew{\intprefix\betav}}
\newcommand{\tollbv}{\Rew{\intprefix\betav}}
\newcommand{\rtoibv}{\rootRew{\intprefix\betav}}
\newcommand{\toibv}{\Rew{\intprefix\betav}}
\newcommand{\rtoccbv}{\rootRew{\tarprefix\betav}}
\newcommand{\toccbv}{\Rew{\tarprefix\betav}}

\newcommand{\totbv}{\Rew{\tarprefix\betav}}

\newcommand{\tollproj}{\Rew{\intprefix\projsym}}
\newcommand{\rtoiproj}{\rootRew{\intprefix\projsym}}
\newcommand{\toiproj}{\Rew{\intprefix\projsym}}
\newcommand{\rtoccproj}{\rootRew{\tarprefix\projsym}}
\newcommand{\toccproj}{\Rew{\tarprefix\projsym}}

\newcommand{\totproj}{\Rew{\tarprefix\projsym}}

\newcommand{\tollcbv}{\Rew{\symfont{\laliprefix cbv}}}

\newcommand{\norm}[1]{\lVert #1 \rVert}
\newcommand{\prnorm}[1]{\lVert #1 \rVert_{\vvar}}
\newcommand{\fnorm}[1]{\lVert #1 \rVert_{\efvar}}
\newcommand{\nnorm}[1]{\lVert #1 \rVert_{\envar}}
\newcommand{\lnorm}[1]{\lVert #1 \rVert_{\efvar}}

\newcommand{\omeas}[1]{\sizep{#1}{\ovsym}}

\newcommand{\stacke}{\stack_{en}}

\newcommand{\ars}{A}
\newcommand{\arstwo}{\ars'}

\ifthenelse{\boolean{acmartstyle}}{
	
	}{
	
	}
	
	\ifthenelse{\boolean{colored}}{
\newcommand{\pack}[2]{\purple{\llbracket}#1\purple{|}#2\purple{\rrbracket}}
\newcommand{\packnmp}[3]{\purple{\llbracket}#1\purple{|}#2\purple{\rrbracket}_{#3}}
}{
\newcommand{\pack}[2]{\llbracket#1|#2\rrbracket}
\newcommand{\packnmp}[3]{\llbracket#1|#2\rrbracket_{#3}}
}

\newcommand{\packnm}[2]{\packnmp{#1}{#2}{n,m}}
\newcommand{\wid}[1]{\symfont{wd}(#1)}
\newcommand{\hg}[1]{\symfont{hg}(#1)}

\newcommand{\unlali}[1]{\underline{#1}}
\newcommand{\machctxhole}{\downarrow}
\newcommand{\dabs}[2]{#1;#2.}
\newcommand{\dabsv}[2]{\dabs{\tuv{#1}}{\tuv{#2}}}

\newcommand{\varlist}{\tuvar}
\newcommand{\vartwolist}{\tuvartwo}
\newcommand{\varthreelist}{\tuvarthree}
\newcommand{\varfourlist}{\tuvarfour}

\newcommand{\TAM}{Source TAM\xspace}

\newcommand{\totam}{\tomachhole{\textup{\tiny STAM}}}

\newcommand{\LTAM}{Int TAM\xspace}

\newcommand{\TTAM}{Target TAM\xspace}

\newcommand{\tottam}{\tomachhole{\textup{\tiny TTAM}}}

\def\bwcirc{\raisebox{1pt}{\scalebox{0.6}{\LEFTcircle}}}

\newcommand{\sousym}{\symfont{sou}}
\newcommand{\intsym}{\symfont{int}}
\newcommand{\tarsym}{\symfont{tar}}
\newcommand{\cbvcal}{\lambda_{\cbvsym}}
\newcommand{\soucal}{\lambda_{\sousym}}
\newcommand{\intcal}{\lambda_{\intsym}}
\newcommand{\tarcal}{\lambda_{\tarsym}}

\newcommand{\tosou}{\Rew{\sousym}}
\newcommand{\toint}{\Rew{\intsym}}
\newcommand{\totar}{\Rew{\tarsym}}
\renewcommand{\tollcbv}{\toint}

\newcommand{\bag}{b}
\newcommand{\bagtwo}{b'}

\newcommand{\valnone}{\val_1}
\newcommand{\valntwo}{\val_2}
\newcommand{\valnthree}{\val_3}

\newcommand{\vvvalnone}{\vv\valnone}
\newcommand{\vvvalntwo}{\vv\valntwo}
\newcommand{\vvvalnthree}{\vv\valnthree}

\newcommand{\itsub}[2]{\lBrace#1{\shortleftarrow}#2\rBrace}
\newcommand{\itsubp}[4]{\lBrace#1;#3{\shortleftarrow}#2;#4\rBrace}
\newcommand{\isubp}[4]{\{#1;#3{\shortleftarrow}#2;#4\}}
\newcommand{\isubsim}[4]{\{#1{\shortleftarrow}#2,\mydots,#3{\shortleftarrow}#4\}}

\newcommand{\prvar}{p}
\newcommand{\prvartwo}{\prvar'}

\newcommand{\svar}{\mathtt{s}}
\newcommand{\lvar}{\mathtt{l}}
\newcommand{\vvar}{\mathtt{v}}
\newcommand{\envar}{\svar}
\newcommand{\efvar}{\lvar}

\newcommand{\tuvar}{\tuv\var}
\newcommand{\tuvartwo}{\tuv\vartwo}
\newcommand{\tuvarthree}{\tuv\varthree}
\newcommand{\tuvarfour}{\tuv\varfour}

\newcommand{\vvvar}{\tuvar}
\newcommand{\vvvartwo}{\tuvartwo}

\newcommand{\vdisjoint}{\#}

\newcommand{\naming}[1]{\underline{#1}_{\tuvartwo,\tuvar}}
\newcommand{\namingp}[3]{\underline{#1}_{#2,#3}}
\newcommand{\emptylist}{\epsilon}

\newcommand{\States}{\symfont{States}}
\renewcommand{\tostrat}{\Rew{\symfont{str}}}

\newcommand{\run}{\exec}
\newcommand{\runtwo}{\run'}
\newcommand{\runthree}{\run''}

\newcommand{\xcal}{\l_{cal}}

\newcommand{\nfov}[1]{\symfont{nf}_{\ovsym}(#1)} 

\newcommand{\evval}{\evsym\val}
\newcommand{\evvaltwo}{\evsym\valtwo}

\newcommand{\Stackable}{Stackable\xspace}

\newcommand{\LAM}{LAM\xspace}

\newsavebox\tempbox
\let\svwidetilde\widetilde
\renewcommand\widetilde[1]{\sbox\tempbox{$#1$}\svwidetilde{\usebox{\tempbox}}}

\ifthenelse{\boolean{lipicsstyle}}{
	\renewcommand{\flag}{\bwcirc}
	}{
	\newcommand{\flag}{\bwcirc}}

\newcommand{\tuv}[1]{\widetilde{#1}}

\newcommand{\cc}[1]{{#1}^\symfont{cc}}
\usepackage{tikz}
\usetikzlibrary{calc}
\usetikzlibrary{matrix,arrows}
\usetikzlibrary{decorations.shapes}
\usetikzlibrary{decorations.text}
\usetikzlibrary{decorations.pathmorphing}
\usetikzlibrary{decorations.markings}
\usetikzlibrary{positioning}

\tikzset{
node distance=1.3cm, auto,
every node/.style={font=\scriptsize },
ocenter/.style={baseline={([yshift=-.5ex, xshift=-.5ex]current bounding box)}},  
labelBeginAbove/.style={postaction={decorate,decoration={markings,mark=at position 0 with {\node[inner sep= 0.6pt, above=1pt]{\tiny #1};}} } },
labelBeginBelow/.style={postaction={decorate,decoration={markings,mark=at position 0 with {\node[inner sep= 0.6pt, below=1pt]{\tiny #1};}}}},
labelEndAbove/.style={postaction={decorate,decoration={markings,mark=at position 1 with {\node[inner sep= 0.6pt, above=1pt]{\tiny #1};}}}},
labelEndBelow/.style={postaction={decorate,decoration={markings,mark=at position 1 with {\node[inner sep= 0.6pt, below=1pt]{\tiny #1};}}}},
labelEndRight/.style={postaction={decorate,decoration={markings,mark=at position 1 with {\node[inner sep= 0.6pt, right=1pt]{\tiny #1};}}}},
labelEndLeft/.style={postaction={decorate,decoration={markings,mark=at position 1 with {\node[inner sep= 0.6pt, left=1pt]{\tiny #1};}}}}
}

\newcommand{\wrapt}{closure\xspace}
\newcommand{\wrapts}{closures\xspace}
\newcommand{\Wrapt}{Closure\xspace}
\newcommand{\Wrapts}{Closures\xspace}

\newcommand{\bagt}{bag\xspace}
\newcommand{\bagts}{bags\xspace}

\newcommand{\Bagts}{Bags\xspace}

\newcommand{\mclosure}{m-closure\xspace}
\newcommand{\mclosures}{m-closures\xspace}

\newcommand{\Mclosures}{M-Closures\xspace}

\newcommand{\lifting}{wrapping\xspace}
\newcommand{\Lifting}{Wrapping\xspace}
\newcommand{\lifted}{wrapped\xspace}
\newcommand{\Lifted}{Wrapped\xspace}

\newcommand{\collapsing}{unwrapping\xspace}
\newcommand{\Collapsing}{Unwrapping\xspace}

\renewcommand{\exec}{r}
\renewcommand{\exectwo}{\exec'}

\renewcommand{\lvar}{\symfont{w}}

\renewcommand{\unlali}[1]{\lceil#1\rceil}
\renewcommand{\lali}[1]{\underline{#1}}

\renewcommand{\clc}[3]{\graybox{#1}^{\footnotesize#2,#3}}
\renewcommand{\naming}[1]{\namingp{#1}{\tuvartwo}{\tuvar}}
\renewcommand{\namingp}[3]{\darkgraybox{#1}_{#2,#3}}

\definecolor{LightGray}{gray}{.80}
\definecolor{DarkGray}{gray}{.60}
\newcommand{\graybox}[1]{\colorbox{LightGray}{\ensuremath{#1}}}
\newcommand{\darkgraybox}[1]{\colorbox{DarkGray}{\ensuremath{#1}}}

\renewcommand{\subwsym}{\subsym_c}

\usepackage{makecell}

\begin{document}

\title{Closure Conversion, Flat Environments, 
	and the Complexity of Abstract Machines}

\author{Beniamino Accattoli}
\affiliation{%
	\institution{Inria, École Polytechnique}
	\country{France}
}
\email{beniamino.accattoli@inria.fr}

\author{Cl\'audio Belo Louren\c co}
\affiliation{%
	\institution{Huawei Central Software Institute}
	\country{UK}
	}
\email{claudio.lourenco@huawei.com}

\author{Dan Ghica}
\affiliation{%
	\institution{Huawei Central Software Institute, University of Birmingham}
	\country{UK}
}
\email{dan.ghica@huawei.com}

\author{Giulio Guerrieri}
\affiliation{%
	\institution{University of Sussex}
	\country{UK}
	}
\email{g.guerrieri@sussex.ac.uk}

\author{Claudio {Sacerdoti Coen}}
\affiliation{%
	\institution{
		Università di Bologna}
	\country{Italy}
}
\email{claudio.sacerdoticoen@unibo.it}
\renewcommand{\shortauthors}{Accattoli, Belo Lourenco, Ghica, Guerrieri, Sacerdoti Coen}

\begin{abstract}
Closure conversion is a program transformation at work in compilers for functional languages to turn inner functions into global ones, by building \emph{closures} pairing the transformed functions with the \emph{environment} of their free variables. Abstract machines rely on similar and yet different concepts of \emph{closures} and \emph{environments}. 

We study the relationship between the two approaches. We adopt a simple $\lambda$-calculus with tuples as source language and study abstract machines for both the source language and the target of closure conversion. Moreover, we focus on the simple case of flat closures/environments 
(no sharing of environments).
We provide three contributions. 
Firstly, a new simple proof technique for the correctness of closure conversion, inspired by abstract machines. 
Secondly, we show how the closure invariants of the target language allow us to design a new way of handling environments in abstract machines, not suffering the shortcomings of~other~styles. 

Thirdly, we study the machines from the point of view of time complexity. 
We show that closure conversion decreases various dynamic costs while increasing the size of the initial code. Despite these changes, the overall complexity of the machines before and after closure conversion \mbox{turns out to be the same}. 
\end{abstract}

\settopmatter{printfolios=true}



\maketitle

\section{Introduction}
\label{sect:intro}
A feature of functional languages, as well as of the $\l$-calculi on which they are based, is the possibility of defining \emph{inner functions}, that is, functions defined inside the definition of other functions. The tricky point is that inner functions can also use the external variables of their enveloping functions. For instance, the Church numeral  $\underline{3} \defeq \la\var\la\vartwo\var(\var(\var\vartwo))$ contains the inner function $\la\vartwo\var(\var(\var\vartwo))$ that uses the externally defined variable $\var$.

The result of functional programs can be a function. In particular, it can be the instantiation of an inner function. For instance, the result of applying $\underline{3}$ above to a value $\val$ is the instantiated function $\la\vartwo \val(\val(\val\vartwo)$, as one can easily see by doing one $\beta$-reduction step.

\paragraph{Closure Conversion} In practice, however, compiled functional programs do not follow $\beta$-reduction literally, nor do they produce the code of the instantiated function itself, because the potential duplications of the substitution process would be too costly. 
The idea is to decompose $\beta$-reduction in smaller, micro steps, delaying substitution as much as possible, and computing  representations of instantiated functions called \emph{closures}. A closure, roughly, is the  pair of the defined inner function (that is, before instantiation) plus the tuple of instantiations for its free variables, called its \emph{environment}.

This is achieved via a program transformation called \emph{closure conversion}, that re-structures the code  turning beforehand all inner functions into closures. At compile time, closures pair inner functions with \emph{initial environments} that simply contain the free variables of the functions. 
Execution shall then dynamically fill up the environments with the actual instantiations.


\paragraph{Compilation and Abstract Machines}  Abstract machines are often seen as a technique 
alternative to compilation and related to the interpretation of programming languages. This is because abstract machines tend to be developed for the source language, before the pipeline of transformations and optimizations of the compiler. 

A \emph{first aim} of our paper is to take a step toward closing the gap between compilation and abstract machines, by studying how transformations used by compilers---
here closure conversion---induce invariants 
exploitable for the design of machines working at further stages of the compilation pipeline, while~still~being~abstract.

\paragraph{Same Terminology, Different Concepts} 
Similarly to compilation, abstract machines do not follow $\beta$-reduction literally, nor do they produce the code of instantiated functions. In particular, 
some abstract machines use data structures called again 
\emph{closures} and \emph{environments}. These notions, however, are similar and yet \emph{different}. 

In closure conversion, an inner function with free variables is transformed at compile time into a closed function paired with the environment of its free variables. 
The environment is in general \emph{open} (it is closed by some enveloping converted function, yet locally it is open), and shall be filled/closed only during execution, 
which provides the instantiations. 
The aim is to hoist the closed function up to global scope 
while its environment stays at the call site.

In abstract machines, every piece of code receives an environment, not just functions, and environments and closures have a mutually recursive structure. Additionally, environments are built dynamically, during execution, not in advance, and they are always closed: the pair of a piece of code and its environment is called a \emph{closure} because they read back (or decode) to a closed $\l$-term. 
Moreover, there is no hoisting of code.

A \emph{second aim} of this paper is to establish and clarify the relationship between these two approaches. To disambiguate, we keep \emph{\wrapts} for the pairs \emph{(closed function, possibly open environment)} of closure conversion and use \emph{\mclosures} for the pairs \emph{(code, closed environment)} used in abstract machines; we further disambiguate calling \emph{\bagts} the environments of \wrapts, 
while the environments of \mclosures 
keep their usual names.
With respect to this terminology, we study \emph{\wrapt conversion for \mclosures}. One of the outcomes of \wrapt conversion turns out to be the elimination of \mclosures.

\paragraph{Complexity of Abstract Machines} Finally, a last influence on our study comes from the recent development by Accattoli and co-authors of both a complexity-based theory of abstract machines  \cite{DBLP:conf/icfp/AccattoliBM14,DBLP:conf/aplas/AccattoliBM15,DBLP:conf/ppdp/AccattoliB17,DBLP:journals/scp/AccattoliG19,DBLP:conf/ppdp/CondoluciAC19,DBLP:conf/ppdp/AccattoliCGC19,DBLP:journals/pacmpl/AccattoliLV21} 
and time and space reasonable cost models for the $\l$-calculus \cite{DBLP:journals/corr/AccattoliL16,DBLP:conf/lics/AccattoliCC21,DBLP:conf/lics/AccattoliLV22}. 
Such a line of work has developed fine analyses of tools and techniques for abstract machines, studying how different kinds of environments and forms of sharing impact on the cost of execution of abstract machines. 
The \emph{third aim} of the paper is to understand how tuples and \wrapt conversions affect environments, forms of sharing, and the cost of execution.

To avoid misunderstandings, our aim is \emph{not} the study of optimized/shared notions of (compiler) closures, their efficient representations, their minimization, or the trade-off between access time and allocation time---in fact, we adopt the basic form of \emph{flat} closure conversion that we apply to \emph{all} functions for \emph{all} their free variables. 

\paragraph{Flat Environments} 
Different data structures for (m-)closures and bags/environments and various closure conversion algorithms can be used. The design space, in particular, is due to chains of nested functions---say, $f$ is nested inside $g$, in turn, nested inside $h$---where two consecutive nested functions, $f$ and $g$, can both use variables, say $\var$, of $h$. Therefore, one might want the closures for $f$ and $g$ to share the environment entry for $\var$. 
Perhaps surprisingly, sharing bags between closures can easily break \emph{safety for space} of closure conversion, \ie might not preserve the space required by the source program, if parts of shared environments survive the lifespan of the associated closures. 
The simplest approach is using \emph{flat bags} (and flat closure conversion), where no sharing between bags is used; with flat bags, there are two distinct but identical entries for $\var$ in the bags of $f$ and $g$. 
Flat bags are safe for space \cite{DBLP:conf/lfp/ShaoA94,DBLP:journals/toplas/ShaoA00,DBLP:journals/pacmpl/Paraskevopoulou19}. 

Machine environments can also either be shared or flat (i.e. with no sharing between environments). The distinction, however, is not often made since it does not show up in the abstract specification of the machine, but only when one concretely implements it or studies the complexity of the overhead of the machine (since the two techniques have different costs). For a meaningful comparison with flat closure conversion, for source programs we adopt an abstract machine meant to be implemented using flat environments.

Our aim is to understand how closure conversion compares \emph{asymptotically} with \mclosures of abstract machines, in particular with respect to one of the key parameters for time analyses, the size of the initial term. 
Flat environments are chosen for their simplicity and their safeness for space.
More elaborated forms are future work. 

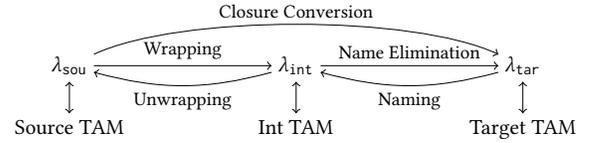
\begin{figure}[t!]
\small
\centering
	\scalebox{0.9}{\begin{tikzpicture}[ocenter]
		\node at (0,0)[align = center](soucal){\normalsize $\soucal$};
		\node at (soucal.center)[right = 85pt](intcal){\normalsize $\intcal$};
		\node at (intcal.center)[right =  85pt](tarcal){\normalsize $\tarcal$};
		
		\node at (soucal.center)[below = 20pt](soutam){\normalsize \TAM};
		\node at (intcal.center)[below = 20pt](inttam){\normalsize \LTAM};
		\node at (tarcal.center)[below = 20pt](tartam){\normalsize \TTAM};
		
		\draw[->](soucal) to node[above] {\small \Lifting} (intcal);
		\draw[->, out =-165, in =-15](intcal) to node[below] {\small \Collapsing} (soucal);
		\draw[->](intcal) to  node[above] {\small Name Elimination} (tarcal);
		\draw[->, out =-165, in =-15](tarcal) to node[below] {\small Naming} (intcal);
		\draw[->, out=25, in=155,looseness=0.6](soucal) to  node[above] {\small \Wrapt Conversion} (tarcal);
		
		\draw[<->](soucal) to (soutam);
		\draw[<->](intcal) to (inttam);
		\draw[<->](tarcal) to (tartam);
	\end{tikzpicture}}
	
	\vspace*{-.5\baselineskip}
\caption{$\l$-Calculi and abstract machines in the paper.}
\label{fig:setting}

\vspace*{-.5\baselineskip}
\end{figure}
\paragraph{Our Setting} 
We 
decompose \wrapt conversion in two phases, 
first going from a source calculus $\soucal$ to an intermediate one $\intcal$, and then to a target calculus $\tarcal$. 
We provide each of the three calculi with its own abstract machine. 
Our setting is summed up in \reffig{setting}, where \emph{TAM} stands for \emph{tupled abstract machine}. 

Our source $\l$-calculus $\soucal$ is the simplest possible setting accommodating \wrapt conversion, that is, Plotkin's effect-free call-by-value $\l$-calculus extended with tuples, because tuples are needed to define the conversion. Our setting is untyped, because it is how abstract machines are usually studied, and for the sake of minimality.

The first transformation $\lali\cdot: \soucal \to \intcal$, dubbed \emph{\lifting}, replaces abstractions $\la\tuvar\tm$ (where $\tuvar$ is a sequence of variables) with the 
dedicated construct of \emph{abstract \wrapts} $\lali{\la\tuvar\tm} \defeq \pack{\la{\tuvartwo}\la\tuvar\lali\tm}{\vv\vartwo}$ where $\tuvartwo$ is the sequence of free variables of $\la\tuvar\tm$, and $\vv\vartwo$ is the tuple of these same variables, forming the \emph{initial \bagt} of the \wrapt.

Roughly, the second transformation $\intcal \to \tarcal$, deemed \emph{name elimination}, turns the abstraction $\l\tuvartwo$ of many variables into the abstraction of a single variable representing the \bagt/environment. 
It is similar to a translation from named variables to de Bruijn indices. In fact, name elimination is only outlined in the paper, the details are given in the Appendix, because they are mostly routine.

The target calculus $\tarcal$ 
is not the low-level target language of a compiler. It is indeed still 
high-level, because we study only flat closure conversion, not the whole compilation pipeline. We do not model the hoisting of closed functions up to global scope; it is 
an easy aspect of closure conversion and is usually avoided~in~its~study.

We give three contributions, 
shaping the paper into three parts.

\paragraph{Contribution 1: A Simple Proof of Correctness} The correctness of \wrapt conversion is often showed by endowing both the source and the target calculus of the transformation with big-step operational semantics and establishing a logical relation between the two \cite{DBLP:conf/popl/MinamideMH96,DBLP:conf/pepm/SullivanDA21,DBLP:journals/pacmpl/Paraskevopoulou19,DBLP:conf/icfp/AhmedB08}. The reason is that adopting a small-step semantics does not seem to work: \wrapt conversion does not commute with meta-level substitution (that is, the substitution of converted terms is not the conversion of the substituted terms), and thus it does not map $\beta$-steps from the source to the target calculus. 

\citet{DBLP:conf/pldi/BowmanA18} are to our knowledge the only ones adopting a small-step semantics. They neatly get around the non-commutation issue by noticing that the substitution of converted terms is \emph{$\eta$-equivalent to} the conversion of the substituted terms. Our proof technique looks at the same issue in a different way, inspired by correctness proofs for abstract machines and independently of $\eta$-equivalence (which we do not consider for our calculi, for the sake of minimality). It might be of interest for languages where $\eta$-equivalence is not sound (\eg because of observable effects, as in OCaml, or when $\eta$-equivalence does not preserve typability, like in languages with mutability and the value restriction), since in these cases Bowman and Ahmed's proof might not scale up.


\paragraph{Contribution 2: New Kinds of Machine Environments} According to \citet{DBLP:journals/entcs/FernandezS09}, there are two kinds of abstract machines, those using many (shared or flat) \emph{local environments}, which are defined by mutual induction with \mclosures, 
and those using a single (necessarily flat) \emph{global environment} or \emph{heap} and no \mclosures. 
Each kind has pros and cons, there is no absolute better style of machine environments; see also 
\citet{DBLP:conf/ppdp/AccattoliB17}.

For our machine for the source calculus $\soucal$, the \TAM, we adopt flat local environments and \mclosures, as to allow us to compare closures (with flat bags) and \mclosures. 
The contribution here is that the invariants enforced by flat closure conversion (more precisely, by wrapping) 
enable a new management of environments, what we dub \emph{stackable environments} and plug into the \LTAM, our machine for the intermediate calculus $\intcal$. Stackable environments have the pros of \emph{both} global and local environments, and \emph{none} of their cons, as explained in \refsect{LTAM}. They are called \emph{stackable} because the current one can be put on hold---on the stack---when entering a closed function with its new environment, and re-activated when the evaluation of the function is over. But be careful: their stackability is not necessarily an advantage, it is just the way they work; the advantage is the lack of the cons of local and global environments.

Moving to the target calculus $\tarcal$ enables a further tweak of environments, adopted by the \TTAM: environments---which usually are \emph{maps} associating variables to values---become \emph{tuples} of values, with no association to variables. This is enabled by \emph{name elimination}, which turns variables into indices referring to the tuple, in a way reminiscent of de Bruijn indices. 

\paragraph{Contribution 3: Time Complexity} 
We study how tuples and flat \wrapt conversion impact the forms of sharing and the time complexity of the machines. Our analyses produce four insights:
\begin{enumerate}
\item \emph{Tuples raise the overhead}: we give a theoretical analysis of the cost of adding tuples to the pure $\l$-calculus, which, to our knowledge, does not appear anywhere in the literature. We show why tuples require their own form of sharing and that the \emph{creation of tuples} at runtime is \emph{unavoidable}. This is done by adapting \emph{size exploding families} from the study of reasonable time cost models for the $\l$-calculus \cite{DBLP:conf/rta/Accattoli19}. Moreover, tuples raise the dependency on the size $\size\tm$ of the initial code $\tm$  of the overhead of the machine. Namely, let the \emph{height} $\hg\tm$ be the maximum number of bound variables of $\tm$ in the scope of which a sub-term of $\tm$ is contained: the dependency for flat environments raises from $\bigo(\hg\tm)$ to $\bigo(\size\tm\cdot\hg\tm)$.

\item \emph{Name elimination brings a logarithmic speed-up}: with variable names, flat environments have at best $\bigo(\log(\hg\tm))$ access time, while de Bruijn indices (or our name elimination) enable $\bigo(1)$ access time---this is true both before and after closure conversion. Before conversion, however, the improvement does not lower the overall asymptotic overhead of the machine with respect to the size of the initial term, which is dominated by the other flat  environment operations. After closure conversion, instead, it \emph{does} lower the overall dependency of the machine from $\bigo(\size{\cc\tm} \cdot\hg{\cc\tm})$ to $\bigo(\size{\cc\tm})$, where $\cc\tm$ is $\tm$ after closure conversion.

\item \emph{Amortized constant cost of transitions}: 
in any abstract machine, independently of their implementation,
the number of transitions of an execution and the cost of some single transitions depend on  the size of the initial term. 
This is related to the higher-order nature of $\l$-calculi. 
Closure conversion 
impacts on the cost of single transitions (but not on their number): their \emph{amortized} cost becomes constant. 
The insight is that the non-constant cost of transitions in ordinary abstract machines is related to inner functions.

\item \emph{Dynamically faster, statically bigger, overall the same}: the previous two points show that closure conversion decreases the dependency of machines on the size of the initial term during execution. The dynamic improvement however is counter-balanced by the fact that $\size{\cc\tm}$ is possibly \emph{bigger} than 
$\size{\tm}$, namely $\size{\cc\tm}\in \bigo(\size\tm \cdot \hg\tm)$. 
Therefore, the overall complexity is $\bigo(\size\tm\cdot\hg\tm)$ also after closure conversion.
\end{enumerate}

\paragraph{OCaml Code and Proofs.} As additional material on GitHub \cite{PPDP25ocaml}, we provide an OCaml implementation of the \TTAM, the machine for closure converted terms, also described in {\Cref{sect:app-implementation}}.

All proofs are in the Appendix, 
which will be uploaded on ArXiv.

\section{Preliminaries: 
$\cbvcal$, a Call-by-Value $\l$-Calculus}
\label{sect:plotkin}
\begin{figure}[t!]
\centering
\setlength{\arraycolsep}{3pt}
\begin{tabular}{c|cc}
		$\begin{array}{r rlllll}
			\textsc{Terms} &
			\tm,\tmtwo,\tmthree,\tmfour& \grameq & \var  \midd \la\var\tm \midd \tm\tmtwo
			\\
			\textsc{Values} &
			\val,\valtwo &  \grameq  & \la\var\tm
			\\
			\textsc{Ev. Ctxs} &
			\evctx,\evctxtwo &  \grameq  & 
			\ctxhole \midd \tm \evctx \midd \evctx \val
		\end{array}$
	&
	\makecell{$
		(\la{\var}\tm)\val  \rtobv  
		\tm\isub{\var}{\val}
	$
	\\[.4em]
	\AxiomC{$\tm \rootRew{\betav} \tmtwo$}
	\UnaryInfC{$\evctxp{\tm} \Rew{\betav} \evctxp{\tmtwo}$}
	\DisplayProof}
	
\end{tabular}

\vspace*{-.5\baselineskip}
\caption{The untyped 
	pure call-by-value calculus $\cbvcal$} 
\label{fig:cbv_calculus}

\vspace*{-.5\baselineskip}
\end{figure}
%
In \reffig{cbv_calculus} we present $\cbvcal$, a variant of Plotkin's call-by-value $\l$-calculus \cite{DBLP:journals/tcs/Plotkin75} with its $\beta$-reduction by value $\tobv$, 
adopting two specific choices. Firstly, the only values are $\l$-abstractions. Excluding variables from values differs from \cite{DBLP:journals/tcs/Plotkin75} but is common in the machine-oriented literature. It does not change the result of evaluation while inducing a faster substitution process, see \cite{DBLP:journals/iandc/AccattoliC17}. 

Secondly, we adopt a small-step operational semantics, defined via evaluation contexts. 
Evaluation contexts $\evctx$ are special terms with exactly one occurrence of the \emph{hole} constant $\ctxhole$. 
We write $\evctxp\tm$ for the term obtained from the evaluation context $\evctx$ by replacing its hole with the term $\tm$ (possibly capturing some free variables~of~$\tm$).

The small-step rule of $\betav$-reduction is \emph{weak}, that is, it does not evaluate abstraction bodies (indeed, the production $\la\var\evctx$ is \emph{absent} in the definition of {evaluation context} $\evctx$ in \reffig{cbv_calculus}), as it is common in functional programming languages, and \emph{deterministic}, namely proceeding from right to left (as forced by the production $\evctx \val$)\footnotemark. 
\footnotetext{
	The right-to-left order (adopted also in \cite{Leroy-ZINC,DBLP:conf/icfp/GregoireL02}) 
	induces a more natural presentation of the machines, but all 
	our results could be restated using the left-to-right order.}

We identify terms up to $\alpha$-renaming; $\tm\isub{\var}{\tmtwo}$ stands for metalevel capture-avoiding substitution of $\tmtwo$ for the free occurrences~of~$\var$~in~$\tm$.

The lemma below rests on the closed hypothesis and will be used as a  design check for next sections' calculi. 
It is an untyped instantiation of \citeauthor{WrightFelleisen94}'s uniform evaluation property~\cite{WrightFelleisen94}.

\begin{lemma}[$\cbvcal$ harmony]
\label{l:cbv-harmony}
If $\tm\in\cbvcal$ is closed, then either $\tm$ is a value or $\tm \tobv \tmtwo$ for some closed $\tmtwo\in\cbvcal$.
\end{lemma}

\paragraph{Notations} We set some notations for both calculi and machines. Let $\to$ be a reduction relation. An evaluation sequence $\deriv: \tm \to^* \tmtwo$ is a possibly empty sequence of $\to$-steps the length of which is noted 
$\size\deriv$. If $a$ and $b$ are sub-reduction (i.e., $\Rew{a}\subseteq \to$ and 
$\Rew{b}\subseteq \to$) then $\Rew{a,b} \defeq \Rew{a}\cup \Rew{b}$ and $\sizep\deriv a$ 
is the number of $a$ steps in $\deriv$.
\begin{figure}[t!]
\centering
\scalebox{0.93}{
\setlength{\tabcolsep}{2pt}
\setlength{\arraycolsep}{2pt}
\!\!\!\!\!\!\!\!\!\!\!\begin{tabular}{c|c}
	\multicolumn{2}{c}{
	$\begin{array}{r@{\hspace{.2cm}} rlllll}
		\textsc{Terms} &
		\tm,\tmtwo,\tmthree&\grameq& \var \midd \tm\tmtwo \midd \proj_i \tm \midd \overgroup{\la{\var_1,\mydots,\var_n}\tm}^{\la{\tuv\var}\tm} \ \ n\geq 0 \midd \overgroup{\tuple{\tm_1,\mydots, \tm_n}}^{\vv\tm}  \ \ n\geq 0
		\\
		\textsc{Values} &
		\val,\valtwo & \grameq  & \la{\tuv\var}\tm \midd \vv\val
		\\
		\textsc{Ev. Ctxs} &
		\evctx,\evctxtwo & \grameq  & 
		\ctxhole \midd \tm \evctx \midd \evctx \val \midd \proj_{i} \evctx \midd \undergroup{ \tuple{\tm_{1},\mydots,\tm_{k},\evctx,\val_{1},\mydots,\val_{h}}}_{\tuple{\tuv\tm, \evctx, \tuv\val}} \ \ k,h \geq 0
	\end{array}$
	}
\\[3em]
\hline
	$\begin{array}{r l l l}
		(\la{\tuv\var}\tm)\vv\val & \rtobv & 
		\tm\isub{\tuv\var}{\vv\val} & \mbox{if }\norm{\tuv\var} = \norm{\vv\val}
		\\
		\proj_{i} \vv\val
		& \rtoproj &
		\val_{i} & \mbox{if } 1 \leq i \leq \norm{\vv\val}
	\end{array}$
&
	\begin{tabular}{ccc}
		\AxiomC{$\tm \rootRew{a} \tmtwo$}
		\UnaryInfC{\!\!$\evctxp{\tm} \Rew{a} \evctxp{\tmtwo}$\!\!}
		\DisplayProof
		&\!$a \!\in\! \{ \betav, \projsym\}$
		\\[6pt]
		\multicolumn{2}{l}{$\tosou \ \defeq \ \tobv \cup \toproj$}
	\end{tabular}
\end{tabular}
}

\vspace*{-.5\baselineskip}
\caption{The source calculus $\soucal$ extending $\cbvcal$ with tuples.}
\label{fig:source_calculus}

\vspace*{-\baselineskip}
\end{figure}

\section{Part 1: The Source Calculus \texorpdfstring{$\soucal$}{}} 
\label{sect:source-calculus}
In this section, we extend 
$\cbvcal$ with tuples, which are needed to define closure conversion, obtaining $\soucal$, our source calculus. 

\paragraph{Terms} The \emph{source calculus} $\soucal$ defined in \reffig{source_calculus} adopts 
$n$-ary \emph{tuples} $\vv\tm = \tuple{\tm_1,\mydots, \tm_n}$, together with projections $\proj_i$ on the $i^\text{th}$ element\footnotemark
\footnotetext{
We do not define projections as tuple-unpacking abstractions because it would turn $\proj$-steps into $\beta$-steps and so blur the cost analysis (that counts $\beta$ but not $\proj$-steps).
}.
Abstractions are now on \emph{sequences} of variables $\tuvar = \var_1, \mydots, \var_n$.
With a slight abuse, we also compact $\tuple{\tm_{1},\mydots,\tm_{n}}$ into $\tuple{\,\tuv\tm\,}$, and write $\tuple{\,\tuv\tm,\tmtwo,\tuv\tmthree\,}$ (note that replacing sequences with tuples changes the meaning: $\tuple{\,\tuv\tm,\tmtwo,\tuv\tmthree\,}$ and $\tuple{\,\vv\tm,\tmtwo,\vv\tmthree\,}$ are different terms, and we shall need both notations). Both tuples and sequences of variables 
can be empty, that is, $\la{}\tm$ and $\tuple{}$ are terms of $\soucal$. 
\emph{Values} now are abstractions \emph{and} tuples \emph{of values}---tuples of arbitrary terms are not values in general. 

\paragraph{Notations and Conventions about Tuples and Sequences.} We assume that in every sequence $\tuvar$ all elements are distinct and, for brevity, we abuse notations and consider sequences of variables also as the sets of their elements, writing $\var_i\in \tuvar$, or $\fv\tm=\tuvar$, or $\vvvar\cup\vvvartwo$. 
We set $\norm{\vv\tm} \defeq n$ if $\vv\tm = \tuple{\tm_{1},\mydots,\tm_{n}}$ and call it the \emph{length} of the tuple $\vv\tm$ 
(so, $\norm{\tuple{}} = 0$), and similarly for $\norm\tuvar$.
Moreover, if $\tuvar = \var_{1},\mydots,\var_{n}$ and $\vv\val = \tuple{\val_{1},\mydots,\val_{n}}$ we then set $\tm\isub{\tuvar}{\vv\val} \defeq 
\tm \isubsim{\var_{1}}{\val_{1}}{\var_{n}}{\val_{n}}$ (the \emph{simultaneous} substitution). 
We write $\tuvar\vdisjoint \tuvartwo$ when $\tuvar$ and $\tuvartwo$ have no element in common, and $\fv\tm \subseteq \tuvar\vdisjoint \tuvartwo$ when moreover $\fv\tm \subseteq (\tuvar\cup\tuvartwo)$.

\paragraph{Small-Step Operational Semantics.} The $\betav$-rule can fire only if the argument is a tuple of values of the right length, and similarly for the 
$\proj$-rule. 
For instance, $(\la\var\var) (\la\vartwo\vartwo\vartwo) \not \tobv \la\vartwo\vartwo\vartwo$, because one needs a unary tuple around the argument, that is, $(\la\var\var) \tuple{\la\vartwo\vartwo\vartwo}  \tobv \la\vartwo\vartwo\vartwo$. Note that evaluation contexts now enter projections and tuples, proceeding right-to-left. As it shall be the case for all calculi in this paper, the operational semantics $\tosou$ of $\soucal$ is \emph{deterministic}: 
if $\tm \tosou \tmtwo$ and $\tm \tosou \tmthree$ then $\tmtwo = \tmthree$ and $\tm$ is not a value. The proof is a routine induction.

\paragraph{Clashes} 
In an untyped setting, there might be terms with \emph{clashes}, that is, irreducible badly formed configurations such as $\proj_i(\la\var\tm)$. 
To exclude clashes without having to have types, we adopt 
a notion of clash-freeness, which would be ensured by any type system.

\begin{definition}[Clashes, clash-free terms]
A term $\tm$ is a \emph{clash} if it has shape $\evctxp{\tmtwo}$ where $\tmtwo$ has one of the following forms
:
\begin{itemize}
\item \emph{Clashing projection}: $\tmtwo = \proj_{i} \val$ and 
if $\val \!=\!\! \vv{\!\!\val\!\!}$ then $\norm{\vv{\val}} <i$;
\item \emph{Clashing abstraction}:  $\tmtwo = (\la{\tuvar}\tmthree) \val$ and 
if $\val \!=\!\! \vv{\!\!\val\!\!}$ \mbox{then $\norm{\tuvar} \!\neq\! \norm{\vv{\!\val\!}}$;}
\item \emph{Clashing tuple}: $\tmtwo = \vv{\tmfour} \tmthree$.
\end{itemize}
A term $\tm$ is \emph{clash-free} when, co-inductively, $\tm$ is not a clash and if $\tm \tosou \tmtwo$ then $\tmtwo$ is clash-free.
\end{definition}
Note that clashes are normal forms. All the calculi and machines of the paper shall come with their notion of clash and clash-freeness, which shall be taken into account in statements and proofs but the definitions of which shall be omitted (they are in the Appendix).

\begin{toappendix}
\begin{lemma}[$\soucal$ harmony]
\label{l:source-harmony}
If $\tm \in \soucal$ is closed and clash-free, then either $\tm$ is a value or $\tm \tosou \tmtwo$ for some closed and clash-free $\tmtwo$.
\end{lemma}
\end{toappendix}

 \section{Part 1: the Intermediate Calculus \texorpdfstring{$\intcal$}{} and the \Lifting Transformation}
 \label{sect:intermediate-calculus}
\begin{figure}[t!]
\centering
\scalebox{0.92}{
\!\!\!\begin{tabular}{c|cc}
	\multicolumn{2}{c}{ \setlength{\arraycolsep}{2pt}
	$\begin{array}{r@{\hspace{.35cm}} rll@{\hspace{.5cm}} r@{\hspace{.35cm}} rllll}
		\textsc{\Bagts} &
		\bag,\bagtwo&\grameq& \vv\var \midd \vv\val
		\\
		\textsc{Terms} &
		\tm,\tmtwo,\tmthree,\tmfour&\grameq& \var \midd \tm\tmtwo \midd \proj_i \tm \midd \vv\tm \midd \pack{\dabsv\vartwo\var\tm}{\bag}  
		\\
		\textsc{Values} &
		\val,\valtwo & \grameq  &\pack{\dabsv\vartwo\var\tm}{\bag} \midd \vv\val
		\\
		\textsc{Eval ctxs} &
		\evctx,\evctxtwo & \grameq  & 
		\ctxhole \midd \tm \evctx \midd \evctx \val \midd \proj_{i} \evctx \midd \tuple{\tuv\tm, \evctx, \tuv\val}
	\end{array}$
	}
\\[20pt]
	\multicolumn{2}{l}{
	$\begin{array}{r l l@{\hspace{.3cm}} l}
		\pack{ \dabsv\vartwo\var\tm }{ \vv{\val_1} }\, \vv{\val_2} & \!\!\!\!\rtoibv\!\!\!\! & 
		\tm\isubp{\tuvartwo}{\vv{\val_1}}{\tuvar}{\vv{\val_2}}
		& \mbox{if } \norm{\tuvartwo} = \norm{\vv{\val_1}} \mbox{ and }\norm{\tuvar} = \norm{\vv{\val_2}}
		\\
		\proj_{i} \vv\val
		& \!\!\!\!\rtoiproj\!\!\!\! &
		\val_{i} & \mbox{if } 1 \leq i \leq \norm{\vv\val}
	\end{array}$
	}
\\[10pt]
		\AxiomC{$\tm \rootRew{a} \tmtwo$}
		\UnaryInfC{$\evctxp{\tm} \Rew{a} \evctxp{\tmtwo}$}
		\DisplayProof
\quad
		$a \in \{ \intprefix\betav, \intprefix\projsym\}$
&
	$
		\tollcbv \ \defeq \ \tollbv \cup \tollproj
$
\end{tabular}
}

\vspace*{-.5\baselineskip}
\caption{The intermediate calculus $\intcal$.}
\label{fig:intermediate_calculus}

\vspace*{-.6\baselineskip}
\end{figure}
 In this section, we define the intermediate calculus $\intcal$ and the 
 \lifting translation from $\soucal$ to $\intcal$.
 We discuss why the natural first attempt to show the correctness of the translation does not work, and solve the issue via a reverse translation from $\intcal$~to~$\soucal$.
 
\paragraph{Terms of $\intcal$.} In $\intcal$, defined in \reffig{intermediate_calculus}, abstractions $\la\tuvar\tm$ are replaced by \emph{\wrapts} $\pack{\la{\tuvartwo}\la{\tuvar}\tm}{\bag}$, which are compactly noted $\pack{\dabsv\vartwo\var\tm}{\bag}$. The \emph{\bagt} $\bag$ of a \wrapt can be of two forms $\vv\varthree$ and $\vv\val$, giving \emph{variable} 
\emph{\wrapts} $\pack{\dabsv\vartwo\var\tm}{\vv\varthree}$ and \emph{evaluated \wrapts} $\pack{\dabsv\vartwo\var\tm}{\vv\val}$, which are both values. 
In a \wrapt $\pack{\dabsv{\vartwo}{\var}\tm}{\bag}$, $\vartwolist$ and $\varlist$ verify $\tuvartwo\vdisjoint\tuvar$ (\ie no elements in common), and 
scope over the \emph{body} $\tm$ of the \wrapt; 
$\vartwolist$ and $\varlist$ do not scope over $\bag$. 
The idea is that  $\fv{\la\varlist\tm}\subseteq \tuvartwo$, so that $\la\tuvartwo\la\tuvar\tm$ is closed. The elements of $\bag$ are meant to replace the variables $\tuvartwo$ in $\tm$.

The rationale behind $\intcal$ is understood by looking at the translation from $\soucal$ to $\intcal$ in \reffig{sou-int_translations}. 
Basically, every abstraction is closed and paired with the \bagt of its free variables. 
Evaluated \wrapts $\pack{\dabsv\vartwo\var\tm}{\vv\val}$ are not in the image of the translation, unless $\fv{\la\tuvar\tm}=\emptyset$, which gives the \wrapt $\pack{;\tuvar.\tm}{\tuple{}}$ that is both a variable and an evaluated \wrapt.
Evaluated \wrapts are generated by the reduction rules, discussed after 
defining well-formedness.

\begin{definition}
A \wrapt $\pack{\dabsv\vartwo\var\tm}{\bag}$ is \emph{well-formed} if $\fv\tm\subseteq \tuvartwo\vdisjoint\tuvar$, $\norm{\bag} = \norm{\tuvartwo}$ and if $\bag$ is a variable \bagt then 
$\bag = \tuple{\tuvartwo}$.
Terms $\tm\in\intcal$ and evaluation contexts $\evctx\in\intcal$ are \emph{well-formed} if all their \wrapts are well-formed, and \emph{prime} if moreover they are variable \wrapts. 
\end{definition}

\paragraph{Operational Semantics} The intermediate variant $\toibv$ of the $\betav$-rule involves a well-formed evaluated \wrapt $\pack{ \dabsv\vartwo\var\tm }{ \vv{\val_1} }$ and an argument $\vv{\val_2}$ of the same length of $\tuvar$, and amounts to substitute the \bagt $\vv{\val_1}$ on $\tuvartwo$ and the argument $\vv{\val_2}$ on $\tuvar$. 
Substitution $\tm\isubp{\tuvartwo}{\vv{\val_1}}{\tuvar}{\vv{\val_2}}$ is as defined expected (see \Cref{sect:app-intermediate-calculus}) by performing $\isub{\tuvartwo}{\vv{\val_1}}$ and $\isub{\tuvar}{\vv{\val_2}}$ \emph{simultaneously} and requires that $\fv\tm\subseteq\tuvartwo\vdisjoint\tuvar$.
Well-formed terms are stable by substitution, when  defined, and the reduct of $\toibv$ is a (well-formed) term of $\intcal$ because (by well-formedness) \wrapts close their bodies, so the closing  substitution generated by the step  turns variable \bagts into value \bagts. Tuple projection is as in $\soucal$. 
Note that evaluation contexts do not enter \wrapts. 
See \Cref{sect:app-intermediate-calculus} for the definition of clash(-freeness) for $\intcal$. We say that $\tm\in\intcal$ is \emph{good} if it is well-formed and clash-free.

The intermediate calculus $\intcal$ is deterministic (if $\tm \toint \tmtwo$ and $\tm \toint \tmthree$ then $\tmtwo = \tmthree$ and $\tm$ is not a value) and harmonic.
	
\begin{toappendix}
\begin{lemma}[Harmony of $\intcal$]
\label{l:intermediate-harmony}
Let $\tm\in\intcal$ be closed and good. Then either $\tm$ is a value or $\tm \toint \tmtwo$ for some closed good $\tmtwo\in\intcal$. 
\end{lemma}
\end{toappendix}

\begin{figure}[t!]
\centering
\scalebox{0.93}{
\setlength{\arraycolsep}{2pt}
$\begin{array}{rll@{\hspace{.5cm}} rll}
		\multicolumn{6}{c}{\textsc{\textsc{\Lifting} translation } \soucal \to \intcal}
		\\[2pt]
		\lali\var  & \defeq  & \var
		&
		\lali{\tm\tmtwo}  & \defeq  & \lali\tm \,\lali\tmtwo
		\\
		\lali{\la{\tuvar}\tm}  & \defeq  & \pack{\dabsv\vartwo\var\lali\tm}{ \vv\vartwo}
&\multicolumn{3}{c}{\mbox{if }\fv{\la{\tuvar}\tm} =\tuvartwo}
		\\
		\lali{\proj_i\tm}  & \defeq  & \proj_i{\lali\tm }
		&
		\lali{\tuple{\tm_1,\mydots,\tm_n}}  & \defeq  & \tuple{\lali{\tm_1} ,\mydots,\lali{\tm_n} }
	
\\[5pt]\hline\\[-4pt]
		\multicolumn{6}{c}{\textsc{Reverse \collapsing translation }\intcal \to \soucal}
		\\[2pt]
		\unlali\var  & \defeq  & \var 
		&
		\unlali{\tm\tmtwo}  & \defeq  & \unlali\tm\,  \unlali\tmtwo
		\\[2pt]
		\unlali{\pack{\dabsv\vartwo\var\tm}{\vv\vartwo}} & \defeq  & \la{\tuvar}\unlali\tm
		&
		\unlali{\pack{\dabsv\vartwo\var\tm}{\vv\val}} & \defeq  & \la{\tuvar}\unlali\tm\isub{\tuvartwo}{\vv{\unlali\val}}
		\\[2pt]
		\unlali{\proj_i\tm}  & \defeq  & \proj_i{\unlali\tm }
		&
		\unlali{\tuple{\tm_1,\mydots,\tm_n}}  & \defeq  & \tuple{\unlali{\tm_1} ,\mydots,\unlali{\tm_n} }
	\end{array}$
}

\vspace*{-.5\baselineskip}
\caption{The translations $\lali{\cdot} \colon \soucal \to \intcal$ and \mbox{$\unlali{\cdot} \colon \intcal \to \soucal$}.}
\label{fig:sou-int_translations}

\vspace*{-.5\baselineskip}
\end{figure}
\paragraph{Translation From Source to Intermediate, or \Lifting.} The wrapping translation from $\soucal$ to $\intcal$ takes a (possibly open) term $\tm\in\soucal$ and returns a term $\lali\tm\in\intcal$, and it is defined in \reffig{sou-int_translations}. As already mentioned, it turns abstractions into \wrapts by closing them and pairing them with the \bagt of their free variables. It is extended to evaluation contexts $\evctx$ as expected, setting $\lali\ctxhole \defeq \ctxhole$.

\begin{toappendix}
\begin{lemma}[Properties of the translation $\soucal\to\intcal$]
\label{l:lali-properties} 
\hfill
\begin{enumerate}
\item 
\emph{Values}: if $\val\in\soucal$ then $\lali\val$ is a value of $\intcal$.
\item 
\emph{Terms}: if $\tm\in\soucal$ then $\lali\tm\in\intcal$ is well-formed and prime.
\item 
\emph{Contexts}: if $\evctx\in\soucal$ then $\lali\evctx$ is an evaluation context of $\intcal$.
\end{enumerate}
\end{lemma}
\end{toappendix}

\paragraph{Problem: Wrapping and Substitution Do not Commute.}
\label{l:lali-subs}
As source values and contexts are translated to their intermediate analogous, one may think that the translation preserves reduction steps: if $\tm \tosou \tmtwo$ then $\lali{\tm} \toint \lali{\tmtwo}$. 
But this is \emph{false}, because the translation does not commute with  substitution, as also discussed in \cite{DBLP:conf/pldi/BowmanA18,DBLP:conf/ppdp/SullivanDA23}.
%
Indeed, in general $\lali\tm \isub\var{\lali\val} \neq \lali{\tm\isub\var\val}$: take $\tm \defeq \la\vartwo\vartwo\var$ and a closed value $\mycyan\val$, we get the terms below, where $\pack{\dabs\var\vartwo\vartwo \var}{\tuple{\lali{\mycyan\val}}} \neq \pack{\dabs{}\vartwo\vartwo \lali{\mycyan\val}}{ \tuple{}}$.
\begin{center}
$\begin{array}{ccccc} 
\lali\tm \isub\var{\lali{\mycyan\val}}  & = &\pack{\dabs\var\vartwo\vartwo \var}{\tuple\var} \isub\var{\lali{\mycyan\val}} & =& \pack{\dabs\var\vartwo\vartwo \var}{\tuple{\lali{\mycyan\val}}}
\\
\lali{\tm\isub\var{\mycyan\val}} & = & \lali{\la\vartwo\vartwo\mycyan{\val}} & = & \pack{\dabs{}\vartwo\vartwo \lali{\mycyan\val}}{ \tuple{}}
\end{array}$
\end{center}
The point being that $\tm\isub\var{\mycyan\val}$ is \emph{closed} and so its \lifting is different from first \lifting $\tm$, which is instead \emph{open}, and then closing the \lifted term $\lali\tm$ using $\isub\var{\lali{\mycyan\val}}$.
Not only $\lali\tm \isub\var{\lali{\mycyan\val}} \neq \lali{\tm\isub\var{\mycyan\val}}$, they are not even related by $\toint$, or by the equational theory generated by $\toint$ (but they can be shown to be contextually equivalent).

The problem is that the translation targets prime terms (\Cref{l:lali-properties}.2) but $\toint$ creates evaluated \wrapts (which are not in the image of translation), i.e., \emph{the reduct of a prime term (of $\intcal$) may not be prime}. 
Consider again $\tm \defeq \la\vartwo\vartwo\var$: then $(\la\var\tm)\tuple{\val} \tosou \tm \isub\var{\mycyan\val} = \la\vartwo\vartwo\mycyan{\val}$ and $\lali{(\la\var\tm)\tuple{\val}} = \pack{\dabs{}\var\pack{\dabs\var\vartwo\vartwo \var}{\tuple\var}}{\tuple{}} \tuple{\lali{\val}}$ is prime, but $\lali{(\la\var\tm)\tuple{\val}}$ $\toint$-reduces to the 
non-prime $\pack{\dabs\var\vartwo\vartwo \var}{\tuple{\lali{\mycyan\val}}} = \lali\tm \isub\var{\lali{\mycyan\val}} \neq \lali{\tm \isub\var{\mycyan\val}}$.

\paragraph{Reverse Translation.} The literature usually overcomes this problem by switching to a different approach, adopting a big-step semantics and a logical relation proof technique. One of our contributions is to show a direct solution, as done also by \citet{DBLP:conf/pldi/BowmanA18}, but in a different way. The idea comes from the correctness of abstract machines, which is proved by projecting the machine on the calculus, rather than the calculus on the machine. Therefore, we define a reverse \emph{\collapsing} translation from $\intcal$ to $\soucal$ and show that it smoothly preserves reduction steps.

The \emph{\collapsing translation} $\unlali\tm$ of a well formed term $\tm\in\intcal$ to $\soucal$ is defined in \reffig{sou-int_translations}. It amounts to substitute the \bagt $\bag$ for the sequence $\tuvartwo$ of variables of a \wrapt $\pack{\dabsv\vartwo\var\tm}{\bag}$. In contrast with 
wrapping, now 
\collapsing and substitution commute.

\begin{toappendix}
\begin{proposition}[Commutation of substitution and the reverse translation]
\label{prop:intcal-key-sub-prop-rev-trans}
If $\tm,\vv	\valnone,\vv\valntwo\in\intcal$ and $\fv\tm\subseteq \tuvar\vdisjoint\tuvartwo$, then $\unlali{\tm\isubp{\tuvar}{\vv\valnone}{\tuvartwo}{\vv\valntwo}} = \unlali\tm\isub{\tuvar}{\unlali{\vv\valnone}}\isub{\tuvartwo}{\unlali{\vv\valntwo}}$.
\end{proposition}
\end{toappendix}

\paragraph{Strong Bisimulation} From the commutation property (\refpropeq{intcal-key-sub-prop-rev-trans}), it easily follows that the reverse translation projects and reflects rewrite steps and---if terms are closed---also normal forms. 
As a consequence, it is a termination-preserving strong bisimulation, possibly the strongest form of \emph{correctness} for a \mbox{program~transformation}.

\begin{toappendix}
\begin{theorem}[Source-intermediate termination-preserving strong bisimulation]
\label{thm:rev-transl-implem-bricks} 
Let $\tm \in \intcal$ be 
closed and well-formed.
\begin{enumerate}
\item 
\emph{Projection}: 
for $a \in \{\betav, \proj\}$, if $\tm\Rew{\intprefix a}\tmtwo$ then $\unlali\tm \Rew{a} \unlali \tmtwo$.
\item 
\emph{Halt}: $\tm$ is $\toint$-normal if and only if $\unlali\tm$ is  $\tosou$-normal.

\item 
\emph{Reflection}: 
for $a \in \{\betav, \proj\}$, if $\unlali\tm \Rew{a} \tmtwo$ then there exists $\tmthree \in \intcal$ such that $\tm \Rew{\intprefix a} \tmthree$ and $\unlali\tmthree=\tmtwo$.

\item 
\emph{Inverse}: if $\tmtwo$ is a source term then $\unlali{\lali\tmtwo} = \tmtwo$.

\end{enumerate}
\end{theorem}
\end{toappendix}

Projection can be extended to sequences of steps. 
The  reflection and inverse properties ensure that reflections of consecutive steps from the source can be composed, bypassing the problem with the wrapping translation, as shown by \ben{the} corollary below. 
Its proof only depends on the abstract properties of our notion of bisimulation. 

\begin{toappendix}
\begin{corollary}[Preservation of reduction steps]
	\label{cor:step-preservation}
\!\!If $\tm \tosou^k~\!\!\tmtwo$ then there exists $\tmthree \in \intcal$ such that $\lali\tm \toint^k \tmthree$ and $\unlali\tmthree = \tmtwo$.
\end{corollary}
\end{toappendix}


\begin{figure}[t!]
\centering
\scalebox{0.93}{
\begin{tabular}{ccc}
	\multicolumn{2}{c}{	\setlength{\arraycolsep}{2pt}
		$\begin{array}{r@{\hspace{.5cm}} rll@{\hspace{1.cm}} r@{\hspace{.5cm}} rll}
		\textsc{Projected vars} & \prvar,\prvartwo &\grameq& \proj_i \efvar \midd \proj_i \envar
		\\
		\textsc{Terms} &
		\tm,\tmtwo,\tmthree,\tmfour&\grameq& \prvar \midd \proj_i \tm \midd \vv\tm \midd \tm\,\tmtwo   \midd  \packnm\tm\bag 
		\\
		\textsc{\Bagts} &
		\bag,\bagtwo&\grameq& \vec\prvar \midd \vec\val
		\\
		\textsc{Values} &
		\val,\valtwo & \grameq  &\packnm\tm\bag \midd \vv\val
		\\
		\textsc{Eval Contexts} &
		\evctx,\evctxtwo & \grameq  & 
		\multicolumn{4}{l}{\ctxhole \midd \tm\, \evctx \midd \evctx\,\vv \val \midd \proj_{i} \evctx \midd \tuple{\tuv\tm, \evctx, \tuv\val} }
	\end{array}$
	}
\\[15pt]\hline\\[-10pt]
$\begin{array}{rll@{\hspace{.85cm}} rll@{\hspace{.85cm}} rlllllrllrlllll}
		\multicolumn{6}{c}{\textsc{Name elimination translation } \intcal \to \tarcal}
		\\[2pt]
		\clc{\vartwo_i}\vvvartwo\vvvar & \!\!\!\!\defeq\!\!\!\!  & \proj_i \lvar
		&
		\clc{\tuple{\tm_1,\mydots,\tm_n}}\vvvartwo\vvvar & \!\!\!\!\defeq\!\!\!\!  & \tuple{\clc{\tm_1}\vvvartwo\vvvar,\mydots,\clc{\tm_n}\vvvartwo\vvvar}
		\\
				\clc{\proj_i\tm}\vvvartwo\vvvar & \!\!\!\!\defeq\!\!\!\!  & \proj_i{\clc\tm\vvvartwo\vvvar}
		&
		\clc{\pack{\dabsv\varthree\varfour\tm}{\bag} }\vvvartwo\vvvar & \!\!\!\!\defeq\!\!\!\! & 
		\pack{\clc\tm{\tuvarthree}{\tuvarfour}}{ \clc\bag\vvvartwo\vvvar}
		\\[5pt]
		\clc{\var_i}\vvvartwo\vvvar & \!\!\!\!\defeq\!\!\!\!  & \proj_i \svar
		&
		\clc{\tm\tmtwo}\vvvartwo\vvvar & \!\!\!\!\defeq\!\!\!\!  & \clc\tm\vvvartwo\vvvar \,\clc\tmtwo\vvvartwo\vvvar
		
	\end{array}$

\end{tabular}
}

\vspace*{-.5\baselineskip}
\caption{
The target calculus $\tarcal$ and name elimination.}
\label{fig:target_calculus}

\vspace*{-.5\baselineskip}
\end{figure}

 \section{Part 1: Outline of the Target Calculus \texorpdfstring{$\tarcal$}{} and of Name Elimination}
 \label{sect:target-calculus}
  In this section, we quickly outline the \emph{target calculus} $\tarcal$ and the translation from $\intcal$ to $\tarcal$, dubbed \emph{name elimination}, which, when composed with the \lifting translation of \Cref{sect:intermediate-calculus}, provides the \emph{\wrapt conversion} transformation. All the unsurprising details (definitions, statements, and proofs) are given in \Cref{sect:app-target-calculus}.
The ideas behind $\tarcal$ and the elimination of names, in \reffig{target_calculus}, are:
\begin{itemize}
\item \emph{Variable binders}: replacing the sequences $\tuvartwo$ and $\tuvar$ of abstracted variables in \wrapts $\pack{\dabsv\vartwo\var\tm}{\bag}$ with the special variables $\lvar$ and $\svar$ (short for \emph{\lifted} and \emph{source}), 
standing for the tuples $\vv\valnone$ and $\vv\valntwo$ meant to be substituted on $\tuvartwo$ and $\tuvar$; 
\item \emph{Variable occurrences}: replacing every occurrence of a variable $\vartwo_i$ or $\var_j$ in $\tm$ with the \emph{projected variables} $\proj_i\lvar$ or $\proj_j\svar$, which are two special compound terms as $\lvar$ and $\svar$ cannot appear without being paired with a projection. 
\end{itemize}
We use $\prvar$ (for \emph{projected variable}) to refer to either $\proj_i\lvar$ or $\proj_j\svar$. 
The transformation is similar to switching to de Bruijn indices, except that---because of the already \lifted setting---the index simply refers to the composite binder of the \wrapt, rather to the nested binders above the variable occurrence. It also slightly differs from standard closure conversion: the standard transformation would eliminate the $\tuvartwo$ names but usually not the $\tuvar$ ones. We eliminate both, as the total elimination of names shall induce a logarithmic speed-up for the abstract machine associated with $\tarcal$, in \refsect{int-complexity}.

The terminology \emph{name elimination} refers to the fact that, after this transformation, a term uses only two variables, $\lvar$ and $\svar$. 
As all \wrapts would then have shape $\pack{\dabs{\lvar}{\svar}\tm}{\bag}$, we simplify the notation and just write $\pack\tm\bag$, with the implicit assumption that every \wrapt now binds $\lvar$ and $\svar$ in $\tm$. 
To define a reverse translation of $\tarcal$ to $\intcal$ (given in \Cref{sect:app-target-calculus}), \wrapts $\packnm\tm\bag$ are annotated with two natural numbers $n,m\in\nat$, which record the length of the replaced sequences of variables $\tuvartwo$ and $\tuvar$. To simplify the notation, however, we shall omit these annotations when not relevant. The \emph{body} of a \wrapt $\pack\tm\bag$ is $\tm$, and---as in $\intcal$---\wrapts are either variable \wrapts, if $\bag=\tuv\prvar$, or evaluated \wrapts, if $\bag=\vv\val$. A term $\tm\in\tarcal$ is \emph{closed} if $\proj_i \efvar$ and $\proj_i \envar$ do not occur out of \wrapt bodies.

Name elimination is parametric in two 
lists $\tuvartwo$ and $\tuvar$ of abstracted variables, which intuitively are those of the closest enclosing closure being translated. 
In particular the two parametric 
lists change when the translation crosses the boundary of a closure in the $\clc{\pack{\dabsv\varthree\varfour\tm}{\bag} }\vvvartwo\vvvar$ clause in \reffig{target_calculus}. On closed terms, the translation is meant to be applied with empty parameter lists (as $\clc\tm\emptylist\emptylist$). 

\paragraph{Results} In \Cref{sect:app-target-calculus}, we prove that name elimination, its reverse naming translation (from $\tarcal$ to $\intcal$), and the composed naming-unwrapping reverse translation (from $\tarcal$ to $\soucal$) are termination-preserving strong bisimulations as in \refthm{rev-transl-implem-bricks}. The final correctness theorem is \refthm{final-bisimulation}, page \pageref{thm:final-bisimulation}. The technical development is rather smooth and follows the structure of 
\Cref{sect:intermediate-calculus} without the subtleties related to the commutation with substitution.

 \section{Part 2 Preliminaries: Abstract Machines}
 \label{sect:prelim-am}
This section starts the second part of the paper. Here, we introduce the terminology and the form of implementation theorem that we adopt for our abstract machines, along the lines of Accattoli and co-authors \cite{DBLP:conf/icfp/AccattoliBM14,DBLP:journals/scp/AccattoliG19,DBLP:conf/ppdp/AccattoliCGC19}, here adapted to handle clashes.
 
\paragraph{Abstract Machines Glossary.}  
An \emph{abstract machine} for a strategy $\tostrat$ of a calculus $\xcal$ is a quadruple $\mach = (\States, \tomach, \compil\cdot, \decode\cdot)$ where $(\States, \tomach)$ is a labeled transition 
system with transitions $\tomach$ partitioned into \emph{principal transitions} $\tomachpr$, 
corresponding to the steps of the strategy and 
labeled with the labels  of the rewrite rules in $\xcal$ (here variants of $\betav$ and $\proj$ steps), and \emph{overhead transitions} $\tomacho$, that take care of the various tasks of the machine (searching, substituting, and $\alpha$-renaming), together with two functions:
\begin{itemize}
\item \emph{Initialization} $\compil\cdot\colon\xcal\to\States$  turns $\xcal$-terms into 
states;
\item \emph{Read-back} $\decode\cdot\colon\States\to\xcal$ turns states into 
$\xcal$-terms and satisfies the 
constraint $\decode{\compil\tm}=\tm$ for every $\xcal$-term $t$. 
\end{itemize}

A state $\state\in\States$ is composed by the \emph{active term} $\tm$, and some data structures. A state $\state$ is \emph{initial} for $\tm$ if $\compil\tm = \state$. 
A state is \emph{final} if no transitions apply; final states are partitioned into \emph{successful} and \emph{clash} states.
 A \emph{run} $\run
 $ is a possibly empty sequence of transitions. 
 For runs, we use notations such as $\size\run$ as for evaluation sequences (\refsect{plotkin}). 
 An \emph{initial run} (\emph{from} $\tm$) is a run from an initial state $\compil\tm$. 
 A state $\state$ is \emph{reachable} if it 
is the target state of an~initial~run. 

Abstract machines manipulate \emph{pre-terms}, that is, terms without implicit $\alpha$-renaming, even if for simplicity we keep calling them terms.
In that setting, we write $\rename{\tm}$ in a state $\state$ for a {\em fresh renaming} of $\tm$,
\ie $\rename{\tm}$ is $\alpha$-equivalent to $\tm$ but all of its bound variables
are fresh (with respect to those in $\tm$ and in the other components of $\state$).

\paragraph{Implementation Theorem, Abstractly.} We now define when a machine \emph{implements} the strategy $\tostrat$ of a calculus $\xcal$, abstracting and generalizing the setting of the previous sections.
\begin{definition}[Machine implementation]
A machine $\mach=(\States, \allowbreak\tomach, \compil\cdot, \decode\cdot)$ \emph{implements the strategy} $\tostrat$ of a calculus $\xcal$ when given a $\tm\in\xcal$ the following 
holds:\label{def:implem}
\begin{enumerate}
\item \emph{Runs to evaluations}: for any $\mach$-run $\exec: \compil\tm \tomach^* \state$ there exists a 
$\tostrat$-evaluation $\deriv: \tm \tostrat^* \decode\state$. Additionally, if $\state$ is a successful state then $\decode\state$ is a clash-free $\tostrat$-normal form.

\item \emph{Evaluations to runs}: for every $\tostrat$-evaluation $\deriv: \tm \tostrat^* \tmtwo$ there exists a 
$\mach$-run $\exec: \compil\tm \tomach^* \state$ such that $\decode\state = \tmtwo$. Additionally, if $\tmtwo$ is a clash-free $\tostrat$-normal form then there exists a successful state $\statetwo$ such that $\state \tomacho^* \statetwo$.

\item \emph{Principal matching}: in both previous points the number $\sizep\deriv\lab$ of steps of of the evaluation $\deriv$ of label $\lab$ are exactly the number $\sizep\exec{\lab}$ of principal $\lab$ transitions in $\exec$, \ie $\sizep\deriv\lab = \sizep\exec\lab$.
\end{enumerate}
\end{definition}
Next, we give sufficient conditions that a machine and a deterministic strategy have to satisfy in order for the former to implement the latter, what we call \emph{an implementation system}. 

\begin{definition}[Implementation system]
  \label{def:implementation}
  A machine $\mach=(\States, \allowbreak \tomach, \compil\cdot, \decode\cdot)$ and a strategy $\tostrat$ form an \emph{implementation system} if:
  \begin{enumerate}
		\item\label{p:def-overhead-transparency} \emph{Overhead transparency}: $\state \tomacho \statetwo$ implies $\decode\state = \decode\statetwo$;
		\item\label{p:def-beta-projection} \emph{Principal projection}: $\state \tomachpr \statetwo$ implies $\decode\state \tostrat \decode\statetwo$ and the two have the same label;
		\item\label{p:def-overhead-terminate}	\emph{Overhead termination}:  $\tomacho$ terminates;	
	\item\label{p:def-halt} \emph{Halt}:  $\mach$ successful states read back to $\tostrat$-normal forms, and clash states to clashes of $\xcal$.
  \end{enumerate}
\end{definition}

Via a simple lemma for the \emph{evaluation to runs} part (in \Cref{sect:app-prelim-am}), we obtain the following abstract implementation theorem.

\begin{toappendix}
\begin{theorem}[Sufficient condition for implementations]
\label{thm:abs-impl}
  Let  $\mach$ be a machine and $\tostrat$ be a strategy forming an implementation system.
  Then, $\mach$ implements $\tostrat$ (in the sense of \refdef{implem}). 
\end{theorem}
\end{toappendix}

 \label{sect:lam-environments}
 \paragraph{Local Environments}
 We overview one the two main forms of environments for abstract machines, the \emph{local} one (as opposed to \emph{global}), for the 
 archetypal
 call-by-value calculus $\cbvcal$ of \refsect{plotkin}, to compare it later with the \TAM for $\soucal$. 
 The terminology local/global is due to \cite{DBLP:journals/entcs/FernandezS09}, and the two techniques are 
 analyzed e.g. (for call-by-name and call-by-need) in \cite{DBLP:conf/ppdp/AccattoliB17} but \mbox{they are~folklore}. 
\begin{figure}[t]
\centering
\scalebox{0.93}{
\begin{tabular}{c@{\hspace{-.2cm}}cc}
$\begin{array}{c@{\hspace{.5cm}} c ccccccc}
 \textsc{\Mclosures} & \textsc{Local Envs} & \textsc{Stacks} 
\\
\bareclos \grameq \pair\tm\env &
\env  \grameq  \emptylist \mid \esub\var\bareclos \!\cons\! \envtwo 
&
\stack  \grameq  \emptylist   \mid   \clos \cons \stack \mid  \nvsym\bareclos \cons \stack 
\end{array}$
\\[2pt]
\!\!\!$\begin{array}{|l|r|r|@{\hspace{.25cm}}c@{\hspace{.25cm}}|l|r|r|}
\hline
\multicolumn{2}{|c|}{\textsc{M-Clos.}} & \textsc{Stack} &\textsc{Trans.}& \multicolumn{2}{c|}{\textsc{M-Clos.}} & \textsc{Stack}
\\\hline
 \tm\,\tmtwo   & \env & \stack
&\tomachseaone&
\tmtwo   & \env & \nvsym\pair\tm\env \cons\stack
\\

\la\var\tm   & \env &\nvsym\bareclos\cons\stack
&\tomachseatwo &
\multicolumn{2}{c|}{\bareclos} & \evsym\pair{\la\var\tm}\env \cons\stack
\\

\la\var\tm   & \env &\clos\cons\stack
&\tomachbeta &
\tm  & \esub\var\bareclos \cons \env & \stack
\\

 \var   & \env &\stack
& \tomachsub &
\multicolumn{2}{c|}{\env(\var)} & \stack
\\
\hline
\end{array}$
\end{tabular}}

\vspace*{-.5\baselineskip}
\caption{The local abstract machine (\LAM) for $\cbvcal$}
\label{fig:machine_local_cbv}

\vspace*{-.5\baselineskip}
\end{figure}
 
  The  \emph{local abstract machine} (\LAM), a machine with local environments for $\cbvcal$, is defined in \reffig{machine_local_cbv}. 
  It is a right-to-left variant of the CEK machine \cite{DBLP:conf/ifip2/FelleisenF87}.
  It uses a \emph{stack} to search for $\betav$-redexes, filling it with entries that encode an evaluation context of $\cbvcal$. The substitutions triggered by the encountered $\betav$-redexes are delayed and stored in environments. 
  There is one principal transition $\tomachbeta$, of label $\betav$, and three overhead transitions $\tomachseaone$, $\tomachseatwo$, and $\tomachsub$ realizing the search for $\betav$-redexes and 
  substitution.    %
  Contrary to the global approach, 
  the \LAM has many environments $\env$, paired with terms as to form \emph{\mclosures} $\bareclos$. In fact, a local environment $\env$ is itself a list of pairs of variables and \mclosures: \mclosures and environments are defined by mutual induction. 
  A \emph{state} is a triple $(\tm,\env,\stack)$ 
  where $\pair\tm\env$ is the \mclosure of an active term $\tm$ and a local environment $\env$, and $\stack$ is a stack.
  Initial states have shape $(\tm,\emptylist,\emptylist)$. 
  
  \Mclosures are called in this way because when evaluating a closed term an 
  invariant ensures that $\fv\tm\subseteq \dom\env$ for any \mclosure $\pair\tm\env$ in a reachable state (including the active one), where the domain $\dom\env$ of an environment $\env$ is the set of variables on which it has a substitution. The use of \mclosures allows one to avoid $\alpha$-renaming (and copying code) in transition $\tomachsub$, at the price of using many environments, thus using space anyway. 
  The duplication of $\env$ in transition $\tomachseaone$ is where different implementation approaches (shared \textit{vs} flat) play a main role. With shared environments (whose simplest implementation is as linked lists), the duplication only duplicates the pointer to the environment, not the whole environment. With flat environments (whose simplest implementation is as arrays), duplication is an actual duplication of the array; we shall further discuss the duplication of flat environments when discussing the complexity of abstract machines. 
  
\section{Part 2: the \TAM for \texorpdfstring{$\soucal$}{the Source Calculus}}
 \label{sect:TAM}
Here we present a machine with local environments, the \emph{source tupled abstract machine} (\TAM) for the source calculus $\soucal$. 
 We adopt local rather than global environments as to have \mclosures, to then show that wrapping removes their need, in the next section.

\begin{figure*}[t]
	\vspace*{-.6\baselineskip}
\centering
\setlength{\arraycolsep}{3pt}
\scalebox{0.92}{\begin{tabular}{c}
\begin{tabular}{c|c}
$\begin{array}{c}
\textsc{Eval. m-clos.}
\\
\clos \grameq \evlab{\pair{\la\tuvar\tm}\env} \midd  \overgroup{\tuple{\clos_{1},\mydots,\clos_{n}}}^{\vv\clos} \ n \!\geq\! 0
\\[2pt]
\textsc{Flagged m-clos.}
\\ 
\laxclos \grameq \clos \midd \nvlab{\pair\tm\env}
\\[2pt]
\textsc{Local envs}
\\
\env,\envtwo  \grameq   \emptyenv \mid \esub\var\clos \cons \env
\\[2pt]
\textsc{St. entries}
\\ 
\stacke \grameq \laxclos \midd \proj_i \midd \pair{\tuple{\tuv\tm,\machctxhole,\tuv\clos}}\env
\\[2pt]
\textsc{Stacks}
\\ 
\stack, \stacktwo \grameq \emptystack  \midd  \stacke \cons \stack
\\[2pt]
\textsc{States}
\\ 
\state, \statetwo \grameq \pairstate{\laxclos}\stack
\end{array}$
&
\begin{tabular}{c}
$\begin{array}{|c|c|c|r|@{\hspace{.2cm}}c@{\hspace{.2cm}}|c|c|c|r|l}
	\hhline{-|-|-|-|-|-|-|-|-}
\multicolumn{3}{|c|}{\textsc{\mclosure}} & \textsc{Stack}  &\textsc{Trans.} &  \multicolumn{3}{|c|}{\textsc{\mclosure}} & \textsc{Stack}
\\
	\hhline{-|-|-|-|-|-|-|-|-}
\nvsym & \tm\tmtwo & \env  &\stack 
&\tomachseaonenv&
\nvsym & \tmtwo & \env  & \nvlab{\pair\tm\env} \cons\stack 
\\

\nvsym & \proj_i \tm & \env & \stack 
& \tomachseatwonv &
\nvsym & \tm & \env &  \proj_i \cons \stack 
\\

\nvsym & \tuple{\mydots,\tm_n} & \env & \stack 
& \tomachseathreenv &
\nvsym &  \tm_n& \env  & \pair{\tuple{\mydots,\machctxhole}}\env \cons \stack 
\\

\nvsym & \tuple{} & \env  & \stack 
& \tomachseafournv &
\evsym & \tuple{} & \emptyenv & \stack 
\\

\nvsym & \la{\tuvar}\tm & \env  & \stack 
& \tomachseafivenv &
\evsym & \la{\tuvar}\tm & \env & \stack 
\\

\nvsym & \var & \env  & \stack 
& \tomachsubnv &
\multicolumn{3}{c|}{\env(\var)} & \stack 
\\
	\hhline{-|-|-|-|-|-|-|-|-}
\multicolumn{3}{|c|}{}&&&&&&\\[-10pt]
\multicolumn{3}{|c|}{\clos} & \nvlab{\pair\tm\env} \cons\stack 
&\tomachseaoneev&
\nvsym & \tm & \env & \clos \cons\stack 
\\
\multicolumn{3}{|c|}{\clos}  & \pair{\tuple{\mydots,\tm,\machctxhole,\mydots}}\env \cons\stack 
&\tomachseasixev&
\nvsym & \tm & \env &  \pair{\tuple{\mydots,\machctxhole,\clos,\mydots}}\env \cons\stack 
\\

\multicolumn{3}{|c|}{\clos}  & \pair{\tuple{\machctxhole,\mydots}}\env \cons\stack 
&\tomachseathreeev&
\multicolumn{3}{c|}{\tuple{\clos,\mydots}}  & \stack 
\\

\evsym & \la{\tuvar}\tm & \env & \vv\clos \cons\stack 
&\tomachbetaev&
\nvsym & \tm & \esub{\tuvar}{\vv\clos} \cons \env  & \stack &(*)
\\

\multicolumn{3}{|c|}{\vv\clos}  & \proj_i \cons \stack 
& \tomachprojev &
\multicolumn{3}{c|}{\clos_i} & \stack &(\#)
\\
\hhline{-|-|-|-|-|-|-|-|-}
\end{array}$
\\
\begin{tabular}{c@{\hspace{.25cm}}|@{\hspace{.25cm}}c}
Side conditions: 
&
\textsc{States read back $\decode{\cdot}$ (to terms of $\soucal$)}
\\
($*$) if $\norm{\tuvar} =\norm{\vv{\clos}}$ \ \ \ \ (\#) if $1 \!\leq\! i \!\leq\! \norm{\vv\clos}$
&
	$\decode{\pairstate{\laxclos}{\stack}}  \defeq  \decode\stack\ctxholep{\decode\laxclos}$
\end{tabular}
\end{tabular}
\end{tabular}
\\[76pt]\hline

\begin{tabular}{c@{\hspace{.25cm}}|@{\hspace{.25cm}}c}
$\begin{array}{rll@{\hspace{.75cm}}rll}
\multicolumn{6}{c}{\textsc{\Mclosures read-back $\decode{\cdot}$ (to terms of $\soucal$)}}
\\
\decode{\flag{\pair\tm{\emptyenv}}}  &\!\!\defeq\!\!&  \tm
&
\decode{\tuple{\clos_{1},\mydots,\clos_{n}}} & \!\!\defeq\!\!&  \tuple{\decode{\clos_{1}},\mydots,\decode{\clos_{n}}}
\\
\multicolumn{6}{c}{\decode{\flag{\pair\tm{\esub\var\clos\!\cons\!\env}}} \,\defeq\,  \decode{\pair{\tm\isub\var{\decode\clos}}{\env}}}
\end{array}$
&
$\begin{array}{rll@{\hspace{.2cm}}rll@{\hspace{.75cm}}rll}
\multicolumn{9}{c}{\textsc{Stacks read-back $\decode{\cdot}$ (to evaluation contexts of $\soucal$)}}
\\
\decode\emptystack  &\defeq&  \ctxhole
&
\decode{\nvlab{\pair\tm\env} \cons \stack} & \defeq&  \decode\stack\ctxholep{\decode{\nvlab{\pair\tm\env}} \ctxhole}
&
\decode{\clos \cons \stack} & \defeq&  \decode\stack\ctxholep{\ctxhole \decode\clos}
\\
\decode{\proj_{i}\cons\stack} & \defeq&  \decode\stack\ctxholep{\proj_{i} \ctxhole}
&
 \multicolumn{6}{c}{\decode{\pair{\tuple{\tm_{1},\mydots,\tm_{n},\machctxhole,\tuv\clos}}\env \cons \stack}
\,\defeq\,  \decode\stack\ctxholep{\tuple{\decode{\pair{\tm_{1}}\env},\mydots, \decode{\pair{\tm_{n}}\env},\ctxhole,\tuv{\decode\clos}}} }
\end{array}$
\end{tabular}
\end{tabular}}

\vspace*{-.5\baselineskip}
\caption{The source tupled abstract machine (\TAM).}
\label{fig:tam}
\end{figure*}

The \TAM is defined in \reffig{tam}. 
\Mclosures carry a flag $\flag\in\set{\nvsym,\evsym}$, where $\nvsym$ stands for \emph{non-(completely)-evaluated} and $\evsym$ for \emph{evaluated}. In $\soucal$, values have a tree structure, the leaves of which are abstractions. The evaluated \mclosures $\clos$ of the \TAM have a similar tree structure, plus local environments. 
A \emph{state} is a couple $\pairstate{\laxclos}{\stack}$, where $\laxclos$ is a \mclosure and $\stack$ is a stack. 
The active \mclosure is also flagged, naturally inducing a partition of transitions in two blocks. 
Stack entries are flagged \mclosures, $\proj_i$ and \emph{partially evaluated tuples} $\pair{\tuple{\tuv\tm,\machctxhole,\tuv\clos}}\env$ ($\tuv\tm$ are non-evaluated~terms).

The \emph{initialization} of $\tm$ is the \emph{initial state} $\compil\tm \defeq \pairstate{\nvsym\pair\tm\emptyenv}\emptystack$.
\emph{Successful states} are $\pairstate{\clos}\emptystack$, that is, an evaluated \mclosure and an empty stack.
Clash and clash-free states are defined in \Cref{sect:app-TAM}.

\paragraph{Transitions} The union of all the transitions of the \TAM is noted $\totam$. The principal transitions are $\tomachbetaev$ and $\tomachprojev$, of label $\betav$ and $\proj$, all the other transitions are overhead ones. If  $\tuvar = \var_{1},\mydots,\var_{n}$ and $\vv\clos = \tuple{\clos_{1},\mydots,\clos_{n}}$, we set $\esub{\tuvar}{\vv\clos} \defeq \esub{\var_{1}}{\clos_{1}}\mydots \esub{\var_{n}}{\clos_{n}}$; this notation is used in transition $\tomachbetaev$. In presence of tuples, there are additional overhead transitions. If the active \mclosure is:
\begin{itemize}
	\item $\nvsym\pair{\proj_i\tm}\env$ then $\tomachseatwonv$ triggers the evaluation of $\tm$ and the projection $\proj_i$ goes on the stack (to trigger transition $\tomachprojev$); 
	\item $\nvsym\pair{\la{\tuvar}\tm}\env$ then $\tomachseafivenv$ flips the flag to $\evsym$ (weak evaluation);
	\item $\nvsym\pair{\tuple{}}\env$ then $\tomachseafournv$ changes the flag to $\evsym$, discarding $\env$;
	\item $\nvsym\pair{\tuple{\mydots,\tm_n}}\env$ then $\tomachseathreenv$ 
	evaluates its elements right-to-left, adding a partially evaluated tuple to the stack $\pair{\tuple{\mydots,\machctxhole}}\env$.
	\item $\clos$ then the behavior depends on the first element of the stack. If it is a partially evaluated tuple $\pair{\tuple{\mydots,\tm,\machctxhole,\mydots}}\env$ then $\tomachseasixev$ swaps $\clos$ and the next element $\tm$ in the tuple 
	(duplicating $\env$),
	similarly to $\tomachseaoneev$. 
	If the tuple on the stack is $\pair{\tuple{\machctxhole,\mydots}}\env$ then $\tomachseathreeev$ plugs $\clos$ on $\machctxhole$ forming a new evaluated \mclosure. 
\end{itemize}
The name of transitions $\tomachseaoneev$ and $\tomachseathreeev$ 
stresses that they do the dual job of 
$\tomachseaonenv$ and $\tomachseathreenv$---such a duality shall be exploited in the complexity analysis of \refsect{source-complexity}. Transition $\tomachseasixev$ has no dual (it is not named $\tomachhole{\evsym sea_2}$ because it is not the dual of $\tomachseatwonv$).

\paragraph{Invariant and Read Back} 
Here is an invariant of the \TAM.
\begin{toappendix}
\begin{lemma}[\mclosure invariant]
\label{l:TAM-closure-invariant}
Let $\state$ be a \TAM reachable state and $\laxclos=\flag\pair\tm\env$ be a \mclosure or $\pair{\tuple{\mydots,\tm,\mydots,\machctxhole,\tuv\clos}}\env$ be a stack entry in $\state$. 
Then $\fv\tm \subseteq \dom\env$.
\end{lemma}
\end{toappendix}
The read-back $\decode{\cdot}$ of the \TAM to $\soucal$ is defined in \reffig{tam}. 
\Mclosures and states read back to terms, 
stacks read back to evaluation contexts. 
In the read back of partially evaluated \mclosures $\pair{\tuple{\tuv\tm,\machctxhole,\tuv\clos}}\env$ on the stack, $\env$ spreads on the non-evaluated terms $\tuv\tm$ (the \mclosures in $\tuv\clos$ have their own environments on their leaves).

\begin{toappendix}
\begin{lemma}[Read-back properties]
\label{l:TAM-decoding-invariants}
\hfill
	\begin{enumerate}
		\item
		$\decode{\clos}$ is a value of $\soucal$ for every 
		$\clos$ of the \TAM.
		\item
		$\decode{\stack}$ is an evaluation context of $\soucal$ for \mbox{every \TAM stack $\stack$}.
	\end{enumerate}
\end{lemma}
\end{toappendix}

\paragraph{Implementation Theorem} According to the recipe in \refsect{prelim-am}, we now prove the properties inducing the implementation theorem. The read-back 
to evaluation contexts (\Cref{l:TAM-decoding-invariants}) is used for principal projection. 
The \mclosure invariant (\Cref{l:TAM-closure-invariant}) is used in the proof of the halt property, to prove that the machine is never stuck on the left-hand side of a $\tomachsubnv$ transition (see \Cref{sect:app-TAM}). 
Overhead termination is proved via a measure, developed in \refsect{source-complexity}, which gives a bound on the number of overhead transitions. 



\begin{toappendix}
\begin{theorem}
	\label{thm:sou-tam-implementation}
	The \TAM and $\soucal$ form an implementation system (as in \refdefeq{implementation}), hence the \TAM implements~$\tosou$.
\end{theorem}
\end{toappendix}

  \section{Part 2: the \LTAM for \texorpdfstring{$\intcal$}{the Intermediate Calculus}}
 \label{sect:LTAM}

Here we present the \emph{intermediate tupled abstract machine} (\LTAM) for the intermediate calculus $\intcal$. 
Its feature is a new way of handling environments, resting on the strong properties of $\intcal$.

\begin{figure*}[t!]
\centering
\scalebox{0.92}{\begin{tabular}{c}
$\begin{array}{r@{\hspace{.3cm}} lll@{\hspace{.5cm}}r@{\hspace{.3cm}} lll}
\textsc{Stackable envs} &
\env, \envtwo &\grameq& \emptyenv \midd \esub\var{\evlab\val} \cons \env
&
\textsc{Activation stacks} & 
\ars,\arstwo & \grameq & \emptystack \midd \pair\stack\env \cons \ars

\\
\textsc{Constructor stacks} &
\stack, \stacktwo &\grameq& \emptystack \midd \nvlab\tm\cons \stack \midd \evlab\val\cons \stack \midd \proj_i \cons \stack \midd \tuple{\tuv{\nvsym\tm},\machctxhole,\tuv{\evsym\val}}\cons\stack 
&
\textsc{States} & 
\state,\statetwo & \grameq & \fourstate{\flag\tm} \stack \env \ars
\end{array}$
\\\\[-5pt]
{\footnotesize
$\begin{array}{|l|r|r|r|@{\hspace{.25cm}}c@{\hspace{.25cm}}|l|r|r|r|l}
\cline{1-9}
\textsc{Focus} & \textsc{Cons. stack} &\textsc{Env}& \textsc{Act. stack} &\textsc{trans.}& \textsc{Focus} & \textsc{Cons. stack}&\textsc{Env}& \textsc{Act. stack}
\\
\cline{1-9}
\nvsym  (\tm\tmtwo)  & \stack & \env & \ars
&\tomachseaonenv&
\nvsym  \tmtwo  & \nvlab\tm \cons\stack & \env & \ars
\\

\nvsym   \proj_i \tm  & \stack & \env & \ars
& \tomachseatwonv &
\nvsym  \tm  &\proj_i \cons \stack & \env & \ars
\\

\nvsym  \tuple{\tm_1,\mydots,\tm_n}  &\stack & \env & \ars
& \tomachseathreenv &
\nvsym   \tm_n & \tuple{\nvsym\tm_1,\mydots,\machctxhole} \cons \stack & \env & \ars
\\

\nvsym  \tuple{}  &\stack & \env & \ars
& \tomachseafournv &
\evsym   \tuple{} & \stack & \env & \ars
\\


\nvsym \pack{ \dabsv\vartwo\var\tm}{\tuple{\tuvartwo}}  &\stack & \env & \ars
& \tomachsubwev &
\evsym \pack{ \dabsv\vartwo\var\nvsym\tm}{\env(\tuvartwo)} & \stack & \env & \ars
\\

\nvsym  \var  &\stack & \env & \ars
& \tomachsubvev &
\env(\var) & \stack & \env & \ars
\\
\cline{1-9}

\evsym  \val &\nvlab\tm \cons\stack & \env & \ars
&\tomachseaoneev&
\nvsym  \tm & \evlab\val \cons\stack & \env & \ars
\\

\evsym  \val & \tuple{\mydots,\nvsym\tm,\machctxhole,\mydots} \cons\stack & \env & \ars
&\tomachseasixev&
\nvsym  \tm & \tuple{\mydots,\machctxhole,\evlab\val,\mydots}\cons\stack & \env & \ars
\\

\evsym  \val   & \tuple{\machctxhole,\mydots} \cons\stack & \env & \ars
&\tomachseathreeev&
\evsym  \tuple{\evlab\val,\mydots}  &\stack & \env & \ars
\\

\evsym  \vv\val  & \proj_i \cons \stack & \env & \ars
& \tomachprojev &
\evsym  \val_i &\stack & \env & \ars
&\mbox{if $1\leq i \leq \norm{\vv\val}$}
\\

\evsym  \pack{ \dabsv\vartwo\var\nvsym\tm}{\vv{\evlab{\val_1}}}  & \vv{\evlab{\val_2}} \cons\stack & \env & \ars
&\tomachbetaev&
\nvsym  \tm  & \emptystack & \!\!\esub{\tuvartwo}{\vv{\evsym\valnone}}\esub{\tuvar}{\vv{\evsym\valntwo}}\!\! & \pair\stack\env \cons \ars
& \mbox{if $\norm{\tuvartwo}\!=\!\norm{\vv{\evlab{\val_2}}}$, $\norm{\tuvar}\!=\!\norm{\vv{\evlab{\val_1}}}$}
\\

\evsym  \val  & \emptystack & \env & \pair\stack\envtwo\cons \ars
&\tomachseasevenev&
\evsym  \val &\stack & \envtwo & \ars
\\
\cline{1-9}
\end{array}$
}
\\\\[-8pt]
\hline
\textsc{Read back $\decode{\cdot}$ to $\intcal$}
\\[3pt]
$\begin{array}{r@{\hspace{.35cm}} r@{\hspace{.7cm}} l@{\hspace{.7cm}} llll}
\textsc{Flagged terms} &
\decode{\nvlab\tm}  \defeq  \tm
& \decode{\evsym  \pack{ \dabsv\vartwo\var\nvsym\tm}{\vv\evval} }  \defeq    \pack{ \dabsv\vartwo\var\tm}{\vv{\decode\evval}} 
& \decode{\evsym  \tuple{\evsym\val_1,\ldots,\evsym\val_n} }  \defeq      \tuple{\decode{\evsym\val_1},\ldots,\decode{\evsym\val_n}} 
\end{array}$
\\[5pt]
$\begin{array}{r@{\hspace{.35cm}} r@{\hspace{.7cm}} l@{\hspace{.7cm}} l@{\hspace{.7cm}} l@{\hspace{.7cm}} l}
\textsc{Constructor stacks} &
\decode\emptystack  \defeq  \ctxhole
&
 \decode{\proj_{i}\cons\stack}  \defeq  \decodep\stack{\proj_{i} \ctxhole}
&
 \decode{\evlab\val \cons \stack}  \defeq  \decodep\stack{\ctxhole \decode\evval}

&
\decode{\nvlab\tm \cons \stack}  \defeq  \decodep\stack{\tm \ctxhole}
& \decode{\tuple{\tuv{\nvsym\tm},\machctxhole,\tuv{\evsym\val}}\cons\stack}
 \defeq  \decodep\stack{\tuple{\tuv\tm,\ctxhole,\tuv{\decode\evval} }}
\end{array}$
\\[5pt]

$\begin{array}{rll@{\hspace{.15cm}}|@{\hspace{.15cm}}rl@{\hspace{.15cm}}|@{\hspace{.15cm}}rl}
\!\!\!\!\!\!\!x\textsc{Act. stacks} &
\decode\emptystack  \defeq  \ctxhole
&
\decode{\pair\stack\env \cons \ars}  \defeq  \decodep\ars{\decode{\stack}\sigma_\env}
&
\textsc{States} & 
\decode{\fourstate{\flag\tm}\stack\env\ars}  \defeq  \decodep\ars{\decode\stack\sigma_\env\ctxholep{\decode{\flag\tm}\sigma_\env}}
&
\textsc{Env-induced subst.}
&
\sigma_{\esub{\tuvar}{\vv\evval}}  \defeq  \isubp{\tuvar}{\vv{\decode\evval}}\emptylist{\tuple{}}
\end{array}$
\end{tabular}}

\vspace*{-.5\baselineskip}
\caption{The intermediate tupled abstract machine (\LTAM).}
\label{fig:lifted-tam}

\vspace*{-.5\baselineskip}
\end{figure*}

The \LTAM is defined in \reffig{lifted-tam}. 
The flags will be explained after an overview of the new aspects of the machine. 
Mostly, the \LTAM behaves as the \TAM by just having removed the structure of \mclosures. 
The principal transitions are $\tomachbetaev$ and $\tomachprojev$, of label $\betav$ and $\proj$, all the other transitions are overhead. There are two new aspects: transition $\tomachsubwev$ that evaluates variable \bagts 
and the use of \emph{stackable environments} rather than \emph{local} environments and \mclosures. 
The machine is akin to those with global environments, except that 
no $\alpha$-renaming is needed, as in the~local~approach. 

\paragraph{Evaluating \Bagts.} 
In the \LTAM, abstractions $\la\tuvar\tm$ are replaced by non-evaluated prime \wrapt{s} $\nvsym\pack{ \dabsv\vartwo\var\tm }{ \tuple{\tuvartwo}}$. 
The new transition $\tomachsubwev$ (replacing $\tomachseafivenv$) substitutes on all variables in $\tuvartwo$ in one shot, producing the evaluated \wrapt $\evsym\pack{ \dabsv\vartwo\var\tm }{ \env(\tuvartwo)}$ where $\env(\tuvartwo) = \tuple{\env(\vartwo_1),\mydots,\env(\vartwo_n)}$  if $\tuvartwo = \vartwo_1, \mydots,\vartwo_n$, and $\env$ is the (global) environment.

\paragraph{\Stackable Environments.} 
Transitions $\tomachbetaev$ and the new $\tomachseasevenev$ encapsulate a second new aspect,  \emph{stackable environments}. Indeed: 
\begin{itemize}
\item \emph{\Wrapt bodies and environments}: when the machine encounters the analogous of a $\toibv$-redex $\pack{ \dabsv\vartwo\var\tm }{ \vv{\val_1} }\, \vv{\val_2}$, the new entries $\esub{\tuvartwo}{\vv{\evsym\valnone}}\,\esub{\tuvar}{\vv{\evsym\valntwo}}$ of the environment created by transition $\tomachbetaev$ in the machine are all that is needed to evaluate $\tm$, because the free variables of $\tm$ are all among $\tuvartwo$ and $\tuvar$. 
Thus, the environment $\env$ that is active before firing the redex is useless to evaluate $\tm$, and can be removed after transition~$\tomachbetaev$. 

\item \emph{Stackability}: $\env$ is not garbage collected, it is pushed on the new \emph{activation stack}, 
along with the ordinary stack $\stack$ (now called \emph{constructor stack}) which contains non-evaluated terms with variables in $\dom\env$. 
This is still done by transition~$\tomachbetaev$.

\item \emph{Popping}: 
when the body $\tm$ of the \wrapt has been evaluated, the focus is on a value $\evval$ and the constructor stack is empty. 
The activation stack has the pair $(\stack,\env)$ that was active before firing the $\betav$-redex. The machine throws away the current environment $\envtwo \!\defeq \esub{\tuvartwo}{\vv{\evsym\valnone}} \esub{\tuvar}{\vv{\evsym\valntwo}}$, since $\tuvartwo$ and $\tuvar$ have no occurrences out of $\tm$, and reactivates the pair $(\stack,\env)$, to keep evaluating terms in $\stack$. This is done \mbox{by the new transition $\tomachseasevenev$.}
\end{itemize}

\paragraph{Flags}
The \LTAM evaluates (well-formed) terms of $\intcal$ decorated (only on top) with a flag $\flag\in\set{\nvsym,\evsym}$: $\nvsym\tm$ denotes that $\tm$ has not been evaluated yet, while $\evsym\tm$ denotes that $\tm$ has been evaluated. 
The results of 
evaluation are values, thus the $\evsym$ flag shall be associated to values only, and an invariant shall ensure that every evaluated value $\evsym\val$ in a reachable state is closed. In $\nvsym\tm$, $\tm$ has no flags, and non-flagged sub-terms are implicitly considered as flagged with $\nvsym$. Every evaluated value carries a $\evsym$ flag, so in $\evval$ there are in general other $\evsym$ flags (when it is a tuple or 
the \bagt of the \wrapt). 
Moreover, evaluated \wrapts shall have shape $\evsym\pack{ \dabsv\vartwo\var\nvsym\tm}{\vv\evval}$, that is, they have an additional (redundant) $\nvsym$ flag on their body (it shall be used in section \refsect{int-complexity} for the complexity analysis). 
The machine is started on \emph{prime} terms of $\intcal$ (because \lifted terms of $\soucal$ are prime, \reflemma{lali-properties}), thus all \wrapts in an initial state have shape $\pack{ \dabsv\vartwo\var\tm}{\tuple{\tuvartwo}}$. 
In fact, all non-evaluated \wrapts $\nvsym\pack{ \dabsv\vartwo\var\tm}{\bag}$ in reachable states shall always be prime (that is, such that $\bag=\tuple{\tuvartwo}$): we prove this invariant and we avoid assuming the general shape of non-evaluated closures, 
which would require \mbox{additional 
and never used transitions.} 

The \emph{initialization} $\compil\tm$ of $\tm \in \intcal$ is given by the \emph{initial state} $\fourstate{\nvsym\tm}\emptyenv\emptystack \emptystack$ where $\tm$ is closed and prime. 
\emph{Successful states} have shape $\fourstate{\evsym\val}\emptystack\env \emptystack$. Clash states are defined in \Cref{sect:app-LTAM}.

\paragraph{Invariants and Read Back} 
Here are the invariants of the \LTAM.

\begin{toappendix}
\begin{lemma}[Invariants]
\label{l:LTAM-invariants}
Let $\state=\fourstate{\flag\tm}\stack\env\ars$ be a \LTAM reachable state.
\begin{enumerate}
	\item \emph{Well-formedness}: all \wrapts in $\state$ are well-formed.
	\item \emph{Closed values}: every value $\evlab\val$ in $\state$ is closed.
	\item \emph{Closure}: $\fv\tm \cup \fv\stack\subseteq \dom\env$ and $\fv\stacktwo\subseteq \dom\envtwo$ for every entry $\pair\stacktwo\envtwo$ of the activation stack $\ars$.
\end{enumerate}
\end{lemma}
\end{toappendix}

The read back $\decode{\cdot}$ is defined in \reffig{lifted-tam}. Flagged terms $\flag\tm$ and states $\state$ read back to terms of $\intcal$, 
the constructor and activation stacks $\stack$ and $\ars$ read back to evaluation contexts of $\intcal$. The read back of $\ars$ and $\state$ is based on a notion of meta-level \emph{environment-induced (simultaneous) substitution} $\sigma_\env$, defined in \reffig{lifted-tam} and applied to 
terms (as read back of flagged terms) and evaluation contexts (as read back of stacks); meta-level substitutions are extended to evaluation contexts as expected, the definition is in \Cref{sect:app-LTAM}. 

The implementation theorem is proved following the schema used for the \TAM, see \Cref{sect:app-LTAM}. Overhead transparency for 
transition $\tomachseasevenev$ relies on the closed values invariant (\reflemma{LTAM-invariants}).

\begin{toappendix}
\begin{lemma}[Read back properties]
\label{l:LTAM-decoding-invariants}
\hfill
	\begin{enumerate}
		\item
		\emph{Values}: $\decode{\evval}$ is a value of $\intcal$  for every $\evval$ of the \LTAM.
		\item
		\emph{Evaluation contexts}: $\decode{\stack}$ and $\decode\ars$ are evaluation contexts of $\intcal$ for every \LTAM constructor and activation stacks $\stack$ and $\ars$.
	\end{enumerate}
\end{lemma}
\end{toappendix}


\begin{toappendix}
\begin{theorem}
\label{thm:int-ltam-implementation}
The \LTAM and $\intcal$ form an implementation system (as in \refdefeq{implementation}), thus the \LTAM implements~$\toint$ on prime~terms.
\end{theorem}
\end{toappendix}

 \section{Part 2: the \TTAM for \texorpdfstring{$\tarcal$}{}}
 \label{sect:TTAM}
Here we present the \emph{target tupled abstract machine} (\TTAM) for the target calculus $\intcal$, a minor variant of the \LTAM. 
Beyond the elimination of variable names, 
its key feature is the use of \emph{tupled environment}, that is, 
a pair of tuples as data structures for environments, instead of a list of explicit \mbox{substitution entries $\esub\var\evval$}.

\begin{figure}[t]
\centering
\!\!\!\!\!\!%
\scalebox{0.93}{
\begin{tabular}{c}
\setlength{\arraycolsep}{2pt}
$\begin{array}{rll@{\hspace{.6cm}} r lll}
\multicolumn{3}{c}{\textsc{Tupled envs}}
&
\multicolumn{4}{c}{\textsc{Tupled envs lookup }(\vvar\!\in\!\set{\lvar,\svar})}
\\
\env, \envtwo &\grameq& \vv{\evval_\lvar};\vv{\evval_\svar}
 & 
 (\vv{\evval_\lvar};\vv{\evval_\svar}) (\proj_i\vvar) &\defeq &(\vv{\evval_\vvar})_i & \mbox{if }i\leq\norm{\vv{\evval_\vvar}}
\end{array}$
\\\\[-5pt]
{\footnotesize
$\begin{array}{|l|r|r|c|@{\hspace{.1cm}}c@{\hspace{.1cm}}|l|r|r|r|}
\hline
\textsc{Focus} & \!\textsc{Co.}\!\! & \!\textsc{En}\!\! & \!\!\textsc{Ac.}\!\!\! &\textsc{trans.}& \textsc{Focus} & \!\textsc{Co.}\!\! & \textsc{En}& \!\!\textsc{Act. st.}\!\!
\\
\hline
\nvsym \pack{ \tm}{\vv\prvar}  &\stack & \env & \ars
& \tomachsubwev &
\evsym \pack{ \nvsym\tm}{\vv{\env(\prvar)}}\!\! & \stack & \env & \ars
\\

\nvsym  \prvar  &\stack & \env & \ars
& \tomachsubvev &
\env(\prvar) & \stack & \env & \ars
\\
\hline

\evsym  \pack{ \nvsym\tm}{\vv{\evval_1}}\!\!\!  & \!\vv{\evval_2} \!\cons\! \stack\!\! & \env & \ars
&\tomachbetaev&
\nvsym  \tm  & \emptystack & \!\!\vv{\evval_1};\vv{\evval_2}\!\!\! & \!\pair\stack\env \!\cons\! \ars\!\!\!
\\
\hline
\multicolumn{9}{c}{}
\\[-8pt]
\multicolumn{9}{c}{\mbox{Side conditions: if look-up is defined in $\tomachsubwev$ and  $\tomachsubvev$.}}
\end{array}$
}
\\\\[-8pt]
\hline
\textsc{Read back $\decode\cdot$ to $\tarcal$}
\\[3pt]
$\begin{array}{r@{\hspace{.3cm}} r c l@{\hspace{.85cm}}l}
\textsc{Flagged terms} 
& \decode{\evsym  \pack{\nvsym\tm}{\vv\evval} }  &\!\!\!\!\defeq\!\!\!\!&    \pack{ \tm}{\vv{\decode\evval}} 
\\
\textsc{Env-induced subst.} &
\sigma_{\vv{\evval_\lvar};\vv{\evval_\svar}}  &\!\!\!\!\defeq\!\!\!\!&  \itsubp\lvar{\vv{\evval_\lvar}}\svar{\vv{\evval_\svar}} 
\end{array}$
\end{tabular}
}

\vspace*{-.5\baselineskip}
\caption{The target tupled abstract machine (\TTAM).}
\label{fig:target-tam}

\vspace*{-.5\baselineskip}
\end{figure}

The \TTAM is defined in \reffig{target-tam}, by giving the only ingredients of the \LTAM that are redefined for the \TTAM, leaving everything else unchanged but for the fact that, when considering the omitted transitions of the \LTAM as transitions of the \TTAM, the symbol $\env$ for environment refers to the new notion of environment adopted here (the omitted transitions do not touch the environment). Here the \wrapts $\packnm\tm\bag$ of $\tarcal$ are written $\pack\tm\bag$ (and are then decorated with flags as for the \LTAM) because $n$ and $m$ play a role only for the reverse translation from $\tarcal$ to $\intcal$ studied in Appendix C, while here they are irrelevant.

The elimination of names enables the use of a \emph{tupled (stackable) environment}: the environment is now a pair of tuples $\vv{\evval_\lvar};\vv{\evval_\svar}$, where $\vv{\evval_\lvar}$ provides values for the \lifted projected variables $\proj_i\lvar$, and $\vv{\evval_\svar}$ for the source ones $\proj_i\svar$, with no need to associate the values of these tuples to variable names via entries of the form $\esub\var\evval$ (as it was the case for the \LTAM). The change is relevant, as the data structure for environments changes from a \emph{map} to a \emph{tuple}, removing the need (and the cost) of creating a map in transition $\tomachbetaev$ and inducing a logarithmic speed-up, as we shall see.

The look up into tupled environments is defined in \reffig{target-tam} and is the only new notion needed in the new transitions.
The implementation theorem is proved following the same schema used for the \TAM and the \LTAM, the details are in \Cref{sect:app-TTAM}.

\begin{toappendix}
\begin{theorem}
\label{thm:tar-ttam-implementation}
The \TTAM implements $\totar$ on prime terms.
\end{theorem}
\end{toappendix}

\paragraph{Actual Implementation of the \TTAM} 
We provide an OCaml implementation of the~\TTAM on GitHub \cite{PPDP25ocaml}, described in \Cref{sect:app-implementation}.
The textual interface asks for a term of the source calculus $\soucal$, which is
translated to the target calculus $\tarcal$ by applying first \lifting and then name elimination, thus passing through the intermediate calculus $\intcal$, as described above.
The obtained $\tarcal$-term is then reduced by the \TTAM until a normal form is reached, if any, and the final
$\tarcal$-term is extracted. The machine state is printed after every step, in ASCII art.

The implementation is not particularly optimized and it does not have a graphical user interface.
It is designed to stay as close as possible to the definitions given in the paper, and to provide evidence
supporting the assumptions of the cost analysis of \refsect{int-complexity}. In particular, we use OCaml arrays for tuples, variables in abstractions,
and \bagts in \wrapts, for achieving $\bigo(1)$ access times.

 \section{Part 3 Preliminaries: Sharing, Size Explosion, and the Complexity of Abstract Machines}
 \label{sect:sharing}
This section starts the third part of the paper, about the time complexity analysis of abstract machines. 
Here, we quickly overview the size explosion problem of the $\l$-calculus as the theoretical motivation for the use of sharing in implementations, as well as the structure of the study of the overhead of abstract machines.

\paragraph{Size Explosion} A well-known issue of the $\l$-calculus is the existence of families of terms whose size  grows \emph{exponentially} with the number of $\beta$-steps. They are usually built exploiting some variant of the duplicator $\delta \defeq \la\var\var\var$. We give an example in $\cbvcal$.  Define:

\vspace*{-.7\baselineskip}
\begin{center}
\begin{tabular}{c|c|c}
\!\!\!\!\makecell{
\textsc{Variant of $\delta$}
\\
$\pi \defeq  \la\var\la\vartwo\vartwo \var\var
$}\!
&
\!\!\!\!$\begin{array}{c@{\hspace{.3cm}} c}
\multicolumn{2}{c}{\textsc{Size explod. family}}
\\
  \tm_0  \defeq  \Id 
  & 
    \tm_{n+1}  \defeq  \pi \tm_n 

\end{array}$\!\!\!\!
&
\!\!\!\!$\begin{array}{c@{\hspace{.3cm}} c}
\multicolumn{2}{c}{\textsc{Exploded results}}
\\
  \tmtwo_0  \defeq  \Id
  &   
  \tmtwo_{n+1}  \defeq  \la\vartwo\vartwo\tmtwo_n \tmtwo_n
\end{array}$
\end{tabular}
\end{center}

\begin{toappendix}
\begin{proposition}[Size Explosion in $\cbvcal$]
\label{prop:size-explosion-plotkin}
Let $n\in \nat$. Then $\tm_n \tobv^n \tmtwo_n$, moreover $\size{\tm_n} = \bigo(n)$, $\size{\tmtwo_n} = 
\Omega(2^n)$, and $\tmtwo_n$ is a value.
\end{proposition}
\end{toappendix}

The proof is in \Cref{sect:app-sharing}.
Size explosion has been extensively analyzed in the study of reasonable cost models---see \cite{DBLP:conf/rta/Accattoli19} for an introduction
---because it suggests that the number $n$ of $\tobv$ steps is not a reasonable time measure for the execution of $\l$-terms: for size exploding families, indeed, it does not even account for the time to write down the normal form, which is of size $\Omega(2^n)$. 

One is tempted to circumvent the problem by tweaking the calculus, with types, by changing the evaluation strategy, restricting to CPS, and 
so on. None of these tweaks works, size explosion can always be adapted: it is an inherent feature of higher-order computations~\cite{DBLP:conf/rta/Accattoli19}.

\paragraph{Sharing for Functions} A solution nonetheless exists: 
it amounts to add a way to \emph{share} sub-terms to avoid their blind duplication during evaluation. 
For size explosion in $\cbvcal$, it is enough to add a simple form of \emph{sub-term sharing} 
by delaying meta-level substitution and avoiding substituting under abstractions.
For instance, evaluating the size exploding term $\tm_n$ above in a variant of $\cbvcal$ where the mentioned sub-term sharing is implemented via explicit substitutions, gives the following normal form, of size \emph{linear}  (rather than exponential) in $n$:
\begin{center}
$\la{\vartwo_1}\vartwo_1 \var_1\var_1 \esub{\var_1}{\la{\vartwo_2}\vartwo_2 \var_2\var_2} \ldots \esub{\var_{n-1}}{\la{\vartwo_n}\vartwo_n \var_n\var_n}\esub{\var_n}{\Id}.$
\end{center}
The explosion re-appears if one unfolds that normal form to an ordinary $\l$-term, but it is now encapsulated in the unfolding.

Abstract machines of the previous sections have environments to implement sub-term sharing and avoid the size explosion due to $\betav$, giving hope for a time complexity lower than exponential.

\paragraph{Parameters for the Time Complexity Analysis of Abstract Machines} 
Given a strategy $\tostrat$ of a calculus $\xcal$ and an abstract machine $\mach$ implementing $\tostrat$ (see \refdef{implem}), the time complexity of $\mach$ is obtained by estimating the cost---when concretely implemented on random access machines (RAMs)---of a run $\exec_\deriv:\compil{\tm_0} \tomach^* \state$ implementing an arbitrary evaluation sequence $\deriv: \tm_0 \tostrat^n\tm_n$ (thus having $\decode\state = \tm_n$) as a function of two parameters:
\begin{enumerate}
\item \emph{Code size}: the size $\size{\tm_0}$ of the initial term $\tm_0$;

\item \emph{Number of $\tostrat$-steps/$\beta$-steps}: the number $n$ of $\tostrat$-steps in $\deriv$. If $\xcal$ has other rules other than $\beta/\betav$, the parameter is often just the number of $\beta/\betav$ steps, which is usually considered the relevant time cost model.
\end{enumerate}

\paragraph{Recipe for Time Complexity} The way 
the time complexity of an abstract machine is established tends to follow 
the same schema:
\begin{enumerate}
\item \emph{Number of overhead transitions}: bounding the number of overhead transitions as a function of $\size{\tm_0}$ and $n$ (which by the principal matching property of implementations---see \refdef{implem}---is enough to bound the length of $\exec_\deriv$);
\item \emph{Cost of single transitions}: bounding the cost of single transitions, which is typically constant or  depends only on $\size{\tm_0}$;
\item \emph{Total cost}: inferring the total cost of a run $\exec$ 
by multiplying the number of steps of each kind of transition for their cost, and summing all the obtained costs.
\end{enumerate}

The key tool for such an analysis is the \emph{sub-term property}, an \emph{invariant} of abstract machines stating that some of the terms in a reachable state are sub-terms of the initial term. 
This allows one to develop bounds with respect to the size $\size{\tm_0}$ of the initial term $\tm_0$.

\paragraph{Time Complexity of the \LAM} 
The \LAM verifies a sub-term invariant, and its time complexity follows a well-known schema in the literature, closely inspected by \citet{DBLP:conf/ppdp/AccattoliB17} for call-by-name and call-by-need, and that smoothly adapts to \cbv. 
Consider an evaluation $\deriv \colon \tm_0 \tobv^n \tm_n$ in $\cbvcal$. 
About the first point of the recipe, the  bound on overhead transitions  is $\bigo((n+1)\cdot\size{\tm_0})$, that is, bilinear. 
Additionally, if one takes \emph{complete} evaluations, that is, for which $\tm_n=\val$ is a value, then the bound lowers to $\bigo(n)$. Such an independence from the initial term is due to the fact that whether a term is a value can be checked in $\bigo(1)$ in $\cbvcal$, by simply checking whether the top-most constructor is an abstraction.

For the second point of the recipe, the cost of single transitions of the \LAM depends on the data structures used for local environments, 
as discussed by \citet{DBLP:conf/ppdp/AccattoliB17}. With flat environments, the cost of manipulating them is $\bigo(\size{\tm_0})$ (because of the duplication of $\env$ in $\tomachseaone$), giving a total cost of $\bigo(n\cdot \size{\tm_0})$ for complete runs. With shared environments, the best structures manipulate them in $\bigo(\log\size{\tm_0})$, giving a total cost of $\bigo(n\cdot \log\size{\tm_0})$ for complete runs. Thus, shared local environments are faster, but they are optimized for time and inefficient  with respect to space, as they prevent some garbage collection to take place. Here, we take as reference \emph{flat} local environment, which induce the same overall $\bigo(n\cdot \size{\tm_0})$ overhead as global environments and enable a better management of space (not discussed here, see \cite{DBLP:conf/lics/AccattoliLV22} instead).

\paragraph{Notions of Flatness for Local Environment} Let us be precise on a subtle point about local environments and their flatness. Local environments are defined by mutual recursion with m-closures, and various notions of sharing and flatness are possible, as one can share environments, or m-closures, both, or none of them. Sharing both is essentially the same as sharing only environments. 

Sharing m-closures rather than environments is what we above called \emph{flat environments}. For instance, if $\env = \esub{\var_1}{\clos_1}\cons\esub{\var_2}{\clos_2}\cons \emptylist$ then $\clos_1$ and $\clos_2$ are shared so that the concrete representation of $\env$ is $\env = \esub{\var_1}{p_1}\cons\esub{\var_2}{p_2}\cons \emptylist$ where $p_1$ and $p_2$ are pointers to $\clos_1$ and $\clos_2$. Copying $\env$ then means copying $\esub{\var_1}{p_1}\cons\esub{\var_2}{p_2}\cons \emptylist$, \emph{without} recursively copying the structure of $\clos_1$ and $\clos_2$. 

The \emph{super flat environments} obtained by removing sharing for \emph{both} environments and m-closures are studied in \cite{DBLP:conf/lics/AccattoliLV22}, where it is shown that the overhead of abstract machines becomes \emph{exponential}. 

\paragraph{No Sharing for Stackable Environments} An interesting aspect of the new stackable environments is that they need no sharing.  Indeed, the Int TAM and the Target TAM never duplicate their stackable environments. 
In particular, when they discard the current environment $\env$ in transition $\tomachseasevenev$, it can be collected.

\paragraph{De Bruijn Indices Do Not Change the Overhead} It is well-known that, representing flat environments as arrays and variables with de Bruijn indices, one can look-up environments in $\bigo(1)$ rather than in $\bigo(\log\size{\tm_0})$, which is the best that one can do with (some ordered domain of) names. However, the $\bigo(\size{\tm_0})$ cost of copying flat environments in transition $\tomachseaone$ dominates, so that turning to de Bruijn indices does not change the overall $\bigo(n\cdot \size{\tm_0})$ overhead.

 \label{sect:background-complexity}

\paragraph{Size Exploding Tuples} In $\soucal$, there is a new form of size explosion, due to tuples, requiring another form of sharing. To our knowledge, the view provided here is novel. Set $\Id \defeq \la\varthree\varthree$ and:

\vspace{-.8\baselineskip}
\begin{center}
\begin{tabular}{c|c|c}
\makecell{\textsc{Variant of $\delta$}
\\
$\tau  \defeq  \la\var\tuple{\var,\var}$}
&
\!\!$\begin{array}{c@{\hspace{.3cm}} c}
\multicolumn{2}{c}{\textsc{Size explod. family}}
\\
 \tmthree_{0}  \defeq  \Id 
 & \tmthree_{n+1} \defeq \tau \tuple{\tmthree_{n}}
\end{array}$\!\!\!\!
&
\!\!\!\!
$\begin{array}{c@{\hspace{.3cm}} c}
\multicolumn{2}{c}{\textsc{Exploded results}}
\\
\tmfour_{0}  \defeq  \Id
 &\tmfour_{n+1}  \defeq  \tuple{\tmfour_{n},\tmfour_{n}}
 \end{array}$
\end{tabular}
\end{center}

\begin{toappendix}
\begin{proposition}[Size explosion of tuples]
\label{prop:size-explosion-tuples}
Let $n\in \nat$. Then $\tmthree_n \tobv^n \tmfour_n$, moreover $\size{\tmthree_n} = \bigo(n)$, $\size{\tmfour_n} = 
\Omega(2^n)$, and $\tmfour_n$ is a value.
\end{proposition}
\end{toappendix}

\paragraph{Sharing for Tuples} To avoid the size explosion of tuples, another form of sharing is used. The idea is the same as  for functions: forms of size explosion are circumvented by forms of sharing designed to limit the substitution process. The key point is that tuples should never be copied, only \emph{pointers} to them should be copied, thus representing $\tmfour_n$ above using \emph{linear} (rather than exponential) space in $n$, as follows (where the $p_i$ are \emph{pointers} and $\esub{p_i}{\val}$ are \emph{heap entries}):
\begin{center}$
\tuple{p_1,p_1} \esub{p_1}{\tuple{p_2,p_2}} \ldots \esub{p_n}{\tuple{p_n,p_n}} \esub{p_n}{\Id}
$\end{center}

The abstract machines of the previous sections have environments for sub-term sharing (needed for $\betav$) but they do \emph{not explicitly} handle tuple sharing. The reason is 
practical: explicitly handling tuple sharing would require a treatment of \emph{pointers} and a \emph{heap} (i.e., a further global environment) 
and more technicalities.
Our machines, however, are meant to be concretely implemented with tuple sharing, 
as in our OCaml \mbox{implementation of the \TTAM \cite{PPDP25ocaml}}. 
\section{Part 3: Complexity of the \TAM, or, Tuples Raise the Overhead}
\label{sect:source-complexity}
Here, we develop the time complexity analysis of the \TAM, stressing the novelty of tuples (see \Cref{sect:app-source-complexity} for proofs).

The \TAM verifies the following \emph{sub-term invariant}.

\begin{toappendix}
\begin{lemma}[Sub-term invariant]
	\label{l:TAM-sub-term-invariant}
Let $\state$ be a \TAM reachable state 
from the initial state $\compil{\tm}$.
\begin{enumerate}
\item $\tmtwo$ is a sub-term of $\tm$ for every \mclosure of shape $\nvlab{\pair\tmtwo\env}$ or $\pair{\tuple{\mydots,\tmtwo,\mydots,\machctxhole,\vv\clos}}\env$ in $\state$.
\item $\la\tuvar\tmtwo$ is a sub-term of $\tm$ for every \mclosure $\evlab{\pair{\la\tuvar\tmtwo}\env}$ in $\state$. 
\end{enumerate}
\end{lemma}
\end{toappendix}

\paragraph{Tuples are Not Sub-Terms.} Point 2 only concerns evaluated \mclosures containing abstractions, and \emph{not} evaluated tuples. Consider $\tm \defeq(\la\var\la\vartwo\tuple{\var,\vartwo})\tuple{\Id}\tuple{\delta} \tobv (\la\vartwo\tuple{\Id,\vartwo})\tuple{\delta} \tobv \tuple{\Id,\delta}$ in $\soucal$ and note that $\tuple{\Id,\delta}$ is \emph{not} a sub-term of $\tm$. 
The run of the \TAM on $\tm$ produces a \mclosure $\pair{\tuple{\evsym\Id,\evsym\delta}}\env$ for some $\env$. 
The leaves of the tree-structure of an evaluated tuple are abstractions ($\Id$ and $\delta$ in the example), which are initial sub-terms, but they might be arranged in ways that were not present in the initial term.

A consequence of this fact is that when the \TAM starts evaluating a non-empty tuple $\tuple{\mydots,\tm_n}$ with transition $\tomachseathreenv$, by adding $\tuple{\mydots,\machctxhole}$ to the stack, it has to \emph{allocate a new tuple} on the heap for $\tuple{\mydots,\machctxhole}$ (which has a cost, discussed below). This never happens in absence of tuples, that is, in the \LAM of \reffig{machine_local_cbv}. More precisely, the \LAM does not copy any code, but it has to allocate new pointers to local environments, when they are extended by $\tomachbeta$.

\paragraph{Step 1 of the Recipe: Number of Transitions and Overhead Measure} For establishing a bound on overhead transitions, we first factor some of them ($\tomachhole{\evsym sea_{1,3}}$) out by simply noticing that they are enabled and thus bound by some others ($\tomachhole{\nvsym sea_{1,3}}$). Actually, the same is true also for the principal transition $\tomachprojev$, bounded by $\tomachseatwonv$.
\begin{toappendix}
\begin{lemma}[Transition match]
	\label{l:TAM-blue-bound}
Let $\exec$ be a \TAM run. Then $\sizep\exec{\proj,\evsym sea_{1,3}}  \leq \sizep\exec{\nvsym sea_{1,2,3}}$.
\end{lemma}
\end{toappendix}
For the other transitions, we use a measure $\omeas\cdot$, defined in \reffig{ovh-measure-tam} together with the size $\size\tm$ of $\soucal$ terms, 
which we use to derive a bilinear bound on overhead/projection transitions (\refpropeq{TAM-number-of-trans}).

\begin{figure}[t!]
	
	\vspace*{-.2\baselineskip}
\centering
\scalebox{0.93}{
\setlength{\arraycolsep}{3pt}
\begin{tabular}{c}
\!\!\!\!$	\begin{array}{rll@{\hspace{.3cm}} rll@{\hspace{.3cm}} rll}
\multicolumn{9}{c}{\textsc{Size of $\soucal$ terms}}
\\
		\size{\var} &\!\!\!\defeq\!\!\!& 1 
		&
		\size{\la{\tuvar}\tm} &\!\!\!\defeq\!\!\!& \size{\tm} + \norm{\tuvar} +1 
		&
		\size{\tm \tmtwo} &\!\!\!\defeq\!\!\!& \size{\tm} + \size{\tmtwo} + 1
		\\
		\size{\proj_{i}\tm} &\!\!\!\defeq\!\!\!& \size{\tm} + 1
		&
		\multicolumn{4}{c}{\size{\tuple{\tm_1, \dots, \tm_n}} \defeq
		n + \textstyle\sum_{i=1}^n \size{\tm_i}}
\end{array}$
\\
\hline
$\begin{array}{r@{\hspace{.2cm}} r@{\ }l@{\ }l@{\hspace{.35cm}} r@{\ }l@{\ }l@{\hspace{.35cm}} r@{\ }l@{\ }l}
\multicolumn{10}{c}{\textsc{Overhead measure for the \TAM}}
\\
\textsc{\!\!\!\!\!\!\!Con. stack entries $\stacke$}
&
\omeas{\clos} & \defeq & 0
&
\omeas{\nvlab{\pair\tm\env}} & \defeq & \size\tm
&
\omeas{\proj_{i}} & \defeq & 0
\\
&\multicolumn{9}{c}{\!\!\!\!\omeas{\pair{\tuple{\nvsym\tm_{1},\mydots,\nvsym\tm_{n},\machctxhole,\tuv\clos}}\env}
\defeq n + \sum_{i=1}^n\omeas{\tm_{i}}}
\\
\textsc{Stacks}
&\omeas\emptystack & \defeq & 0
& 
\multicolumn{6}{c}{\omeas{\stacke\cons\stack} \defeq \omeas\stacke + \omeas \stack}
\\
\textsc{States} &
\multicolumn{9}{c}{\omeas{\pairstate{\laxclos}{\stack}} \defeq \omeas\laxclos + \omeas \stack}
\end{array}$
\end{tabular}
}

\vspace*{-.7\baselineskip}
\caption{Size $\size{\cdot}$ of $\soucal$ terms and overhead measure $\omeas{\cdot}$.} 
\label{fig:ovh-measure-tam}

\vspace*{-.7\baselineskip}
\end{figure}

\begin{toappendix}
\begin{lemma}[Overhead measure properties]
	\label{l:TAM-overhead}
Let $\compil\tm \totam^{*} \state$ a \TAM run and $\state \tomachhole{a} \statetwo$.
\begin{enumerate}
\item
if $a = \evsym\betav$ then $\omeas\statetwo \leq \omeas\state + \size{\tm}$;
\item
if $a \in\set{\nsubsym, \nvsym sea_{1-5}, \evsym sea_{6}}$ then $\omeas\statetwo < \omeas\state$;
\item 
if $a \in\set{\evsym sea_{1,3}, \evsym \projsym}$ then $\omeas\statetwo = \omeas\state$.
\end{enumerate}
\end{lemma}
\end{toappendix}

\begin{toappendix}
\begin{proposition}[Bilinear bound on the number of 
	transitions]
\label{prop:TAM-number-of-trans}
Let $\tm \!\in\! \soucal$ be closed. 
If $\exec \colon \!\compil{\tm} \!\totam^{*}\! \state$ then \mbox{$\size\exec \!\in\! \bigo\big((\sizebeta\exec\!\!+\!1)\!\cdot\!\size\tm\big)$.}
\end{proposition}
\end{toappendix}
\paragraph{Tuples Raise the Overhead.} 
\refpropeq{TAM-number-of-trans} shows that projection transitions are also bi-linear. 
In \Cref{sect:sharing}, we mentioned that, without tuples, the bound improves to $\bigo(\sizebeta\exec)$ if one considers complete runs (that is, runs ending on values). With tuples, there is no such improvement. 
Indeed, even just checking that the initial term is actually a value $\val$ takes time $\bigo(\size\val)$ with tuples: if $\val$ is a tree of tuples, the \TAM has to visit the tree and check that all the leaves are abstractions; in absence of tuples the check instead costs~$\bigo(1)$.

\paragraph{Step 2 of the Recipe: Cost of Single Transitions} To obtain fine bounds with respect to the initial term, we introduce the notions of width and height of a term $\tm$, both bounded by the size $\size\tm$ of $\tm$.

\begin{definition}[Width, height]
The \emph{width} $\wid\tm\in\nat$ of $\tm\in\soucal$ is the maximum length of a tuple or of a sequence of variables in $\tm$. The \emph{height} $\hg\tm\in \nat$ of $\tm\in\soucal$ is the maximum number of bound variables of $\tm$ in the scope of which a sub-term of $\tm$ is contained.
\end{definition}

As discussed after 
\Cref{l:TAM-sub-term-invariant}, $\tomachseathreenv$ has to allocate a new tuple, thus its cost seems to depend on  $\wid\tm$, and similarly for $\tomachbetaev$. 
However, the price related to tuples can be considered as absorbed by the cost of search (by changing the multiplicative constant), since the new tuple of $\tomachseathreenv$ is then traversed, if the run is long enough, and the one of $\tomachbetaev$ was traversed before the transition. 
Therefore, 
if we consider \emph{complete} runs (i.e. ending on final states), $\tomachseathreenv$ and $\tomachbetaev$ have \emph{amortized} cost independent~of~$\wid\tm$. 

Transitions $\evsym\seasym_{1}$, $\evsym\seasym_{3}$, and $\evsym\seasym_{6}$ duplicate $\env$, the length of which is bound by $\hg\tm$. With flat environments, this costs $\bigo(\hg\tm)$. Transition $\tomachsubnv$ has to look-up the environment, the cost of which is $\bigo(\hg\tm)$. De Bruijn indices or an ordered domain of names might improve the cost of look-up, but at no overall advantage, because of the dominating cost of duplicating environments for $\evsym\seasym_{1}$, $\evsym\seasym_{3}$, and $\evsym\seasym_{6}$. All other transitions have constant cost, assuming that accessing the $i$-th component of a tuple (needed for $\tomachprojev$) takes constant time. The next proposition sums it up.

\begin{proposition}[Cost of single transitions]
\label{prop:TAM-cost-trans}
Let $\run: \compil\tm \totam^{*} \state$ be a complete \TAM run. A  transition $\tomachhole{a}$ of $\run$ has cost $\bigo(\hg\tm)$ if $a \in\set{\evsym\seasym_{1},\evsym\seasym_{3},\evsym\seasym_{6}, \nvsym sub}$, and $\bigo(1)$ otherwise.
\end{proposition}

\paragraph{Step 3 of the Recipe: Total Complexity} By simply multiplying the number of single transitions (\refprop{TAM-number-of-trans}) for their cost (\refprop{TAM-cost-trans}), we obtain the complexity of the \TAM.
\begin{toappendix}
\begin{theorem}
\label{thm:TAM-ram-implem}
Let $\tm\in\soucal$ be closed and $\exec:\compil\tm \totam^{*} \state$ be a complete \TAM run. Then, $\exec$ can be implemented on RAMs in time $ \bigo\big( (\sizebeta\exec +1) \cdot\size\tm \cdot \hg\tm \big)$.
\end{theorem}
\end{toappendix}
If one flattens $\hg\tm$ as $\size\tm$, the complexity of the \TAM is $ \bigo\big( (\sizebeta\exec +1) \cdot\size\tm^2\big)$, that is, quadratic in $\size\tm$, while in absence of tuples---that is, for the \LAM---it is linear in $\size\tm$.

\section{Part 3: Complexity of the \TTAM, or, Closure Conversion Preserves the Overhead}
\label{sect:int-complexity}
In this section, we adapt the time complexity analysis of the \TAM to an analysis of the \TTAM (skipping the less efficient \LTAM), and then connect source and target by considering the impact of closure conversion on the given analysis.

The \TTAM has a sub-term invariant, expressed compactly thanks to our flags. In particular the part about abstractions of the sub-term invariant for the \TAM (\reflemma{TAM-sub-term-invariant}) is here captured by having flagged the body of unevaluated \wrapts~with~$\nvsym$.

\begin{toappendix}
\begin{lemma}[Sub-term invariant]
	\label{l:LTAM-sub-term-invariant}
Let $\state$ be a \TTAM reachable state 
from the initial state $\compil{\tm}$.
Then  $\tmtwo$ is a sub-term of $\tm$ for every non-evaluated term $\nvsym\tmtwo$ in $\state$.
\end{lemma}
\end{toappendix}

\paragraph{Step 1: Number of (Overhead) Transitions} The bound on the number of (overhead) transitions is obtained following the same reasoning used for the \TAM. The new transition $\tomachseasevenev$ is part of the transitions that are factored out, since each $\tomachseasevenev$ transition is enabled by a $\tomachbetaev$ transition, which adds an entry to the activation stack. We also use an overhead measure (in \Cref{sect:app-int-complexity}) which is a direct adaptation to the \TTAM of the one given for the \TAM. Note indeed that the measure ignores environments, which are the main difference between the two machines.

\begin{toappendix}
\begin{proposition}
\label{prop:LTAM-number-of-trans}
Let $\tm \in \tarcal$ be closed. 
If $\exec \colon \compil{\tm} \tottam^{*} \state$ then $\size\exec \in \bigo\big((\sizebeta\exec+1)\cdot\size\tm\big)$.
\end{proposition}
\end{toappendix}

\paragraph{Step 2: Cost of Single Transitions}  For all transitions of the \TTAM but $\tomachsubvev$, $\tomachsubwev$, and $\tomachseasevenev$ the cost is the same. 
For $\tomachsubvev$, the cost is now $\bigo(1)$ since tupled environments have $\bigo(1)$ access time via indices 
(in the \LTAM its cost is instead $\bigo(\hg\tm)$).

For the new transition $\tomachsubwev$, 
$\bigo(\wid\tm)$ look-ups in the environment are needed. 
Because of tupled environments, each look-up costs $\bigo(1)$. Thus, the cost seems to depend on $\wid\tm$, but---reasoning as for $\tomachhole{\evsym\betav}$ and $\tomachhole{\nvsym sea_{3}}$ in \Cref{sect:source-complexity}---one can amortize 
it with the cost of search in complete runs. That is, we shall consider $\tomachsubwev$ to have constant cost. The new transition $\tomachseasevenev$ has constant cost as well. 
Summing up, we get 
one of the insights mentioned in the introduction: the amortized cost of all single transitions is $\bigo(1)$.

\begin{proposition}[Cost of single transitions]
\label{prop:LTAM-cost-trans}
Let $\run:\compil\tm \tottam^{*} \state$ be a complete \TTAM run. Every transition of $\run$ costs $\bigo(1)$.
\end{proposition}

\paragraph{Step 3: Total Complexity} Multiplying the number of single transitions (\refpropeq{LTAM-number-of-trans}) for their cost (\refpropeq{LTAM-cost-trans}), we obtain the time complexity of the \TTAM, which seem better than for the \TAM. After the theorem we discuss why it is not \mbox{necessarily~so}.
\begin{theorem}
\label{thm:TTAM-ram-implem}
Let $\tm\in\tarcal$ be closed and $\exec:\compil\tm \tottam^{*} \state$ be a complete \TTAM run. Then, $\exec$ can be implemented on RAMs in time $ \bigo\big( (\sizebeta\exec +1) \cdot\size\tm \big)$.
\end{theorem}

\paragraph{Factoring in the Size Growth of \Lifting} To complete the analysis, 
take into account that the term $\tm\in\tarcal$ on which the \TTAM is run is meant to be the closure conversion (that is, \lifting + name elimination) of a term $\tmtwo\in\soucal$. While name elimination does not affect the size of terms, \lifting $\lali{\cdot}$ does (proof in \Cref{sect:app-int-complexity}). 

\begin{toappendix}
\begin{lemma}[\Lifting size growth bound]
	\label{l:lifting-growth-bound}~
	\begin{enumerate}
	\item If $\tm\in\soucal$ then $\size{\lali\tm}\in \bigo(\hg\tm\cdot\size\tm)$.
	\item There are families of terms $\{\tm_n\}_{n\in\nat}$ for which $\hg{\tm_n} = \Theta(\size{\tm_n})$, so that $\size{\lali{\tm_n}}\in \Theta(\size{\tm_n}^2)$.
	\end{enumerate}
\end{lemma}
\end{toappendix}


We can now instantiate the bounds for running the \TTAM on the closure conversion of a $\soucal$ term, obtained by substituting the bounds in \reflemma{lifting-growth-bound}.2 in \refthm{TTAM-ram-implem}. We end up obtaining the same complexity as for the Source TAM (\refthm{TAM-ram-implem}), despite the bound here being the outcome of a different reasoning.

\begin{theorem}
\label{thm:TTAM-ram-implem-with-lifting}
Let $\tm\in\soucal$ be closed and $\exec:\compil{(\clc{\lali\tm}\emptylist\emptylist)} \tottam^{*} \state$ be a complete \TTAM run. Then, $\exec$ can be implemented on RAMs in time $ \bigo\big( (\sizebeta\exec +1) \cdot\size\tm\cdot\hg\tm \big)$.
\end{theorem}

%
%

\section{Related Work and Conclusions}
\label{sect:related-work}

\paragraph{Related work.} 
The seminal work of \citet{DBLP:conf/popl/MinamideMH96} uses existential types to type closure conversion, and other works study the effect of closure conversion on types, as well as conversion towards typed target languages~\cite{DBLP:journals/entcs/MorrisettH97,DBLP:conf/icfp/AhmedB08,DBLP:conf/pldi/BowmanA18}.
This line of work explores the use of types in compilation which can span the entire compiler stack, from $\l$-calculus to (typed) assembly~\cite{DBLP:journals/toplas/MorrisettWCG99}. 

Appel and co-authors have studied closure conversion, its efficient variants, and its \emph{space safety} \cite{DBLP:conf/popl/AppelJ89,DBLP:conf/lfp/ShaoA94,DBLP:journals/toplas/ShaoA00,DBLP:journals/pacmpl/Paraskevopoulou19}.
Unlike them, we address a direct style $\l$-calculus, without relying on a CPS. 
We are however inspired by \cite{DBLP:journals/pacmpl/Paraskevopoulou19} in studying flat environments, 
also known as \emph{flat closures}. Optimizations of flat closures are studied~in~\cite{DBLP:conf/icfp/KeepHD12}.

Closure conversion of only some abstractions, and of only some of their free variables is studied by \citet{DBLP:conf/popl/WandS94}. Closure conversion has also been studied in relation to graphical languages \cite{DBLP:journals/entcs/SchweimeierJ99},
non-strict languages \cite{DBLP:conf/pepm/SullivanDA21}, formalizations \cite{DBLP:conf/popl/Chlipala10,DBLP:conf/esop/WangN16}, extended to mutable state by \citet{DBLP:conf/ppdp/MatesPA19} and to a type-preserving transformation for dependent types by \citet{DBLP:conf/pldi/BowmanA18}. 

Many of the cited works study various aspects of the efficiency of closure 
conversion. 
As said in the introduction, our concerns here are orthogonal, as we are rather interested in the asymptotic overhead of the machines with respect to flat 
closure~conversion.

\citet{DBLP:conf/ppdp/SullivanDA23} study a call-by-push-value $\l$-calculus where converted and non-converted functions live together, considering also an abstract machine. They do not study the complexity of the machine, nor the new notions of environments studied here.

For abstract machines, we follow Accattoli and co-authors, see for instance \cite{DBLP:conf/icfp/AccattoliBM14,DBLP:journals/scp/AccattoliG19,DBLP:conf/ppdp/AccattoliCGC19}. Another framework is \cite{DBLP:books/daglib/0023092}.

Work orthogonal to our concerns is the \emph{derivation} of abstract machines using closure conversion as a step in the process, along with CPS transformation and defunctionalization~\cite{DBLP:conf/ppdp/AgerBDM03}.

\label{sect:conclusions}

\paragraph{Conclusions.} 
We study the relationship between closure conversion and abstract machines in 
probably the simplest possible setting, an extension  with tuples of Plotkin's untyped call-by-value $\l$-calculus, and with respect to the simple notion of flat environments.  
Our starting point is 
to decompose closure conversion in two sub-transformations, dubbed  \emph{\lifting} and \emph{name elimination}, turning the source calculus into a target calculus, via \mbox{an intermediate one.}

Each calculus is then paired with a variant of the \emph{tupled abstract machine} (TAM). The \TAM has machine-closures and local environments, while the \LTAM and the \TTAM have forms of converted closures. Moreover, they exploit the invariants enforced by the transformations,  adopting new, better behaved forms of environments, namely \emph{stackable} and \emph{tupled environments}.

We give proof of correctness---under the form of termination-preserving strong bisimulations---for \lifting and name elimination, as well as implementation theorems for every machine with respect to its associated calculus. In particular, the proof technique for the correctness of closure conversion is new and simple.

Lastly, we study the time complexity of the abstract machines, showing that flat closure conversion reshuffles the costs, 
lowering the dependency on the initial term, while at the same time increasing the size of the initial term, ending up with the same~complexity.


\bibliographystyle{ACM-Reference-Format}
\bibliography{\macrospath/biblio}

\newpage
\appendix
\onecolumn
\section*{Appendix}
There is a section of the Appendix for every section of the paper, containing the omitted proofs plus, possibly, some auxiliary notions.

\section{Proofs and Auxiliary Notions of \Cref{sect:source-calculus} (The Source Calculus $\soucal$)}
\label{sect:app-source-calculus}

\begin{lemma}[Evaluation contexts compose]
\label{l:source-evctxs-compose}
Let $\evctx$ and $\evctxtwo$ be source evaluation contexts. Then $\evctxp\evctxtwo$ is a source evaluation context.
\end{lemma}

\begin{proof}
Straightforward induction on $\evctx$.
\end{proof}

\begin{lemma}[Determinism]
	\label{l:source-determinism}
	If $\tm \tosou \tmtwo$ and $\tm \tosou \tmthree$ then $\tmtwo = \tmthree$ and $\tm$ is not a value.
\end{lemma}

\begin{proof}
By induction on $\tm$.
\end{proof}

\gettoappendix{l:source-harmony}

\begin{proof}
	If $\tm \tosou \tmtwo$, the fact that $\tmtwo$ is closed and clash-free follows from the hypothesis on $\tm$.
Let us prove, by induction on $\tm$, that $\tm$ is a value or $\tm \tosou \tmtwo$ for some $\tm$. Cases:
\begin{itemize}
\item \emph{Variable}: impossible, because $\tm$ is closed.

\item \emph{Abstraction}: then $\tm$ does not reduce.

\item \emph{Projection}, that is, $\tm = \proj_{i} \tmtwo$. 
By \ih, $\tmtwo$ is either a value or it reduces. 
If $\tmtwo$ is a value, by clash-freeness, it must be a tuple $\tmtwo = \vv\val$ of length $\geq i$;
so, $\tm \toproj \val_{i}$. 
If instead $\tmtwo$ reduces then $\tm$ reduces, because $\proj_{i}\ctxhole$ is an evaluation context.

\item \emph{Application}, that is, $\tm = \tmtwo \tmthree$. By \ih, $\tmthree$ is either a value or it reduces. If $\tmthree$ reduces then $\tm$ reduces because $\tmtwo\ctxhole$ is an evaluation context. 
If $\tmthree$ is a value $\val$ then we apply the \ih to $\tmtwo$. 
If $\tmtwo$ reduces then $\tm$ reduces because $\ctxhole\val$ is an evaluation context. Otherwise $\tmtwo$ is a value $\valtwo$. By clash-freeness, $\valtwo$ is an abstraction $\la{\tuvar}\tmtwo'$ and $\val = \vv\valthree$ with $\norm{\tuvar} = \norm{\vv\valthree}$. Then $\tm \tobv \tmtwo'\isub{\tuvar}{\vv\valthree}$.

\item \emph{Tuple}, that is, $\tm = \vv\tmtwo$. Let $i \in \nat$ be the smallest index $i \leq \norm{\vv\tmtwo}=k$ such that $\tmtwo_{i}$ is not a value and $\tmtwo_{j}$ is a value for all $j$ satisfying $i <j \leq k$. If $i = 0$ then $\tm$ is a value. Otherwise, by \ih $\tmtwo_{i}$ reduces. Then $\tm$ reduces, since $\tuple{\tmtwo_{1},\mydots,\tmtwo_{i-1},\ctxhole,\val_{i+1},\mydots,\val_{k}}$ is an evaluation context.
\qedhere
\end{itemize}
\end{proof}

\section{Proofs and Auxiliary Notions of \refsect{intermediate-calculus} (The Intermediate Calculus $\intcal$ and \Lifting)}
\label{sect:app-intermediate-calculus}

	\begin{lemma}[Determinism]
		\label{l:intermediate-determinism}
		Let $\tm\in\intcal$. If $\tm \toint \tmtwo$ and $\tm \toint \tmthree$ then $\tmtwo = \tmthree$ and $\tm$ is not a value.
	\end{lemma}

\begin{proof}
By induction on $\tm$.
\end{proof}

\begin{lemma}[Evaluation contexts compose]
\label{l:intermediate-evctxs-compose}
Let $\evctx$ and $\evctxtwo$ be intermediate evaluation contexts. Then $\evctxp\evctxtwo$ is an intermediate evaluation context.
\end{lemma}

\begin{proof}
Straightforward induction on $\evctx$.
\end{proof}

\begin{definition}[Clashes]
A well-formed term $\tm\in\intcal$ is a \emph{clash} if it has shape $\evctxp{\tmtwo}$ where $\tmtwo$ has one of the following forms, called a \emph{root clash}:
\begin{itemize}
\item \emph{Clashing projection}: $\tm = \proj_{i} \val$ and ($\val$ is not a tuple or it is but $\norm\val <i$);
\item \emph{Clashing \wrapt}:  $\tm = \pack{ \dabsv\vartwo\var\tmtwo }{ \bag }  \val_2$ and ($\val_2$ is not a tuple or it is but $\norm{\tuvar} \neq \norm{\vv{\val_2}}$);
\item \emph{Clashing tuple}: $\tm = \vv\tmtwo  \tmthree$.
\end{itemize}
\end{definition}

\begin{definition}[Simultaneous substitution]
If $\fv\tm \subseteq \tuvartwo\vdisjoint\tuvar$ then the simultaneous substitution $\tm\isubp\tuvartwo\vvvalnone\tuvar\vvvalntwo$ is defined as follows:
\begin{center}
	$\begin{array}{rll}
		\multicolumn{3}{c}{\textsc{Simultaneous substitution}}
		\\
		\vartwo_i\isubp\tuvartwo\vvvalnone\tuvar\vvvalntwo & \defeq & (\val_{1})_i
		\\
		\var_i\isubp\tuvartwo\vvvalnone\tuvar\vvvalntwo & \defeq & (\val_{2})_i
		\\
		(\proj_i\tm)\isubp\tuvartwo\vvvalnone\tuvar\vvvalntwo & \defeq & \proj_i\tm\isubp\tuvartwo\vvvalnone\tuvar\vvvalntwo
		\\
		\tuple{\tm_1,\mydots,\tm_n} \isubp\tuvartwo\vvvalnone\tuvar\vvvalntwo & \defeq & \tuple{\tm_1\isubp\tuvartwo\vvvalnone\tuvar\vvvalntwo,\mydots,\tm_n\isubp\tuvartwo\vvvalnone\tuvar\vvvalntwo} 
		\\
		(\tm\,\tmtwo)\isubp\tuvartwo\vvvalnone\tuvar\vvvalntwo & \defeq & \tm\isubp\tuvartwo\vvvalnone\tuvar\vvvalntwo\,
		\vv\tmtwo\isubp\tuvartwo\vvvalnone\tuvar\vvvalntwo
		\\
		\pack{\dabsv\vartwo\var\tm}\bag \isubp\tuvartwo\vvvalnone\tuvar\vvvalntwo & \defeq & \pack{\dabsv\vartwo\var\tm}{\bag\isubp\tuvartwo\vvvalnone\tuvar\vvvalntwo}
	\end{array}$
	\end{center}
\end{definition}

\gettoappendix{l:intermediate-harmony}
\begin{proof}
	If $\tm \tollcbv \tmtwo$, the fact that $\tmtwo$ is closed and clash-free follows from the hypothesis on $\tm$.
Let us prove, by induction on $\tm$, that $\tm$ is a value or $\tm \tollcbv \tmtwo$ for some $\tmtwo\in\intcal$. Cases:
\begin{itemize}
\item \emph{Variable}: impossible, because $\tm$ is closed.

\item \emph{\Wrapts}: then $\tm=\pack{\dabsv\vartwo\var\tmthree}{\bag}$ does not reduce and it is a value.

\item \emph{Projection}, that is, $\tm = \proj_{i} \tmtwo$. 
By \ih, $\tmtwo$ is either a value or it reduces. Two cases:
\begin{itemize}
\item \emph{$\tmtwo$ is a value}. By clash-freeness, it must be a tuple $\tmtwo = \vv\val$ of length $\geq i$. Therefore,
 $\tm \tollproj \val_{i}$. 
 \item \emph{$\tmtwo$ reduces}. Then $\tm$ reduces, because $\proj_{i}\ctxhole$ is an evaluation context. 
\end{itemize}

\item \emph{Tuple}, that is, $\tm = \vv\tmtwo$. Let $i \in \nat$ be the smallest index $i \leq \norm{\vv\tmtwo}=k$ such that $\tmtwo_{i}$ is not a value and $\tmtwo_{j}$ is a value for all $j$ satisfying $i <j \leq k$. Two cases:
\begin{itemize}
\item \emph{$i = 0$}. Then $\tm$ is a value. 
 \item \emph{$i>0$}. Then $\tmtwo_i$ is not a value and by \ih it reduces. Then $\tm$ reduces, since $\tuple{\tmtwo_{1},\mydots,\tmtwo_{i-1},\ctxhole,\val_{i+1},\mydots,\val_{k}}$ is an evaluation context.
\end{itemize}

\item \emph{Application}, that is, $\tm = \tmtwo   \tmthree$. By \ih, there are two cases:
\begin{enumerate}
\item \emph{$\tmthree$ reduces}. Then $\tm$ reduces because $\tmtwo \ctxhole$ is an evaluation context. 
\item \emph{$\tmthree$ is a value $\valnone$}. Then we apply the \ih to $\tmtwo$. Two sub-cases:
\begin{enumerate}
\item \emph{$\tmtwo$ reduces}. Then $\tm$ reduces because $\ctxhole \valnone$ is an evaluation context. 
\item \emph{$\tmtwo$ is a value $\val$}. By clash-freeness, well formedness, and \ih, $\val$ is a \wrapt $\pack{\dabsv\vartwo\var\tmfour}{\vv\valntwo}$ with $\norm{\tuvartwo} = \norm{\vv\valntwo}$ and $\valnone$ is a tuple verifying $\norm{\tuvar} = \norm{\valnone}$. Then $\tm \toibv \tmfour\isubp{\tuvartwo}{\vv\valntwo}{\tuvar}{\vv\valnone} $. \qedhere

\end{enumerate}
\end{enumerate}
\end{itemize}
\end{proof}

\gettoappendix{l:lali-properties}
\begin{proof}
Straightforward inductions.
\end{proof}

\begin{lemma}[Basic properties of the reverse translation]
\label{l:rev-transl-properties} 
The following properties hold with respect to well-formed $\intcal$ terms and evaluation contexts.
\begin{enumerate}
\item 
\emph{Values}: if $\val\in\intcal$ then $\unlali\val$ is a value of $\soucal$.
\item 
\emph{Evaluation contexts}: if $\evctx\in\intcal$ then $\unlali\evctx$ is an evaluation context of $\intcal$.
\item \emph{Root clashes}: if $\tm\in\intcal$ is a root clash then $\unlali\tm$ is a root clash of $\soucal$.
\end{enumerate}
\end{lemma}
\begin{proof}
The first two points are straightforward, the third is a simple unfolding of the definitions.\qedhere
\end{proof}

\gettoappendix{prop:intcal-key-sub-prop-rev-trans}
\begin{proof}
By induction on $\tm$. The only interesting case is the \wrapt one. Let $\tm = \pack{\dabsv\varthree\varfour\tmthree}{\bag}$. Cases of the \bagt $\bag$:
\begin{itemize}
\item Variable \wrapt $\bag = \vv\varthree$:
\begin{center}$
\begin{array}{rlllll}
\unlali{\tm \isubp\tuvar{\vv\valnone}\tuvartwo{\vv\valntwo}} &= &\unlali{\pack{\dabsv\varthree\varfour\tmthree}{\tuple{\varthree_1\isubp\tuvar{\vv\valnone}\tuvartwo{\vv\valntwo},\ldots,\varthree_n\isubp\tuvar{\vv\valnone}\tuvartwo{\vv\valntwo}}}}
\\[4pt]
&=& \la\tuvarfour\unlali\tmthree\isub{\varthree_{1}}{\unlali{\varthree_1\isubp\tuvar{\vv\valnone}\tuvartwo{\vv\valntwo}}},\ldots,\isub{\varthree_{n}}{\unlali{\varthree_n\isubp\tuvar{\vv\valnone}\tuvartwo{\vv\valntwo}}}
\\[4pt]
&=_{\ih}& \la\tuvarfour\unlali\tmthree\isub{\varthree_{1}}{\unlali{\varthree_1}\isub\tuvar{\unlali{\vv\valnone}}\isub\tuvartwo{\unlali{\vv\valntwo}}},\ldots,\isub{\varthree_{n}}{\unlali{\varthree_n}\isub\tuvar{\unlali{\vv\valnone}}\isub\tuvartwo{\unlali{\vv\valntwo}}}
\\[4pt]
&=& (\la\tuvarfour\unlali\tmthree\isub{\tuvarthree}{\vv\varthree} )\isub\tuvar{\unlali{\vv\valnone}}\isub\tuvartwo{\unlali{\vv\valntwo}}
\\[4pt]
&=& (\la\tuvarfour\unlali\tmthree)\isub\tuvar{\unlali{\vv\valnone}}\isub\tuvartwo{\unlali{\vv\valntwo}}
\\[4pt]
&=& \unlali{\pack{\dabsv\varthree\varfour\tmthree}{\vv\varthree}}\isub\tuvar{\unlali{\vv\valnone}}\isub\tuvartwo{\unlali{\vv\valntwo}}
\\[4pt]
&=& \unlali{\tm}\isub\tuvar{\unlali{\vv\valnone}}\isub\tuvartwo{\unlali{\vv\valntwo}} 
\end{array}$
\end{center}

\item Evaluated \wrapt $\bag = \vv\valnthree$:
\begin{center}$
\begin{array}{rlllll}
\unlali{\tm \isubp\tuvar{\vv\valnone}\tuvartwo{\vv\valntwo}} &= &\unlali{\pack{\dabsv\varthree\varfour\tmthree}{\tuple{\val_{3_1}\isubp\tuvar{\vv\valnone}\tuvartwo{\vv\valntwo},\ldots,\val_{3_n}\isubp\tuvar{\vv\valnone}\tuvartwo{\vv\valntwo}}}}
\\[4pt]
&=& \la\tuvarfour\unlali\tmthree\isub{\varthree_{1}}{\unlali{\val_{3_1}\isubp\tuvar{\vv\valnone}\tuvartwo{\vv\valntwo}}},\ldots,\isub{\varthree_{n}}{\unlali{\val_{3_n}\isubp\tuvar{\vv\valnone}\tuvartwo{\vv\valntwo}}}
\\[4pt]
&=_{\ih}& \la\tuvarfour\unlali\tmthree\isub{\varthree_{1}}{\unlali{\val_{3_1}}\isub\tuvar{\unlali{\vv\valnone}}\isub\tuvartwo{\unlali{\vv\valntwo}}},\ldots,\isub{\varthree_{n}}{\unlali{\val_{3_n}}\isub\tuvar{\unlali{\vv\valnone}}\isub\tuvartwo{\unlali{\vv\valntwo}}}
\\[4pt]
&=& (\la\tuvarfour\unlali\tmthree\isub{\tuvarthree}{\vv\valnthree} )\isub\tuvar{\unlali{\vv\valnone}}\isub\tuvartwo{\unlali{\vv\valntwo}}
\\[4pt]
&=& \unlali{\pack{\dabsv\varthree\varfour\tmthree}{\vv{\valnthree}}}\isub\tuvar{\unlali{\vv\valnone}}\isub\tuvartwo{\unlali{\vv\valntwo}}
\\[4pt]
&=& \unlali{\tm}\isub\tuvar{\unlali{\vv\valnone}}\isub\tuvartwo{\unlali{\vv\valntwo}} 
\end{array}$\qedhere
\end{center}
\end{itemize}
\end{proof}

\gettoappendix{thm:rev-transl-implem-bricks}
\begin{proof}
\hfill
\begin{enumerate}
\item We only show the case of reduction at top level, as the others follow from the preservation of evaluation contexts by reverse translation (\reflemma{rev-transl-properties}.3). The projection case is also obvious. For the $\rtoibv$ case, let $\tm = \pack{ \dabsv\vartwo\var\tmthree }{ \vv\valnone } \vv\valntwo  \rtoibv  
\tmthree\isubp\tuvartwo{\vv\valnone}\tuvar{\vv\valntwo} = \tmtwo$. Then:
\[\unlali\tm = (\la\tuvar\unlali\tmthree\isub\tuvartwo{\vv{\unlali\valnone}}) \vv{\unlali\valntwo} \tobv \unlali\tmthree\isub\tuvartwo{\vv{\unlali\valnone}}\isub\tuvar{\vv{\unlali\valntwo}} =_{\refpropeq{intcal-key-sub-prop-rev-trans}}\unlali{\tmthree\isubp\tuvartwo{\vv{\valnone}}\tuvar{\vv{\valntwo}}} = \unlali\tmtwo\]

\item We prove the two implications separately.
\begin{itemize}
\item \emph{$\tm$ $\toint$-normal implies $\unlali\tm$ $\tosou$-normal}. By contradiction, suppose that $\unlali\tm =\evctxp\tmtwo$ with $\tmtwo$ source root redex, so that $\unlali\tm$ is not normal. We show that then $\tm = \evctxtwop{\tmtwo'}$ such that $\unlali\evctxtwo = \evctx$, $\unlali{\tmtwo'}=\tmtwo$, and $\tmtwo'$ is an intermediate root redex. The proof is by induction on $\evctx$. The base case is the crucial one, since it maps back source root redexes to intermediate root redexes. We give the crucial case, also showing the need of the closure hypothesis. 
\begin{itemize}
\item \emph{Source root $\betav$}, that is, $\unlali\tm = (\la\tuvar\tmthree)\vv\valnone \rtobv \tmthree\isub\vvvar\vvvalnone$. Then the shape of $\tm$ has to be $\pack{ \dabsv\vartwo\var\tmthree' }{ \bag } \vv\valntwo$ with $\unlali{\pack{ \dabsv\vartwo\var\tmthree' }{ \bag }} = \la\tuvar\tmthree$ and $\unlali\vvvalntwo = \vv\valnone$. For $\tm$ to be a redex, $\bag$ should not be a variable \bagt, which is indeed enforced by the closure hypothesis. Then, $\bag = \vvvalnthree$ and we have $\tm = \pack{ \dabsv\vartwo\var\tmthree' }{ \vvvalnthree } \vv\valntwo \rtoibv \tmthree'\isubp\vvvartwo\vvvalnthree\vvvar\vvvalntwo$. By point 1, the step in $\intcal$ is mapped to a step in $\soucal$, which coincides with $(\la\tuvar\tmthree)\vv\valnone \rtobv \tmthree\isub\vvvar\vvvalnone$ by determinism of $\soucal$.
\end{itemize}
The inductive cases all follows from the \ih and the fact that each case of $\evctx$ matches one and only one case for $\evctxtwo$, satisfying the statement.

\item \emph{$\unlali\tm$ $\tosou$-normal implies $\tm$ $\toint$-normal}. By contradiction, assume that $\tm$ is not $\toint$-normal. Then, $\tm$ reduces and by projection (Point 1) $\unlali\tm$ reduces. Absurd.
\end{itemize}

\item Note that if $\tm$ is $\toint$-normal then, by the halt property (Point 2) $\unlali\tm$ is $\tosou$-normal, against hypotheses. Then $\tm \toint \tmthree$ for some $\tmthree$. By projection (Point 1), $\unlali{\tm} \tosou \unlali{\tmthree}$ and it is the same kind of step than on $\intcal$. Since $\tosou$ is deterministic (\reflemma{source-determinism}), $\unlali{\tmthree}=\tmtwo$. 

\item Straightforward induction on $\tmtwo$.
\qedhere

\end{enumerate}
\end{proof}

\gettoappendix{cor:step-preservation}
\begin{proof}
	By induction on $k$. If $k=0$ then $\tm=\tmtwo$;
	taking $\tmthree\defeq\lali\tm$, the inverse property gives $\unlali\tmthree = \unlali{\lali\tm} = \tm = \tmtwo$. 
	For $k>0$, one has $\tm \tosou^{k-1} \tmtwo' \tosou \tmtwo$. By \ih, $\lali\tm \toint^{k-1} \tmthree'$ with $\unlali{\tmthree'}=\tmtwo'$, so $\unlali{\tmthree'}\tosou \tmtwo$; 
	by reflection $\tmthree' \toint \tmthree$ for some $\tmthree$ such that $\unlali{\tmthree}=\tmtwo$. 
\end{proof}
\section{Proofs and Auxiliary Notions of \refsect{target-calculus} (The Target Calculus $\tarcal$ and Name Elimination)}
\label{sect:app-target-calculus}

\begin{figure}[t!]
\centering
\fbox{
\begin{tabular}{ccc}
	\multicolumn{2}{c}{	\setlength{\arraycolsep}{2pt}
		$\begin{array}{r@{\hspace{.5cm}} rll@{\hspace{1.cm}} r@{\hspace{.5cm}} rll}
		\textsc{Projected vars} & \prvar,\prvartwo &\grameq& \proj_i \efvar \midd \proj_i \envar
		&
		\textsc{Terms} &
		\tm,\tmtwo,\tmthree,\tmfour&\grameq& \prvar \midd \proj_i \tm \midd \vv\tm \midd \tm\,\tmtwo   \midd  \packnm\tm\bag 
		\\
		\textsc{\Bagts} &
		\bag,\bagtwo&\grameq& \vec\prvar \midd \vec\val
		&
		\textsc{Values} &
		\val,\valtwo & \grameq  &\packnm\tm\bag \midd \vv\val
		\\
		\textsc{Eval Contexts} &
		\evctx,\evctxtwo & \grameq  & 
		\multicolumn{4}{l}{\ctxhole \midd \tm\, \evctx \midd \evctx\,\vv \val \midd \proj_{i} \evctx \midd \tuple{\tuv\tm, \evctx, \tuv\val} }
	\end{array}$
	}
\\\\[-10pt]
	\multicolumn{2}{c}{
	$\begin{array}{llllll}
			\multicolumn{3}{c}{\textsc{\Lifted/source norms } (\vvar\in\set{\lvar,\svar})}
			\\
			\prnorm\tm &\defeq& \max\set{i\in\nat \,|\, \proj_i\vvar\mbox{ appears in $\tm$ out of \wrapts bodies}}
	\end{array}$
	} 
\\\\[-10pt]

	\multicolumn{2}{c}{
	\setlength{\arraycolsep}{1.6pt}
	$\begin{array}{rll| rll}
		\multicolumn{6}{c}{\textsc{Projecting substitution}}
		\\
		(\proj_i\lvar)\itsubp\lvar\vvvalnone\svar\vvvalntwo & \defeq & \val_{1_i}
		&
		\vv\tm \itsubp\lvar\vvvalnone\svar\vvvalntwo & \defeq & \vv{\tm \itsubp\lvar\vvvalnone\svar\vvvalntwo}
		\\
		(\proj_i\svar)\itsubp\lvar\vvvalnone\svar\vvvalntwo & \defeq & \val_{2_i}
		&
		(\tm\,\tmtwo)\itsubp\lvar\vvvalnone\svar\vvvalntwo & \defeq & \tm\itsubp\lvar\vvvalnone\svar\vvvalntwo\,
		\tmtwo\itsubp\lvar\vvvalnone\svar\vvvalntwo
		\\
		(\proj_i\tm)\itsubp\lvar\vvvalnone\svar\vvvalntwo & \defeq & \proj_i\tm\itsubp\lvar\vvvalnone\svar\vvvalntwo
		&
		\packnm\tm\bag \itsubp\lvar\vvvalnone\svar\vvvalntwo & \defeq & \packnm\tm{\bag\itsubp\lvar\vvvalnone\svar\vvvalntwo}
	\end{array}$
	}
\\\\[-10pt]

	\multicolumn{2}{c}{
	$\begin{array}{r l l@{\hspace{.5cm}} l}
		\multicolumn{4}{c}{\textsc{Root rewriting rules}}
		\\
		\packnm\tm\vvvalnone \, \vvvalntwo & \rtoccbv & 
		\tm\itsubp\efvar\vvvalnone\envar\vvvalntwo
		& \mbox{if } n = \norm\vvvalnone \mbox{ and } m = \norm\vvvalntwo
		\\
		\proj_{i} \vv\val
		& \rtoccproj &
		\val_{i} & \mbox{if }i \leq \norm{\vv\val}
	\end{array}$
	}
\\\\[-10pt]

	\begin{tabular}{ccc}
		\multicolumn{2}{c}{\textsc{Contextual closure}}
		\\
		\AxiomC{$\tm \rootRew{a} \tmtwo$}
		\UnaryInfC{$\evctxp{\tm} \Rew{a} \evctxp{\tmtwo}$}
		\DisplayProof
		& 
		for $a \in \{ \tarprefix\betav, \tarprefix\projsym\}$
	\end{tabular}
&
	$\begin{array}{c}
		\textsc{Notation}
		\\
		\totar  \ \ \defeq \ \ \totbv \cup \totproj
	\end{array}$
\end{tabular}
}
\caption{The target calculus $\tarcal$.}
\label{fig:target_calculus-bis}
\end{figure}
In this section, we extend the outline of the target calculus $\tarcal$ and of name elimination given in \refsect{target-calculus} (the notions introduced there are not repeated here), giving more definitions as well as all the properties and proofs. The definition of $\tarcal$ in \reffig{target_calculus-bis} adds the rewriting rules (and the notion of substitution they rely upon) to the grammars already given in \refsect{target-calculus}. The next paragraphs explain the involved concepts.
 
\paragraph{Projecting Substitution.}
The elimination of names rises a slight technical issue in mapping steps from $\intcal$ to $\tarcal$, because it commutes with meta-level substitutions only \emph{up to projections}. The replacement in $\intcal$ of a variable $\var$ by a value $\val\in\intcal$ indeed is simulated in $\tarcal$ by:
 \begin{enumerate}
 \item \emph{Substitution}: the replacement of, say, $\lvar$ in $\pi_i\lvar$ by a vector of target values $\vv\valtwo\in\tarcal$ containing at the $i$-th position the translation $\valtwo_i\in\tarcal$ of $\val\in\intcal$, obtaining $\pi_i\vv\valtwo$, 
\item \emph{Up to projection}: followed by the projection $\pi_i\vv\valtwo \toccproj \valtwo_i$.
\end{enumerate}
To circumvent this issue, we define an ad-hoc notion of \emph{projecting substitution} for the target calculus that performs the projections \emph{on-the-fly}, that is, while substituting. Moreover, as for $\intcal$ we define it as a simultaneous projecting substitution on both $\lvar$ and $\svar$.

In the next definition, $\fnorm\tm$ and $\nnorm\tm$ are the \lifted and source norms defined in \reffig{target_calculus-bis}. They are extended to evaluation contexts as expected.

\begin{definition}[Well-formed \wrapts/terms]
A \wrapt $\packnm\tm\bag$ is \emph{well-formed} if $\fnorm\tm\leq n$, $\nnorm\tm\leq m$, and $\norm{\bag} = n$.
Terms $\tm\in\tarcal$ and evaluation contexts $\evctx\in\tarcal$ are \emph{well-formed} if all their \wrapts are well-formed. 
\end{definition}

\begin{definition}[Projecting substitution]
Let $\tm\in\intcal$ be well-formed target and $\vvvalnone,\vvvalntwo\in\tarcal$ be tuples of values such that $\norm\vvvalnone \geq \fnorm\tm$ and $\norm\vvvalntwo \geq \nnorm\tm$. Then $\tm \itsubp\lvar\vvvalnone\svar\vvvalntwo$ is defined in \reffig{target_calculus-bis}.
\end{definition}

\paragraph{Operational Semantics} The rewriting rules of $\tarcal$ (in \reffig{target_calculus-bis}) are as those of $\intcal$ up to the different notion of substitution in $\totbv$. 

\begin{definition}[Clashes]
A well-formed term $\tm\in\tarcal$ is a \emph{clash} if it has shape $\evctxp{\tmtwo}$ where $\tmtwo$ has one of the following forms, called a \emph{root clash}:
\begin{itemize}
\item \emph{Clashing projection}: $\tm = \proj_{i} \val$ and ($\val$ is not a tuple or it is but $\norm\val <i$);
\item \emph{Clashing \wrapt}:  $\tm = \packnm{ \tmtwo }{ \bag }  \val_2$ and ($\val_2$ is not a tuple or it is but $m \neq \norm{\vv{\val_2}}$);
\item \emph{Clashing tuple}: $\tm = \vv\tmtwo  \tmthree$.
\end{itemize}
\end{definition}

 We say that $\tm\in\tarcal$ is \emph{good} if it is well-formed and clash-free.

\begin{lemma}[Determinism]
	\label{l:target-determinism}
Let $\tm\in\tarcal$. If $\tm \totar \tmtwo$ and $\tm \totar \tmthree$ then $\tmtwo = \tmthree$ and $\tm$ is not a value.
\end{lemma}

\begin{proof}
By induction on $\tm$.
\end{proof}

\begin{lemma}[Harmony for $\tarcal$]
\label{l:target-harmony}
Let $\tm\in\tarcal$ be closed and good. Then either $\tm$ is a value or $\tm \totar \tmtwo$ for some closed good $\tmtwo\in\tarcal$.
\end{lemma}

\begin{proof}
	If $\tm \totar \tmtwo$, the fact that $\tmtwo$ is closed and clash-free follows from the hypothesis on $\tm$.
Let us prove, by induction on $\tm$, that $\tm$ is a value or $\tm \totar \tmtwo$ for some well-formed $\tmtwo\in\tarcal$. Cases:
\begin{itemize}
\item \emph{Projected variable}: impossible, because $\tm$ is closed.

\item \emph{\Wrapts}: then $\tm=\pack\tmthree\bag$ does not reduce and it is a value.

\item \emph{Projection}, that is, $\tm = \proj_{i} \tmtwo$. 
By \ih, $\tmtwo$ is either a value or it reduces. 
If $\tmtwo$ is a value, by clash-freeness, it must be a tuple $\tmtwo = \vv\val$ of length $\geq i$. Therefore, $\tm \toccproj \val_{i}$. 
If instead $\tmtwo$ reduces then $\tm$ reduces, because $\proj_{i}\ctxhole$ is an evaluation context.

\item \emph{Application}, that is, $\tm = \tmtwo \, \tmthree$. By \ih, $\tmthree$ is either a value or it reduces. If $\tmthree$ reduces then $\tm$ reduces because $\tmtwo\,\ctxhole$ is an evaluation context. 
If $\tmthree$ is a value $\val$ then we apply the \ih to $\tmtwo$. 
If $\tmtwo$ reduces then $\tm$ reduces because $\ctxhole\,\val$ is an evaluation context. Otherwise $\tmtwo$ is a value $\valtwo$. By clash-freeness and properness, $\valtwo$ is a \wrapt $\pack\tmfour{\vv\valthree}$ with $\fnorm{\tmfour} \leq \norm{\vv\valthree}$ and $\val$ is a tuple such that $\nnorm{\tmfour} \leq \norm{\val}$. Then $\tm \totar \tmfour\itsubp\efvar{\vv\valthree}\envar{\val}$.

\item \emph{Tuple}, that is, $\tm = \vv\tmtwo$. Let $i \in \nat$ be the smallest index $i \leq \norm{\vv\tmtwo}=k$ such that $\tmtwo_{i}$ is not a value and $\tmtwo_{j}$ is a value for all $j$ satisfying $i <j \leq k$. If $i = 0$ then $\tm$ is a value. Otherwise, by \ih $\tmtwo_{i}$ reduces. Then $\tm$ reduces, since $\tuple{\tmtwo_{1},\mydots,\tmtwo_{i-1},\ctxhole,\val_{i+1},\mydots,\val_{k}}$ is an evaluation context.
\qedhere
\end{itemize}
\end{proof}

\begin{figure}[t!]
\centering
\fbox{
\begin{tabular}{ccc}
\setlength{\arraycolsep}{2pt}
$\begin{array}{rll@{\hspace{.85cm}} rll@{\hspace{.85cm}} rlllllrllrlllll}
		\multicolumn{9}{c}{\textsc{Name elimination translation } \intcal \to \tarcal}
		\\[2pt]
		\clc{\vartwo_i}\vvvartwo\vvvar & \defeq  & \proj_i \lvar
		&
		\clc{\proj_i\tm}\vvvartwo\vvvar & \defeq  & \proj_i{\clc\tm\vvvartwo\vvvar}
		&
		\clc{\tuple{\tm_1,\mydots,\tm_n}}\vvvartwo\vvvar & \defeq  & \tuple{\clc{\tm_1}\vvvartwo\vvvar,\mydots,\clc{\tm_n}\vvvartwo\vvvar}
		\\[5pt]
		\clc{\var_i}\vvvartwo\vvvar & \defeq  & \proj_i \svar
		&
		\clc{\tm\tmtwo}\vvvartwo\vvvar & \defeq  & \clc\tm\vvvartwo\vvvar \,\clc\tmtwo\vvvartwo\vvvar
		&
		\clc{\pack{\dabsv\varthree\varfour\tm}{\bag} }\vvvartwo\vvvar & \defeq  & 
		\pack{\clc\tm{\tuvarthree}{\tuvarfour}}{ \clc\bag\vvvartwo\vvvar}

\\[4pt]\hline
		\multicolumn{9}{c}{\textsc{Naming reverse translation }\tarcal \to \intcal}
		\\[2pt]
		\naming{\proj_i \lvar}  & \defeq  & \vartwo_i 
		&
		\naming{(\proj_i\tm)} & \defeq  & \proj_i{\naming\tm}
		&
		\naming{\tuple{\tm_1,\mydots,\tm_n}}  & \defeq  & \tuple{\naming{\tm_1} ,\mydots,\naming{\tm_n} }
		\\[4pt]
		\naming{\proj_i \svar}  & \defeq  & \var_i 
		&
		\naming{(\tm\,\tmtwo)}  & \defeq  & \naming\tm\,  \naming\tmtwo
		&
		\naming{\packnm\tm\bag} & \defeq  & \pack{\dabsv\varthree\varfour\namingp\tm\tuvarthree\tuvarfour}{\naming\bag} 
		\\
		&&&&&&\multicolumn{3}{r}{\mbox{with $\tuvarthree\vdisjoint\tuvarfour$ fresh, }\norm\tuvarthree=n\mbox{, and } \norm\tuvarfour=m}
	\end{array}$
\end{tabular}
}
\caption{Definition of the translation $\clc\tm\tuvartwo\tuvar:\intcal \to \tarcal$ and the reverse translation $\naming\tm:\tarcal \to \intcal$.}
\label{fig:int-tar_translations}
\end{figure}
\paragraph{Translation.} The \emph{name elimination} translation in \reffig{int-tar_translations} takes a well-formed and possibly open  term $\tm\in\intcal$ and two sequences of variables $\tuvartwo\vdisjoint\tuvar$, the sequence of \emph{\lifted} variables $\vvvartwo$ and the sequence of \emph{source} variables $\vvvar$, which together cover the free variables of $\tm$, that is, $\fv\tm \subseteq \vvvartwo\vdisjoint\vvvar$. We write $\clc\tm\vvvartwo\vvvar$ for the translation of $\tm$ with respect to the sequences $\vvvartwo $ and $\vvvar$. 

\begin{lemma}[Basic properties of name elimination]
\label{l:varelim-properties}
Let $\vartwolist$ and $\varlist$ be two disjoint sequences of variables.
\begin{enumerate}
\item 
\emph{Values}: if $\val\in\intcal$ then $\clc\val\vartwolist\varlist$ is a value of $\tarcal$.
\item 
\emph{Evaluation contexts}: if $\evctx\in\intcal$ then $\clc\evctx\vartwolist\varlist$ is an evaluation context of $\tarcal$.
\end{enumerate}
\end{lemma}
\begin{proof}
\hfill
\begin{enumerate}
\item Straightforward induction on $\val$.
\item Straightforward induction on $\evctx$, using the first point for the productions for evaluation contexts that involve values. \qedhere
\end{enumerate}
\end{proof}

The name elimination translation is a termination-preserving strong bisimulation between $\intcal$ and $\soucal$ as stated by \refthm{tar-simulation} below, which is proved using the commutation of the substitution of $\tarcal$ with the translation of the next proposition. The theorem mimics the one relating  $\soucal$ and $\intcal$ (\refthm{rev-transl-implem-bricks}), but the statement is stronger as it does not need the closure hypothesis. 

\begin{proposition}[Key commutation of projecting substitution and name elimination]
\label{prop:tarcal-key-sub-prop-trans}
Let $\tm,\valnone,\valntwo\in\intcal$ be well-formed and such that  $\fv\tm\subseteq \tuvartwo\vdisjoint\tuvar$, $\fv\valnone\subseteq\varthreelist\vdisjoint\varfourlist$, and $\fv\valntwo\subseteq\varthreelist\vdisjoint\varfourlist$. Then:
\begin{center}
$\clc \tm\vartwolist\varlist\itsubp\lvar{\vv{\clc\valnone\varthreelist\varfourlist} }\svar{\vv{\clc\valntwo\varthreelist\varfourlist}} = \clc{\tm\isubp	\vvvartwo\vvvalnone\vvvar\vvvalntwo}\varthreelist\varfourlist  $.
\end{center}
\end{proposition}

\begin{proof}
 By induction on $\tm$. Cases:
\begin{itemize}
\item Variable $\var_i$ ($\vartwo_j$ would be analogous). 
\begin{center}
$\begin{array}{llll}
 \clc \tm\vartwolist\varlist\itsubp\efvar{\vv{\clc{\val_1}\varthreelist\varfourlist} }\envar{\vv{\clc{\val_2}\varthreelist\varfourlist}} 
 &=
 \\[5pt]
 \clc {\var_i}\vartwolist\varlist\itsubp\efvar{\vv{\clc{\val_1}\varthreelist\varfourlist} }\envar{\vv{\clc{\val_2}\varthreelist\varfourlist}} 
&=
\\[5pt]
 \proj_i\envar \itsubp\efvar{\vv{\clc{\val_1}\varthreelist\varfourlist} }\envar{\vv{\clc{\val_2}\varthreelist\varfourlist}} 
&=
\\[5pt]
 \clc{{\val_2}_i}\varthreelist\varfourlist			
&=
\\[5pt]
\clc{\var_i\isubp\tuvartwo{\vvvalnone}\tuvar{\vvvalntwo}}\varthreelist\varfourlist
&=
\\[5pt]
\clc{\tm\isubp\tuvartwo{\vvvalnone}\tuvar{\vvvalntwo}}\varthreelist\varfourlist 
\end{array}$
\end{center}

\item \emph{Tuple}:
\begin{center}
$\begin{array}{llll}
 \clc{\tuple{\tm_1,\ldots,\tm_n}}\vartwolist\varlist \itsubp\efvar{\vv{\clc{\val_1}\varthreelist\varfourlist} }\envar{\vv{\clc{\val_2}\varthreelist\varfourlist}}
&=
\\[5pt]
\tuple{
	\clc{
		\tm_{1}}\vartwolist\varlist,
		\ldots,
		\clc{\tm_{n}}\vartwolist\varlist
		}
		\itsubp\efvar{ \vv{\clc{\val_1}\varthreelist\varfourlist} }	
		\envar{ \vv{\clc{\val_2}\varthreelist\varfourlist} }
&=
\\[5pt]
	\tuple{
	\clc{
		\tm_{1}}\vartwolist\varlist\itsubp\efvar{\vv{\clc{\val_1}\varthreelist\varfourlist} }\envar{\vv{\clc{\val_2}\varthreelist\varfourlist}},
		\ldots,
		\clc{\tm_{n}}\vartwolist\varlist\itsub\efvar{\vv{\clc{\val_1}\varthreelist\varfourlist} }\itsub\envar{\vv{\clc{\val_2}\varthreelist\varfourlist}
		}
	}
&=_{\ih}
\\[5pt]
	\tuple{
\clc{\tm_1\isubp\tuvartwo\vvvalnone\tuvar\vvvalntwo}\varthreelist\varfourlist,
		\ldots,
\clc{\tm_n\isubp\tuvartwo\vvvalnone\tuvar\vvvalntwo}\varthreelist\varfourlist
	}
&=
\\[5pt]
\clc{	
\tuple{
\tm_1\isubp\tuvartwo\vvvalnone\tuvar\vvvalntwo,
		\ldots,
\tm_n\isubp\tuvartwo\vvvalnone\tuvar\vvvalntwo
	}
}\varthreelist\varfourlist
&=
\\[5pt]
\clc{	
\tuple{
\tm_1,
		\ldots,
\tm_n
	}\isubp\tuvartwo\vvvalnone\tuvar\vvvalntwo
}\varthreelist\varfourlist

\end{array}$
\end{center}

\item \emph{\Wrapt} $\tm = \pack{\dabsv{\vartwo'}{\var'}\tmthree}{\bag}$. Then 
\begin{center}
$\begin{array}{llll}
\clc \tm\vartwolist\varlist\itsubp\efvar{\vv{\clc{\val_1}\varthreelist\varfourlist} }\envar{\vv{\clc{\val_2}\varthreelist\varfourlist}} 
 &=
 \\[5pt]
 \clc{\pack{\dabsv{\vartwo'}{\var'}\tmthree}{\bag}}\vartwolist\varlist \itsubp\efvar{\vv{\clc{\val_1}\varthreelist\varfourlist} }\envar{\vv{\clc{\val_2}\varthreelist\varfourlist}}
&=
\\[5pt]
 \pack{\clc\tmthree{\vv{\vartwo'}}{\vv{\var'}}}{
	\clc{\bag}\vartwolist\varlist}\itsubp\efvar{\vv{\clc{\val_1}\varthreelist\varfourlist} }\envar{\vv{\clc{\val_2}\varthreelist\varfourlist}}
&=
\\[5pt]
 \pack{\clc\tmthree{\vv{\vartwo'}}{\vv{\var'}}}{
	\clc{\bag}\vartwolist\varlist\itsubp\efvar{\vv{\clc{\val_1}\varthreelist\varfourlist} }\envar{\vv{\clc{\val_2}\varthreelist\varfourlist}}}
& =_{\ih}
\\[5pt]
 \pack{\clc\tmthree{\vv{\vartwo'}}{\vv{\var'}}}{
	\clc{\bag\isubp\tuvartwo\vvvalnone\tuvar\vvvalntwo}\varthreelist\varfourlist
	}
&=
\\[5pt]
 \clc{
 \pack{\dabsv{\vartwo'}{\var'}\tmthree}{\bag\isubp\tuvartwo\vvvalnone\tuvar\vvvalntwo}
	}\varthreelist\varfourlist
&=
\\[5pt]
 \clc{
 \pack{\dabsv{\vartwo'}{\var'}\tmthree}{\bag}
	\isubp\tuvartwo\vvvalnone\tuvar\vvvalntwo
	}\varthreelist\varfourlist
&=
\\[5pt]
\clc{\tm\isubp\tuvartwo\vvvalnone\tuvar\vvvalntwo}\varthreelist\varfourlist
\end{array}$
\end{center}

\item \emph{Application} and \emph{projection}: they follow from the \ih as the previous cases.
\qedhere

\end{itemize}
\end{proof}

\begin{theorem}[Intermediate-target termination-preserving strong bisimulation]
\label{thm:tar-simulation}
Let $\tm\in\intcal$ be well-formed and such that $\fv\tm\subseteq\vartwolist\vdisjoint\varlist$.
\begin{enumerate}
\item \emph{Projection}: if $\tm\toibv\tmtwo$ then $\clc\tm\vartwolist\varlist \totbv \clc\tmtwo\vartwolist\varlist$ and if $\tm\toiproj\tmtwo$ then $\clc\tm\vartwolist\varlist \totproj \clc\tmtwo\vartwolist\varlist$.
\item \label{p:tar-simulation-halt}\emph{Halt}: $\tm$ is $\toint$-normal (resp. a value, resp. a clash) if and only if $\clc\tm\vartwolist\varlist$ is  $\totar$-normal (resp. a value, resp. a clash).
\item \emph{Reflection}: if $\clc\tm\vartwolist\varlist \totbv \tmtwo$ then there exists $\tmthree$ such that $\tm\toibv\tmthree$  and $\clc\tmthree\vartwolist\varlist = \tmtwo$,  and if $\clc\tm\vartwolist\varlist \totproj \tmtwo$ then there exists $\tmthree$ such that $\tm\toiproj\tmthree$  and $\clc\tmthree\vartwolist\varlist = \tmtwo$.
\end{enumerate}
\end{theorem}

\begin{proof}
\hfill
\begin{enumerate}
\item We only show the case of reduction at top level, as the others follow from \reflemma{varelim-properties}. The projection case is also obvious. For the $\beta$ case, let $\tm = \pack{ \dabsv\varthree\varfour\tmthree }\vvvalnone\, \vvvalntwo  \rtollbv  
\tmthree\isubp{\tuvarthree}\vvvalnone{\tuvarfour}\vvvalntwo = \tmtwo$. Then:
\begin{center}
$\begin{array}{llllll}
\clc\tm\vartwolist\varlist 
& = &\clc{ \pack{ \dabsv\varthree\varfour\tmthree }\vvvalnone\, \vvvalntwo }\vartwolist\varlist 
\\[5pt]
& = &\pack{\clc \tmthree\varthreelist\varfourlist }{ \vv{\clc\valnone\vartwolist\varlist} }\, \clc\vvvalntwo\vartwolist\varlist
\\[5pt]
& \toccbv &
\clc \tmthree\varthreelist\varfourlist \itsubp\efvar{\vv{\clc\valnone\vartwolist\varlist }}\envar{\vv{\clc\valntwo\vartwolist\varlist} }
\\[5pt]
&=_{\refpropeq{tarcal-key-sub-prop-trans}}&
\clc{\tmthree\isubp{\tuvarthree}\vvvalnone{\tuvarfour}\vvvalntwo}\vartwolist\varlist
& = & \clc\tmtwo\vartwolist\varlist
\end{array}$
\end{center}
\item We prove the two implications separately.
\begin{itemize}
\item \emph{$\tm$ $\toint$-normal implies $\clc\tm\vartwolist\varlist$ $\totar$-normal}. It is easily seen that the elimination of names cannot create redexes, values, or clashes, as it only acts on variables, replacing them with variables. For redexes, it amounts to show that if $\clc\tm\vartwolist\varlist \totbv \tmtwo$ then $\tm\toibv \tmthree$ (then $\clc\tmthree\vartwolist\varlist = \tmtwo$ follows by projection and determinism of $\totbv$) by induction on the target evaluation context for the $\totbv$ step, and similarly for $\totproj$. The base case is a simple inspection of the redexes and the translation, the inductive cases follow from the \ih and from the immediate fact that the translation reflects values (that is, if $\clc\tm\vartwolist\varlist$ is a value of $\tarcal$ then $\tm$ is a value of $\intcal$). The reasoning for clashes is analogous, and for values it holds by \reflemma{varelim-properties}.
\item \emph{$\clc\tm\vartwolist\varlist$ $\totar$-normal implies $\tm$ $\toint$-normal}. By contradiction, assume that $\tm$ is not $\totar$-normal. Then, $\tm$ reduces and by projection (Point 1) $\clc\tm\vartwolist\varlist$ reduces. Absurd.
\end{itemize}

\item Note that if $\tm$ is $\toint$-normal then, by the halt property (Point 2) $\clc\tm\vartwolist\varlist$ is $\totar$-normal, against hypotheses. Then $\tm \toint \tmthree$ for some $\tmthree$. By projection (Point 1), $\clc\tm\vartwolist\varlist \totar \clc\tmthree\vartwolist\varlist$ and it is the same kind of step than on $\intcal$. Since $\totar$ is deterministic (\reflemma{target-determinism}), $\clc\tmthree\vartwolist\varlist=\tmtwo$. 
\qedhere

\qedhere
\end{enumerate}
\end{proof}

\paragraph{From Source to Target} The two obtained strong bisimulation theorems---namely \refthm{rev-transl-implem-bricks} and \refthm{tar-simulation}---do not compose, because the first one uses the reverse translation from $\intcal$ to $\soucal$, while the second one uses the direct translation from $\intcal$ to $\tarcal$. To solve the issue, we need a reverse translation from $\tarcal$ to $\intcal$, which is smoothly obtained. Such reverse \emph{naming} translation is defined in \reffig{int-tar_translations}. It takes a well-formed $\tm\in\tarcal$ and two sequences of variables $\tuvartwo\vdisjoint\tuvar$ such that $\norm\tuvartwo \geq \fnorm\tm$ and $\norm\tuvar \geq \nnorm\tm$ and uses them to give names to the indices of the projections. We write $\naming\tm$ for the translation of $\tm$ with respect to the sequences $\vvvartwo $ and $\vvvar$. On closed terms, the translation is meant to be applied with empty parameter lists (as $\namingp\tm\emptylist\emptylist$). 

Unsurprisingly, naming is a termination-preserved strong bisimulation from the target calculus to the intermediate one, as stated by the theorem below. For its proof, we need the mirror image of some of the results proved for name elimination. 

\begin{lemma}[Basic properties of naming]
\label{l:varelim-properties-mirror}
Let $\vartwolist$ and $\varlist$ be two suitable disjoint sequences of variables.
\begin{enumerate}
\item 
\emph{Values}: if $\val\in\tarcal$ then $\naming\val$ is a value of $\intcal$.
\item 
\emph{Evaluation contexts}: if $\evctx\in\tarcal$ then $\naming\evctx$ is an evaluation context of $\intcal$.
\end{enumerate}
\end{lemma}

\begin{proof}
Straightforward inductions.\qedhere
\end{proof}

\begin{proposition}[Commutation of substitution and naming]
\label{prop:tarcal-key-sub-prop-trans-mirror}
Let $\tm,\valnone,\valntwo\in\tarcal$ be well-formed. Then:
\begin{center}
$\namingp\tm\tuvarthree\tuvarfour\isubp{\tuvarthree}{\naming\vvvalnone}{\tuvarfour}{\naming\vvvalnone} = \naming{\tm\itsubp\lvar\vvvalnone\svar\vvvalntwo}$.
\end{center}
for any $\tuvar$, $\tuvartwo$, $\tuvarthree$, and $\tuvarfour$ such that the involved expressions are well defined.
\end{proposition}
\begin{proof}
Analogous to the proof of \refprop{tarcal-key-sub-prop-trans}.\qedhere
\end{proof}

The next auxiliary lemma is needed for the inverse property.
\begin{lemma}
\label{l:names-elim-stability-by-renaming}
Let $\tm\in\intcal$ such that $\fv\tm\subseteq\tuvartwo\vdisjoint\tuvar$ and let $\tuvarthree$ and $\tuvarfour$ be such that $\norm\tuvartwo=\norm\tuvarthree$ and $\norm\tuvar=\norm\tuvarfour$. Then $\clc\tm\tuvartwo\tuvar = \clc{\tm\isubp\tuvartwo\tuvarthree\tuvar\tuvarfour}\tuvarthree\tuvarfour$.
\end{lemma}

\begin{proof}
By induction on $\tm$.
\end{proof}

The next theorem is given in the special case of $\tm$ closed because it is how it is applied in the final bisimulation theorem after it, but it could be more generally stated for open terms, as it is the case for name elimination (\refthm{tar-simulation}).

\begin{theorem}[Target-intermediate termination-preserving strong bisimulation]
\label{thm:tar-int-implementation}
Let $\tm\in\tarcal$ be closed.
\begin{enumerate}
\item \emph{Projection}: if $\tm\totbv\tmtwo$ then $\namingp\tm\emptylist\emptylist \toibv \namingp\tmtwo\emptylist\emptylist$ and if $\tm\totproj\tmtwo$ then $\namingp\tm\emptylist\emptylist \toiproj \namingp\tmtwo\emptylist\emptylist$.
\item \emph{Halt}: $\tm$ is $\totar$-normal (resp. a value, resp. a clash) if and only if $\namingp\tm\emptylist\emptylist$ is  $\totar$-normal (resp. a value, resp. a clash).
\item \emph{Reflection}: if $\namingp\tm\emptylist\emptylist \toibv \tmtwo$ then there exists $\tmthree$ such that $\tm\totbv\tmthree$  and $\namingp\tmthree\emptylist\emptylist = \tmtwo$,  and if $\namingp\tm\emptylist\emptylist \toiproj \tmtwo$ then there exists $\tmthree$ such that $\tm\totproj\tmthree$  and $\namingp\tmthree\emptylist\emptylist = \tmtwo$.

\item \emph{Inverse}: let $\tm\in\intcal$ such that $\fv\tm\subseteq\tuvartwo\vdisjoint\tuvar$. Then $\naming{\clc\tm\tuvartwo\tuvar} =_\alpha \tm$.
\end{enumerate}
\end{theorem}

\begin{proof}
\hfill
\begin{enumerate}
\item We only show the case of reduction at top level, as the others follow from \reflemma{rev-transl-properties}.2. The projection case is also obvious. For the $\beta$ case, let $\tm = \packnm\tm\vvvalnone \, \vvvalntwo  \rtoccbv  
		\tm\itsubp\efvar\vvvalnone\envar\vvvalntwo = \tmtwo$ with $n = \norm\vvvalnone$ and $m = \norm\vvvalntwo$. Then:
\begin{center}
$\begin{array}{llllll}
\namingp\tm\emptylist\emptylist
& = &\namingp{\packnm\tm\vvvalnone \, \vvvalntwo}\emptylist\emptylist
\\[5pt]
& = &\pack{\dabsv\varthree\varfour\namingp\tm\tuvarthree\tuvarfour}{\namingp\vvvalnone\emptylist\emptylist}\,\namingp\vvvalnone\emptylist\emptylist 
\\[5pt]
& \toibv &
\namingp\tm\tuvarthree\tuvarfour\isubp{\tuvarthree}{\namingp\vvvalnone\emptylist\emptylist}{\tuvarfour}{\namingp\vvvalnone\emptylist\emptylist}
\\[5pt]
&=_{\refpropeq{tarcal-key-sub-prop-trans-mirror}}&
\namingp{\tm\itsubp\lvar\vvvalnone\svar\vvvalntwo}\emptylist\emptylist
& = & \namingp\tmtwo\emptylist\emptylist
\end{array}$
\end{center}
\item Analogous to the proof of the halt property for name elimination (\refthmp{tar-simulation}{halt}).

\item Note that if $\tm$ is $\totar$-normal then, by the halt property $\naming\tm$ is $\toint$-normal, against hypotheses. Then $\tm \totar \tmthree$ for some $\tmthree$. By projection (Point 1), $\namingp{\tm}\emptylist\emptylist \tosou \naming{\tmthree}\emptylist\emptylist$ and it is the same kind of step than on $\tarcal$. Since $\toint$ is deterministic (\reflemma{intermediate-determinism}), $\naming{\tmthree}\emptylist\emptylist=\tmtwo$. 

\item
By induction on $\tm$. Cases:
\begin{itemize}
\item \emph{Variable in $\tuvartwo$}, that is, $\tm=\vartwo_i$. Then, $\naming{\clc{\vartwo_i}\tuvartwo\tuvar} = \naming{\proj_i\lvar} = \vartwo_i$.

\item \emph{Variable in $\tuvar$}, that is, $\tm=\var_i$. Then, $\naming{\clc{\var_i}\tuvartwo\tuvar} = \naming{\proj_i\svar} = \var_i$.

\item \emph{\Wrapt}, that is, $\tm = \pack{\dabsv\varthree\varfour\tm}\bag$ with $\norm\tuvarthree=n$ and $\norm\tuvarfour=m$. Then:
\begin{center}
$\begin{array}{lllllll}
&&\naming{\clc\tm\tuvartwo\tuvar} 
\\[6pt]
& = & \naming{\clc{ \pack{\dabsv\varthree\varfour\tm}\bag }\tuvartwo\tuvar}
\\[6pt]
& = & \naming{ \packnm{ \clc\tm\tuvarthree\tuvarfour }{ \clc\bag\tuvartwo\tuvar } }
\\[6pt]
& = &  \pack{ \dabsv{\varthree'}{\varfour'}\namingp{\clc\tm\tuvarthree\tuvarfour}{\tuvarthree'}{\tuvarfour'} }{ \naming{\clc\bag\tuvartwo\tuvar} }  & \mbox{with }\tuvarthree'\mbox{, }\tuvarfour'\mbox{ fresh, }\norm{\tuvarthree'}=n\mbox{ and }\norm{\tuvarfour}=m
\\[6pt]
& =_\reflemmaeq{names-elim-stability-by-renaming} &  \pack{ \dabsv{\varthree'}{\varfour'}\namingp{\clc{\tm\isubp\tuvarthree{\tuvarthree'}\tuvarfour{\tuvarfour'}}{\tuvarthree'}{\tuvarfour'}}{\tuvarthree'}{\tuvarfour'} }{ \naming{\clc\bag\tuvartwo\tuvar} } 
\\[6pt]
(\ih) & =_\alpha &  \pack{ \dabsv{\varthree'}{\varfour'}\tm\isubp\tuvarthree{\tuvarthree'}\tuvarfour{\tuvarfour'} } \bag  
\\[6pt]
 & =_\alpha &  \pack{\dabsv\varthree\varfour\tm}\bag
 \\[6pt]
 & =\tm 
\end{array}$\qedhere
\end{center}
\end{itemize}
\end{enumerate}
\end{proof}

The following theorem states that the reverse transformation $\unlali{\namingp\tm\emptylist\emptylist}$ of closure conversion (obtained by composing naming and unwrapping) is a termination-preserving strong bisimulation. It is our correctness result for closure conversion. The proof is obtained by simply composing the two bisimulations results for unwrapping (\refthm{rev-transl-implem-bricks}) and naming (\refthm{tar-int-implementation}).

\begin{theorem}[Target-source termination-preserving strong bisimulation]
\label{thm:final-bisimulation}
Let $\tm\in\tarcal$ be closed.
\begin{enumerate}
\item \emph{Projection}: if $\tm\totbv\tmtwo$ then $\unlali{\namingp\tm\emptylist\emptylist} \tobv \unlali{\namingp\tmtwo\emptylist\emptylist}$ and if $\tm\totproj\tmtwo$ then $\unlali{\namingp\tm\emptylist\emptylist} \toproj \unlali{\namingp\tmtwo\emptylist\emptylist}$ .
\item \emph{Halt}: $\tm$ is $\totar$-normal  if and only if $\unlali{\namingp\tm\emptylist\emptylist}$ is  $\tosou$-normal.
\item \emph{Reflection}: if $\unlali{\namingp\tm\emptylist\emptylist} \tobv \tmtwo$ then there exists $\tmthree$ such that $\tm\totbv\tmthree$  and $\unlali{\namingp\tmthree\emptylist\emptylist} = \tmtwo$,  and if $\unlali{\namingp\tm\emptylist\emptylist} \toproj \tmtwo$ then there exists $\tmthree$ such that $\tm\totproj\tmthree$  and $\unlali{\namingp\tmthree\emptylist\emptylist} = \tmtwo$.

\item \emph{Inverse}: if $\tmtwo\in\soucal$ then $\unlali{\namingp{\clc{\lali\tmtwo}\emptylist\emptylist}\emptylist\emptylist}=\tmtwo$.
\end{enumerate}
\end{theorem}
\section{Proofs and Auxiliary Notions of \refsect{prelim-am} (Preliminaries About Abstract Machines)}
\label{sect:app-prelim-am}

\begin{lemma}[One-step simulation]
  \label{l:one-step-simulation}
  Let $\mach=(\States, \tomach, \compil\cdot, \decode\cdot)$ be a machine and $\tostrat$ be a deterministic strategy forming an implementation system. 
  For any state $\state$ of $\mach$, if $\decode\state \tostrat \tmtwo$ is a step of label $\lab$ then there is a state $\statetwo$ of $\mach$ such that $\state \tomacho^*\tomachpr \statetwo$ where the last transition has label $\lab$ and $\decode{\statetwo} = \tmtwo$.
\end{lemma}

\begin{proof}
  For any state $\state$ of $\mach$, let $\nfov{\state}$ be a normal form of $\state$ with respect to $\tomacho$: such a state exists because overhead transitions terminate (by the overhead termination property of implementation systems).
  Since $\tomacho$ is mapped on identities (by overhead transparency), one has $\decode{\nfov{\state}} = \decode\state$. 
The hypothesis then becomes $\decode\state = \decode{\nfov{\state}} \tostrat \tmtwo$. By the halt property, $\nfov{\state}$ cannot be a final state (successful or clash), otherwise $\decode\state = \decode{\nfov{\state}}$ could not reduce. Since $\nfov{\state}$ is $\tomachov$-normal, we have $\nfov{\state}\tomachpr \statetwo$ for some $\statetwo$. By principal projection, $\decode{\nfov{\state}} \tostrat \decode\statetwo$, and the step and the transition have the same label. By determinism of $\tostrat$, $\decode\statetwo = \tmtwo$.
\end{proof}

\gettoappendix{thm:abs-impl}
\begin{proof}
  According to \refdef{implem}, given a term $\tm\in\xcal$, we have to show that:
  \begin{enumerate}
\item \label{p:exec-to-deriv} \emph{Runs to evaluations with $\beta$-matching}: for any $\mach$-run $\exec: \compil\tm \tomachine^* \state$ there exists a 
$\tostrat$-evaluation $\deriv: \tm \tostrat^* \decode\state$ such that $\sizep\deriv\lab = \sizep\exec\lab$ for every label $\lab$ of steps in $\xcal$. Additionally, if $\state$ is a successful state then $\decode\state$ is a $\tostrat$-normal form.

    \item \label{p:deriv-to-exec} \emph{Evaluations to runs with $\beta$-matching}: for every $\tostrat$-evaluation $\deriv: \tm \tostrat^* \tmtwo$ there exists a 
$\mach$-run $\exec: \compil\tm \tomachine^* \state$ such that $\decode\state = \tmtwo$ and $\sizep\deriv\lab = \sizep\exec\lab$ for every label $\lab$ of steps in $\xcal$. Additionally, if $\tmtwo$ is a $\tostrat$-normal form then there exists a successful state $\statetwo$ such that $\state \tomacho^* \statetwo$.
  \end{enumerate}

  \paragraph{Proof of \refpoint{exec-to-deriv}}  By induction on the number of principal transition $\sizepr\run \in \nat$ in $\run$.
  
  If $\sizepr\run = 0$ then $\run \colon \compil\tm \tomacho^* \state$ and hence $\decode{\compil\tm} = \decode\state$ by overhead transparency (\refpoint{def-overhead-transparency} of \refdef{implementation}).
  Moreover, $\tm = \decode{\compil\tm}$ since read-back is the inverse of initialization on initial states, therefore the statement holds by taking the empty (\ie without steps) evaluation $\deriv$ with starting (and end) term $\tm$.
  
  Suppose $\sizepr\run > 0$: then, $\run \colon \compil{\tm} \tomachine^* \state$ is the concatenation of a run $\runtwo \colon \compil{\tm} \tomachine^* \statetwo$ followed by a run $\runthree \colon \statetwo \tomachpr \statethree \tomacho^* \state$.
  By \ih applied to $\runtwo$, there exists an evaluation $\derivtwo \colon \tm \tostrat^* \decode\statetwo$ with $\sizep\runtwo\lab = \sizep\derivtwo\lab$.
  By principal projection (\refpoint{def-beta-projection} of \refdef{implementation}) and overhead transparency (\refpoint{def-overhead-transparency} of \refdef{implementation}) applied to $\runthree$, one has $\derivthree \colon \decode\statetwo \tostrat \decode\statethree = \decode\state$ for a $\tostrat$-step of the same label $\lab'$ as the transition $\statetwo \tomachpr \statethree$.
  Therefore, the evaluation 
  $\deriv$ defined as the concatenation of $\derivtwo$ and $\derivthree$ is such that $\deriv \colon \tm \tostrat^* \decode\state$ Moreover:
  \begin{itemize}
  \item  $\sizep\deriv{\lab'} = \sizep\derivtwo{\lab'} + \sizep\derivthree{\lab'} = \sizep\runtwo{\lab'} + 1 = \sizep\run{\lab'}$.
  
  \item  $\sizep\deriv{\lab} = \sizep\derivtwo{\lab} + \sizep\derivthree{\lab} = \sizep\runtwo{\lab} + 0 = \sizep\run{\lab}$ for every label $\lab\neq \lab'$.
  \end{itemize}
 If $\state$ is a successful state, by the halt property (\refpoint{def-halt} of \refdef{implementation}) $\decode\state$ is a $\tostrat$-normal form.

  \paragraph{Proof of \refpoint{deriv-to-exec}}  By induction on $\size\deriv \in \nat$.

  If $\size\deriv = 0$ then $\tm = \tmtwo$.
  Since read-back is the inverse of initialization on initial states, one has $\decode{\compil{\tm}} = \tm$.
  The statement holds by taking the empty (\ie without transitions) run $\run$ with initial (and final) state $\compil{\tm}$.
  
  Suppose $\size\deriv > 0$: so, $\deriv\colon \tm \tostrat^* \tmtwo$ is the concatenation of an evaluation $\derivtwo \colon \tm \tostrat^* \tmtwop$ followed by the step $\tmtwop \tostrat \tmtwo$.
  By \ih, there exists a $\mach$-run $\runtwo\colon \compil\tm \tomachine^* \statetwo$ such that $\decode\statetwo = \tmtwop$ and $\sizep\derivtwo\lab = \sizep\runtwo\lab$ (the state $\statetwo$ here is not the one of the \emph{additionally} part of the statement).
By  one-step simulation (\reflemma{one-step-simulation}, since $\decode{\statetwo} \tostrat \tmtwo$ and $\tostrat$ and $\mach$ form an implementation system), there is a state $\state$ of $\mach$ such that $\statetwo \tomacho^*\tomachpr \state$ with its last transition having the same label $\lab'$ of the step $\decode{\statetwo} \tostrat \tmtwo$, and such that $\decode\state = \tmtwo$.
  Therefore, the run $\run \colon \compil\tm \tomachine^*\statetwo \tomacho^*\tomachpr \state$ is such that:
    \begin{itemize}
  \item  $\sizep\run{\lab'} = \sizep\runtwo{\lab'} +1 = \sizep\derivtwo{\lab'} + 1 = \sizep\deriv{\lab'}$;
  
  \item  $\sizep\run{\lab} = \sizep\runtwo{\lab} = \sizep\derivtwo{\lab} = \sizep\deriv{\lab}$; for every label $\lab\neq \lab'$.
  \end{itemize}
If $\tmtwo$ is a $\tostrat$-normal form, then consider $\nfov\state$. If $\nfov\state$ is not final, then by overhead transparency $\decode{\nfov\state}=\decode\state=\tmtwo$ and by principal projection $\tmtwo$ is not $\tostrat$-normal, absurd. Then it is a final state. By the halt property, $\nfov\state$ is a successful state, providing the state $\statetwo$ of the statement.
\qedhere
\end{proof}

 \section{Proofs and Auxiliary Notions of \refsect{TAM} (A Machine for $\soucal$: the Tupled Abstract Machine)}
 \label{sect:app-TAM}

\begin{definition}[Clashes]
A state $\state$ is a \emph{clash} if it has one of the following forms:
\begin{itemize}
\item \emph{Clashing projection}: $\state = \pair\clos{\proj_i\cons\stack}$ and ($\clos$ is not a tuple or it is but $\norm\clos < i$);
\item \emph{Clashing abstraction}: $\state = \pair{\evlab{\pair{\la{\tuvar}\tm}\env}}{\clos\cons\stack}$ and ($\clos$ is not a tuple or it is but $\norm{\tuvar}\neq \norm\clos$);
\item \emph{Clashing tuple}: $\state = \pair{\vv\clos}{\evlab{\pair{\vv\tm}\env}\cons\stack}$.
\end{itemize}
A state $\state$ is \emph{clash-free} when, co-inductively:
\begin{itemize}
\item $\state$ is not a clash;
\item If $\state \totam \tmtwo$ then $\tmtwo$ is clash-free.
\end{itemize}
\end{definition}

\gettoappendix{l:TAM-closure-invariant}
\begin{proof}
By induction on the length of the execution ending on $\state$. For each transition, it immediately follows from the \ih \qedhere
\end{proof}

\begin{lemma}\label{l:TAM-decode-laxclosure}
For every \TAM flagged closure $\laxclos$, $\decode{\laxclos}$ is a term. In particular:
\begin{enumerate}
	\item\label{p:TAM-decode-laxclosure-abstraction} If $\laxclos = \flag(\la{\tuvar}\tm,\env)$, then $\decode{\flag(\la{\tuvar}\tm, \env)} = \la{\var}\decode{(\tm, \env)}$.
	
	\item\label{p:TAM-decode-laxclosure-application}If $\laxclos = \flag(\tm\tmtwo,\env)$, then $\decode{\flag(\tm\tmtwo, \env)} = \decode{(\tm, \env)} \, \decode{(\tmtwo, \env)}$.
	
	\item\label{p:TAM-decode-laxclosure-projection} If $\laxclos = \flag(\proj_{i}\tm,\env)$, then $\decode{\flag(\proj_i\tm, \env)} = \proj_i\decode{(\tm, \env)}$.
		
	\item\label{p:TAM-decode-laxclosure-variable} If $\laxclos = \flag(\var,\env)$ and $\var\in\dom\env$, then $\decode{\flag(\var, \env)} = \decode{\env(\var)}$.
\end{enumerate}
\end{lemma}

\begin{proof}
By induction on the length of the environment $\env$.
\end{proof}

\gettoappendix{l:TAM-decoding-invariants}
\begin{proof}
\hfill
\begin{enumerate}
		\item By induction on $\clos$. The base case $\clos=\evsym\pair{\la\tuvar\tm}\env$ follows from \reflemmap{TAM-decode-laxclosure}{abstraction}. The inductive case $\clos=\vv\clos$ follows from the \ih
		\item By induction on $\stack$. Cases:
		\begin{itemize}
		\item \emph{Empty} $\stack=\emptystack$. Then $\decode\stack = \ctxhole$ is an evaluation context.
		\item \emph{Projection} $\stack=\proj_{i}\cons\stacktwo$. Then $\decode\stack=\decode\stacktwo\ctxholep{\proj_{i} \ctxhole}$. By \ih, $\decode\stacktwo$ is an evaluation context. Note that $\proj_{i} \ctxhole$ is also an evaluation context. By the composition of evaluation contexts (\reflemma{source-evctxs-compose}), so does $\decode\stack=\decode\stacktwo\ctxholep{\proj_{i} \ctxhole}$.
		
		\item \emph{Closure} $\stack= \decode{\clos \cons \stacktwo}$. By \ih, $\decode\stacktwo$ is an evaluation context. By the value invariant (Point 1), $\decode\clos$ is a value, thus $\ctxhole\decode\clos$ is an evaluation context. By the composition of evaluation contexts (\reflemma{source-evctxs-compose}), so does $\decode\stack = \decode\stacktwo\ctxholep{\ctxhole \decode\clos}$.

		\item \emph{Non-evaluated term} $\stack =\decode{\nvlab{\pair\tm\env} \cons \stacktwo}$. By \ih, $\decode\stacktwo$ is an evaluation context. Note that $\decode{\nvlab{\pair\tm\env}} \ctxhole$ is also an evaluation context. By the composition of evaluation contexts (\reflemma{source-evctxs-compose}), so does
		Then $\decode\stack= \decode\stacktwo\ctxholep{\decode{\nvlab{\pair\tm\env}} \ctxhole}$.
		
		\item \emph{Tuple} $\stack=\pair{\tuple{\tm_{1},\mydots,\tm_{n},\machctxhole,\clos_{1},\mydots,\clos_{m}}}\env \cons \stacktwo$. By \ih, $\decode\stacktwo$ is an evaluation context. By the value invariant (Point 1), $\decode{\clos_i}$ is a value for $1\leq i\leq m$, and so $\tuple{\decode{\pair{\tm_{1}}\env},\mydots, \decode{\pair{\tm_{n}}\env},\ctxhole,\decode{\clos_{1}},\mydots, \decode{\clos_{m}}}$ is an evaluation context. By the composition of evaluation contexts (\reflemma{source-evctxs-compose}),  so does $\decode\stack= \decode\stacktwo\ctxholep{\tuple{\decode{\pair{\tm_{1}}\env},\mydots, \decode{\pair{\tm_{n}}\env},\ctxhole,\decode{\clos_{1}},\mydots, \decode{\clos_{m}}  }}$.\qedhere
		\end{itemize}
	\end{enumerate}
\end{proof}

\begin{proposition}[Source implementation system]
	\label{prop:tam-requir-implem}
	Let $\state$ be a \TAM reachable state.
	\begin{enumerate}
		\item \emph{Principal project.}: for $a \in \{\betav, \proj\}$, if $\state \tomachhole{\evsym a} \statetwo$ then $\decode\state \Rew{a} \decode{\statetwo\!}\!$.
		\item \emph{Overhead transparency}: if $\state \tomacho \statetwo$ then $\decode\state = \decode\statetwo$.
		\item \emph{Overhead termination}: $\tomacho$ terminates.
		\item \emph{Halt}: successful states read back to $\tosou$-normal terms, and clash states to clashes of $\soucal$.
	\end{enumerate}
\end{proposition}
\begin{proof}
\hfill
\begin{enumerate}
		\item 
		\begin{enumerate}
			\item \emph{Beta}: $\pairstate{\evsym (\la\tuvar\tm,  \env)}{\vv\clos \cons\stack}
			\tomachbeta \pairstate{
			\nvsym (\tm, \esub{\tuvar}{\vv\clos} \cons \env)} {\stack} $
			with $\norm{\tuvar} =\norm{\vv\clos}$. Then,
			\begin{center}$\begin{array}{llll}
				\decode{\pairstate{\evsym (\la\tuvar\tm,  \env)}{\vv\clos \cons\stack}}
				&= &\decodep\stack{\decode{\evsym (\la\tuvar\tm,  \env)} \decode{\vv{\clos}}}
				\\[5pt] &=_{\reflemmaeqp{TAM-decode-laxclosure}{abstraction}}&
				 \decodep\stack{(\la\tuvar\decode{(\tm,  \env)}) \decode{\vv{\clos}}}
				\\[5pt] &\tobv &
				\decodep\stack{\decode{(\tm,  \env)} \isub{\tuvar}{\decode{\vv{\clos}}}}
				\\[5pt] &= &
				\decodep\stack{\decode{\nvsym (\tm, \esub{\tuvar}{\vv\clos} \cons \env)}}
				\\[5pt] &=&
				\decode{\pairstate{
					\nvsym (\tm, \esub{\tuvar}{\vv\clos} \cons \env)} {\stack}} .
			\end{array}$\end{center}
		The $\tobv$ step is correct because $\decodep\stack{\ctxhole \decode{\vv{\clos}}}$ is an evaluation context by the read back invariants (\reflemma{TAM-decoding-invariants}).
			
			\item \emph{Projection}: $\pairstate{\vv\clos} {\proj_i \cons \stack}
			\tomachproj 
			\pairstate{\clos_i}{\stack} $ with $i \leq \norm{\vv\clos}$.
			Then, 
			\begin{center}$\begin{array}{llll}
				\decode{\pairstate{\vv\clos} {\proj_i \cons \stack}}
				& = \decodep\stack{\proj_{i} \ctxholep{\decode{\vv\clos}}]}
				\toproj \decodep\stack{\decode{\clos_{i}}} =
				\decode{\pairstate{\clos_i}{\stack}}. 
			\end{array}$\end{center}
			The $\toproj$ step is correct because $\decode{\stack}$ is an evaluation context by the read back invariants (\reflemma{TAM-decoding-invariants}).
		\end{enumerate}
	
		\item \begin{enumerate}
			\item $\pairstate{\nvsym (\tm\tmtwo, \env)}{\stack }
			\tomachseaonenv
			\pairstate{\nvsym (\tmtwo, \env)} { \nvlab{\pair\tm\env} \cons\stack }$.
			Then, 
			\begin{center}$\begin{array}{llll}
				\decode{\pairstate{\nvsym (\tm\tmtwo, \env)}{\stack }}
				&= &\decode{\stack } \ctxholep{\decode{\nvsym (\tm\tmtwo, \env)}}
				\\[5pt]& =_{\reflemmaeqp{TAM-decode-laxclosure}{application}} &
				 \decode{\stack}\ctxholep{\decode{\nvlab{\pair\tm\env}}\ctxholep{\decode{\nvsym (\tmtwo, \env)}}}
				\\[5pt]&
				= &\decode{\pairstate{\nvsym (\tmtwo, \env)} { \nvlab{\pair\tm\env} \cons\stack }}.
			\end{array}$\end{center}
		
			\item $\pairstate{\nvsym (\proj_i \tm , \env) }{ \stack} 
			 \tomachseatwonv 
			\pairstate{\nvsym (\tm,\env)}{ \proj_i \cons \stack }$. Then,
			\begin{center}$\begin{array}{llll}
				\decode{\pairstate{\nvsym (\proj_i \tm , \env) }{ \stack} }
				&= \decode{\stack}\ctxholep{\decode{\nvsym (\proj_i \tm , \env) } }
				=_{\reflemmaeqp{TAM-decode-laxclosure}{projection}}
				\decode{\stack}\ctxholep{\proj_i \decode{\nvsym (\tm , \env) } }
				=
				\decode{\pairstate{\nvsym (\tm,\env)}{ \proj_i \cons \stack }}.
			\end{array}$\end{center}
		
			\item $\pairstate{\nvsym (\tuple{\mydots,\tm_n}, \env)} {\stack} 
			 \tomachseathreenv 
			\pairstate{\nvsym   (\tm_n, \env)  }{\pair{\tuple{\mydots,\machctxhole}}\env \cons \stack}$ with  $ n >0$. Then,
			\begin{center}$\begin{array}{llll}
				\decode{\pairstate{\nvsym (\tuple{\mydots,\tm_n}, \env)}{\stack} }
				&= &\decode{\stack}\ctxholep{\decode{\nvsym (\tuple{\mydots,\tm_n}, \env)}}
				\\[5pt]&=& 
				\decode{\stack}\ctxholep{\tuple{ \mydots, \decode{\nvsym (\tm_n, \env)}}}
				\\[5pt]&=&
				\decode{\pairstate{\nvsym   (\tm_n, \env)  }{\pair{\tuple{\mydots,\machctxhole}}\env \cons \stack}}.
			\end{array}$\end{center}
		
			\item $\pairstate{\nvsym (\tuple{}, \env) }{\stack}  \tomachseafournv 
			\pairstate{\evsym (\tuple{}, \emptyenv) } {\stack} $. Then, 
			\begin{center}$\begin{array}{llll}
				\decode{\pairstate{\nvsym (\tuple{}, \env) }{\stack}}  
				&= 
				\decode{\stack}\ctxholep{\decode{ (\tuple{}, \env) }}
				=
				\decode{\pairstate{\evsym (\tuple{}, \emptyenv) } {\stack}}.
			\end{array}$\end{center}
		
			\item $\pairstate{\nvsym  (\la\tuvar\tm, \env)} {\stack} 
			\tomachseafivenv 
			\pairstate{\evsym (\la\tuvar\tm , \env)} {\stack }$. Then, 
			\begin{center}$\begin{array}{llll}
				\decode{\pairstate{\nvsym  (\la\tuvar\tm, \env)} {\stack} }
				&=
				\decode{\stack} \ctxholep{\decode{ (\la\tuvar\tm, \env)}} 
				=
				\decode{\pairstate{\evsym (\la\tuvar\tm , \env)} {\stack }}.
			\end{array}$\end{center}
		
			\item $\pairstate{\nvsym (\var, \env)} {\stack} 
			\tomachsub
			\pairstate{\env(\var)} {\stack} $. Then,
			\begin{center}$\begin{array}{llll}
				\decode{\pairstate{\nvsym (\var, \env)} {\stack} }
				&=
				\decode{\stack}\ctxholep{\decode{\nvsym (\var, \env)}}
				=_{\reflemmaeqp{TAM-decode-laxclosure}{variable}}
				\decode{\stack}\ctxholep{\decode{\env(\var)}}
				=
				\decode{\pairstate{\env(\var)} {\stack}}.
			\end{array}$\end{center}
		
			\item $\pairstate{\clos}{ \nvlab{\pair\tm\env} \cons\stack} 
			\tomachseaoneev
			\pairstate{\nvsym (\tm,\env)} {\clos \cons\stack}$. Then, 
			\begin{center}$\begin{array}{llll}
				\decode{\pairstate{\clos}{ \nvlab{\pair\tm\env} \cons\stack} }
				&=
				\decode{\stack}\ctxholep{\decode{\nvsym (\tm,\env)} \, \decode{\clos}}
				=
				\decode{\pairstate{\nvsym (\tm,\env)} {\clos \cons\stack}}.
			\end{array}$\end{center}
		
			\item $\pairstate{\clos} {\pair{\tuple{\mydots,\tm,\machctxhole,\mydots}}\env \cons\stack }
			\tomachseasixev
			\pairstate{\nvsym (\tm ,\env) }  {\pair{\tuple{\mydots,\machctxhole,\clos,\mydots}}\env \cons\stack }$. Then, 
			\begin{center}$\begin{array}{llll}
				\decode{\pairstate{\clos} {\pair{\tuple{\mydots,\tm,\machctxhole,\mydots}}\env \cons\stack }}
				&= \decode{\stack}\ctxholep{\tuple{\mydots, \pair{\tm}{\env}, \decode{\clos}, \mydots}}
				=
				\decode{\pairstate{\nvsym (\tm ,\env) }  {\pair{\tuple{\mydots,\machctxhole,\clos,\mydots}}\env \cons\stack }}.
			\end{array}$\end{center}
		
			\item $\pairstate{\clos}{ \pair{\tuple{\machctxhole,\mydots}}\env \cons\stack} 
			\tomachseathreeev
			\pairstate{\tuple{\clos,\mydots}} {\stack} $. Then, 
			\begin{center}$\begin{array}{llll}
				\decode{\pairstate{\clos}{ \pair{\tuple{\machctxhole,\mydots}}\env \cons\stack} }
				&=
				\decode{\stack} \ctxholep{\tuple{\decode{\clos}, \mydots}}
				=
				\decode{\pairstate{\tuple{\clos,\mydots}} {\stack} }.
			\end{array}$\end{center}
		\end{enumerate}
		
		\item The proof of this point is omitted here, as it is a consequence of the complexity analysis of  \refsect{source-complexity}, see \refprop{TAM-number-of-trans}.		
	
		\item 
		Let $\pair\laxclos\stack$ be a final state. Two cases for $\laxclos$:
		\begin{itemize}
		\item \emph{Non-evaluated closure} $\nvclos$. Then the only possibility is that $\nvclos=\nvlab{\pair\var\env}$ with $\var\notin\dom\env$, but by the closure invariant (\reflemma{TAM-closure-invariant}) this is impossible. 
		
				\item \emph{Evaluated closure} $\clos$. Then consider the stack $\stack$. 
				\begin{itemize}
				\item \emph{$\stack$ is empty}. Then $\state$ reads back to a value $\decode{\pairstate{\clos}{\emptyenv}} = \decode{\clos}$ by the read back invariants (\reflemma{TAM-decoding-invariants}.1), which by harmony (\reflemma{source-harmony}) is a normal form.
				\item \emph{$\stack$ has a non-evaluated closure $\nvclos$ on top}, that is, $\stack = \nvclos\cons\stacktwo$. Then transition $\tomachseaoneev$ applies, and the state was not final, absurd.
				\item \emph{$\stack$ has a tuple $\pair{\tuple{\mydots,\tm,\machctxhole,\mydots}}\env$ on top}, that is, $\stack = \pair{\tuple{\mydots,\tm,\machctxhole,\mydots}}\env\cons\stacktwo$. Then transition $\tomachseasixev$ applies, and the state was not final, absurd.
				\item \emph{$\stack$ has a tuple $\pair{\tuple{\machctxhole,\mydots}}\env$ on top}, that is, $\stack = \pair{\tuple{\machctxhole,\mydots}}\env \cons\stacktwo$. Then transition $\tomachseathreeev$ applies, and the state was not final, absurd.

				\item \emph{$\stack$ has a closure $\clostwo$ on top}, that is, $\stack = \clostwo\cons\stacktwo$. Cases of $\clos$:
				\begin{itemize}
				\item \emph{$\clos$ is a tuple}. Then $\state$ is a clash state (clashing tuple) and its read back $\decode\state = \decode\stacktwo\ctxholep{\decode\clos\, \decode\clostwo}$ is the corresponding kind of term clash.
				\item  \emph{$\clos$ is an abstraction closure $\evlab{\pair{\la\tuvar\tm}\env}$}. Two sub-cases.
				\begin{itemize}
				\item  \emph{$\clostwo$ is not a tuple}. Then $\state$ is a clash state (clashing abstraction with wrong argument) and its read back is the corresponding kind of term clash.
				\item  \emph{$\clostwo$ is a tuple}. Two sub-cases:
				\begin{enumerate}
				\item  $\norm{\tuvar}\neq \norm\clostwo$. Then $\state$ is a clash state (clashing abstraction with arity mismatch) and its read back is the corresponding kind of term clash.
				\item  $\norm{\tuvar}= \norm\clostwo$. This case is impossible because then $\state$ can do a $\tomachbetaev$ transition, which is absurd.
				\end{enumerate}
				\end{itemize}
				\end{itemize}
				\item \emph{$\stack$ has a projection $\proj_i$ on top}, that is, $\stack = \proj_i\cons\stacktwo$. Cases of $\clos$:
				\begin{itemize}
				\item \emph{$\clos$ is a tuple}. Two sub-cases.
				\begin{itemize}
				\item  $\norm{\tuvar}\neq \norm\clostwo$. Then $\state$ is a clash state (clashing projection with arity mismatch) and its read back is the corresponding case of term clash.
				\item  $\norm{\tuvar}= \norm\clostwo$. This case is impossible because then $\state$ can do a $\tomachprojev$ transition, which is absurd.
				\end{itemize}
				\item \emph{$\clos$ is an abstraction closure $\evlab{\pair{\la\tuvar\tm}\env}$}. Then $\state$ is a clash state (clashing projection with wrong argument) and its read back is the corresponding case of term clash.\qedhere
				\end{itemize}
				\end{itemize}
		\end{itemize}
	\end{enumerate}
\end{proof}

\gettoappendix{thm:sou-tam-implementation}
\begin{proof}
By \refprop{tam-requir-implem}, \TAM and $\tosou$ form an implementation system, thus \TAM implements $\tosou$ by the abstract theorem \refthm{abs-impl}.\qedhere
\end{proof}
 \section{Proofs and Auxiliary Notions of \refsect{LTAM} (A Machine for $\intcal$: the \LTAM)}
 \label{sect:app-LTAM}

\begin{definition}[Clashes]
A state $\state$ is a \emph{clash} if it has one of the following forms:
\begin{itemize}
\item \emph{Clashing projection}: $\state = \fourstate{\evsym\val} {\proj_i\cons\stack} \env \ars $ and ($\evval$ is not a tuple or it is but $\norm\evval < i$);
\item \emph{Clashing \wrapt}: $\state = \fourstate{\evsym \pack{ \dabsv\vartwo\var \nvsym\tm}{\vv{\evlab\val}}} {\evvaltwo\cons\stack} \env \ars$ and ($\evvaltwo$ not a tuple or it is but $\norm{\tuvar}\neq \norm\evvaltwo$);
\item \emph{Clashing tuple}: $\state = \fourstate{\evsym\vv\val} {\evlab\valtwo\cons\stack} \env\ars$.
\end{itemize}
\end{definition}

\gettoappendix{prop:LTAM-invariants}
\begin{proof}
\hfill
\begin{enumerate}
\item By induction on the length of the execution ending on $\state$. For the initial state, it follows from the requirement that compiled terms are prime (which is an extension of well-formed). For the inductive case, one considers each transition, and the statement for the target state immediately follows from the \ih about the source state.

\item By induction on the length of the execution ending on $\state$. For each transition, it immediately follows from the \ih, but for $\tomachsubwev$ which needs the well-formedness invariant (\reflemma{LTAM-invariants}) to prove that $\evsym  \pack{ \dabsv\vartwo\var\nvsym\tm}{\evlab\val}$ is closed.

\item The two parts: 
	\begin{enumerate}
		\item[(a)] $\fv\tm \cup \fv\stack\subseteq \dom\env$;
		\item[(b)] $\fv\stacktwo\subseteq \dom\envtwo$ for every entry $\pair\stacktwo\envtwo$ of the environment stack $\ars$.
	\end{enumerate}
are proved at the same time, by induction on the length of the execution ending on $\state$. For most transitions, it immediately follows from the \ih The relevant cases are (they show in particular that the two parts have to be proved simultaneously):
\begin{itemize}
\item In $\tomachsubwev$ and $\tomachsubvev$, we need the closed values invariant (Point 2) on the source state to prove part $a$ about $\env(\var)$ for the target state.
\item In $\tomachseasevenev$, we need part $b$ and the closed values invariant (Point 2) on the source state to prove part $a$ for the target state.
\item In $\tomachbetaev$, we need part $a$ on the source state to prove part $b$ on the target state. Part $a$ for the target state holds because of the well-formed invariant (Point 1).\qedhere
\end{itemize}
\end{enumerate}
\end{proof}

\begin{definition}[Simultaneous substitution on evaluation contexts]
Meta-level simultaneous substitution $\evctx  \isubp\tuvartwo\vvvalnone\tuvar\vvvalntwo$ on evaluation contexts $\ctx$ of $\intcal$ is defined as follows:
\begin{center}
$\begin{array}{rlllll}
\ctxhole\isubp\tuvartwo\vvvalnone\tuvar\vvvalntwo & \defeq & \ctxhole
\\
(\tm\, \evctx) \isubp\tuvartwo\vvvalnone\tuvar\vvvalntwo  & \defeq & \tm\isubp\tuvartwo\vvvalnone\tuvar\vvvalntwo \,\evctx \isubp\tuvartwo\vvvalnone\tuvar\vvvalntwo
\\
 (\evctx\,\val) \isubp\tuvartwo\vvvalnone\tuvar\vvvalntwo  & \defeq & \evctx\isubp\tuvartwo\vvvalnone\tuvar\vvvalntwo \, \val \isubp\tuvartwo\vvvalnone\tuvar\vvvalntwo
 \\
( \proj_i \evctx)\isubp\tuvartwo\vvvalnone\tuvar\vvvalntwo  & \defeq & \proj_i \evctx\isubp\tuvartwo\vvvalnone\tuvar\vvvalntwo
\\
\tuple{\tuv\tm, \evctx, \tuv\val}\isubp\tuvartwo\vvvalnone\tuvar\vvvalntwo  & \defeq &  \tuple{\tuv{\tm\isubp\tuvartwo\vvvalnone\tuvar\vvvalntwo},\machctxhole,\tuv{\val\isubp\tuvartwo\vvvalnone\tuvar\vvvalntwo}}
\end{array}$
\end{center}
\end{definition}

\gettoappendix{l:LTAM-decoding-invariants}
\begin{proof}
\label{lab:LTAM-read-back-properties-proof}
For constructor stacks, we proceed by induction on $\stack$. Cases:
		\begin{itemize}
		\item \emph{Empty} $\stack=\emptystack$. Then $\decode\stack = \ctxhole$ is an evaluation context.
		
		\item \emph{Projection} $\stack=\proj_{i}\cons\stacktwo$. Then $\decode\stack=\decodep\stacktwo{\proj_{i} \ctxhole}$. By \ih, $\decode\stacktwo$ is an evaluation context. Note that $\proj_{i} \ctxhole$ is also an evaluation context. By the composition of evaluation contexts (\reflemma{intermediate-evctxs-compose}), so does $\decode\stack=\decodep\stacktwo{\proj_{i} \ctxhole}$.
		
		\item \emph{Evaluated value} $\stack= \decode{\evlab\val \cons \stacktwo}$. By \ih, $\decode\stacktwo$ is an evaluation context. Note that $\ctxhole \decode\evval$ is also an evaluation context, because by Point 1 $\decode\evval$ is a value. By the composition of evaluation contexts (\reflemma{intermediate-evctxs-compose}), so does $\decode\stack = \decodep\stacktwo{\ctxhole \decode\evval}$.

		\item \emph{Non-evaluated term} $\stack =\decode{\nvlab\tm \cons \stacktwo}$. By \ih, $\decode\stacktwo$ is an evaluation context. Note that $\tm \ctxhole$ is also an evaluation context. By the composition of evaluation contexts (\reflemma{intermediate-evctxs-compose}), so does
		Then $\decode\stack= \decodep\stacktwo{\tm \ctxhole}$.
		
		\item \emph{Tuple} $\stack=\tuple{\tm_{1},\mydots,\tm_{n},\machctxhole,\evsym\val_{1},\mydots,\evsym\val_{m}} \cons \stacktwo$. By \ih, $\decode\stacktwo$ is an evaluation context. Note that $\tuple{\tm_{1},\mydots, \tm_{n},\ctxhole,\decode{\evval_{1}},\mydots, \decode{\evval_{m}}}$ also is an evaluation context, because by Point 1 $\decode{\evval_{m}}$ is a value. By the composition of evaluation contexts (\reflemma{intermediate-evctxs-compose}),  	so does $\decode\stack= \decodep\stacktwo{\tuple{\tm_{1},\mydots, \tm_{n},\ctxhole,\decode{\evval_{1}},\mydots, \decode{\evval_{m}}}}$.
		\end{itemize}

For activation stacks, we proceed by induction on $\ars$. Cases:
		\begin{itemize}
		\item \emph{Empty} $\ars=\emptystack$. Then $\decode\ars = \ctxhole$ is an evaluation context.
		
		\item \emph{Non-empty} $\ars=\pair\stack\env\cons\arstwo$. By \ih, $\decode\arstwo$ is an evaluation context. Since $\decode\stack$ is an evaluation context by the point about constructor stacks, $\decode{\stack}\sigma_\env$ is an evaluation context. By the composition of evaluation contexts (\reflemma{intermediate-evctxs-compose}), so does $\decodep\ars{\decode\stack\sigma_\env}$.\qedhere
		\end{itemize}
\end{proof}

\begin{proposition}[Intermediate implementation system]
	\label{prop:LTAM-impl-requirements}
	Let $\state$ be a \LTAM reachable state.
	\begin{enumerate}
		\item \emph{Principal project.}: for $a \in \{\betav, \proj\}$, if $\state \tomachhole{\evsym a} \statetwo$ then $\decode\state \Rew{\intprefix a} \decode{\statetwo\!}\!$.
		\item \emph{Overhead transparency}: if $\state \tomacho \statetwo$ then $\decode\state = \decode\statetwo$.
		\item \emph{Overhead termination}: $\tomacho$ terminates.
		\item \emph{Halt}: successful states read back to $\toint$-normal terms, and clash states read back to clashes of $\intcal$.
	\end{enumerate}
\end{proposition}
\begin{proof}
\hfill
\begin{enumerate}
		\item 
		\begin{itemize}
			\item \emph{Beta}: $\fourstate  {\evsym\pack{ \dabsv\vartwo\var\nvsym\tm}{\vv{\evlab{\val_1}}}} {\vv{\evlab{\val_2}} \cons\stack}  {\env}  {\ars}
\tomachbetaev
\fourstate{\nvsym  \tm}  {\emptystack}  {\envtwo}  {\pair\stack\env \cons \ars}$
with $\norm\tuvartwo=\norm{\vv{\evlab{\val_2}}}$, $\norm{\tuvar}=\norm{\vv{\evlab{\val_1}}}$, and $\envtwo \defeq \esub{\tuvartwo}{\vv{\evsym\valnone}}\,\esub{\tuvar}{\vv{\evsym\valntwo}}$.
Then,
			\begin{center}$\begin{array}{lclll}
				\decode{\fourstate  {\evsym\pack{ \dabsv\vartwo\var\nvsym\tm}{\vv{\evlab{\val_1}}}} {\vv{\evlab{\val_2}} \cons\stack}  {\env}  {\ars}}
				&= &
				\decodep\ars{\decode{\vv{\evsym\val_2} \cons\stack}\sigma_\env\ctxholep{ \pack{ \dabsv\vartwo\var\nvsym\tm}{\vv{\decode{\evval_1}}} \sigma_\env }}
				\\[5pt]&= &
				\decodep\ars{\decode\stack\sigma_\env\ctxholep{ \pack{ \dabsv\vartwo\var\nvsym\tm}{\vv{\decode{\evval_1}}} \sigma_\env \, \vv{\decode{\evval_2}}\sigma_\env }}
				\\[5pt] &=^* & 
				\decodep\ars{\decode\stack\sigma_\env\ctxholep{ \pack{ \dabsv\vartwo\var\nvsym\tm}{\vv{\decode{\evval_1}}} \, \vv{\decode{\evval_2}}}}
				\\[5pt] &\toibv & 
				\decodep\ars{\decode\stack\sigma_\env\ctxholep{ \tm \isubp{\tuvartwo}{\vv{\decode{\evval_1}}}{\tuvar}{\vv{\decode{\evval_2}}}}}
				\\[5pt] &= & 
				\decodep\ars{\decode\stack\sigma_\env\ctxholep{ \tm \sigma_{\envtwo} }}
				\\[5pt] & = &				
				\decode{ \fourstate{\nvsym  \tm}  {\emptystack}  {\envtwo}  {\pair\stack\env \cons \ars} } 
			\end{array}$\end{center}
		The $=^*$-step uses the closed values invariant (\reflemma{LTAM-invariants}). The $\tollbv$ step is correct because  $\decodep\ars{\decode\stack\sigma_\env}$ is an evaluation context by the read back properties (\reflemma{LTAM-decoding-invariants}).

			\item \emph{Projection}: $\fourstate {\evsym  \vv\val}   {\proj_i \cons \stack}  \env  \ars
 \tomachprojev 
\fourstate {\evsym  \val_i} \stack  \env  \ars$ with $1\leq i \leq \norm{\vv\val}$.
			Then, 
			\begin{center}$\begin{array}{lllll}
				\decode{ \fourstate {\evsym  \vv\val}   {\proj_i \cons \stack}  \env  \ars }
				& = & 
				\decodep\ars{\decode{\proj_i\cons\stack}\sigma_\env\ctxholep{\vv{\decode\evval\sigma_\env} }}
				\\[5pt] & = & 
				\decodep\ars{\decode\stack\sigma_\env\ctxholep{ \proj_i\vv{\decode\evval\sigma_\env} }}
				\\&\toproj&
				\decodep\ars{\decode\stack\sigma_\env\ctxholep{ \decode\evval_i\sigma_\env } } 
				\\&=&
				\decode{ \fourstate {\evval_i} \stack  \env  \ars } 
			\end{array}$\end{center}
			The $\toproj$ step is correct because $\decodep\ars{\decode{\stack\sigma_\env}}$ is an evaluation context by the read back properties (\reflemma{LTAM-decoding-invariants}).
		\end{itemize}

		\item \begin{itemize}
			\item $\fourstate {\nvsym  (\tm\,\tmtwo)} \stack \env \ars
\tomachseaonenv
\fourstate {\nvsym  \tmtwo}  {\nvlab\tm \cons\stack} \env \ars$.
			Then, 
			\begin{center}$\begin{array}{lllll}
				\decode{ \fourstate {\nvsym  (\tm\,\tmtwo)} \stack \env \ars }
				&= &\decodep\ars{\decode\stack\sigma_\env\ctxholep{ \tm\sigma_\env\,\tmtwo\sigma_\env}}
				\\[5pt] &=&
				\decodep\ars{\decode{\nvlab\tm \cons \stack}\sigma_\env\ctxholep{\tmtwo\sigma_\env}}
				\\[5pt] & = &
				\decode{ \fourstate {\nvsym  \tmtwo}  {\nvlab\tm \cons\stack} \env \ars }
			\end{array}$\end{center}
		
			\item $\fourstate		{\nvsym   \proj_i \tm}   \stack  \env \ars
 \tomachseatwonv 
\fourstate {\nvsym  \tm}  {\proj_i \cons \stack}  \env  \ars$. Then,
			\begin{center}$\begin{array}{lllllllll}
				\decode{ \fourstate		{\nvsym   \proj_i \tm}   \stack  \env \ars }
				&=& 
				\decodep\ars{\decode\stack\sigma_\env\ctxholep{ \proj_i \tm\sigma_\env}}
				\\&=& 
				\decodep\ars{\decode{\proj_i\cons\stack}\sigma_\env\ctxholep{  \tm\sigma_\env}}
				\\&=&
				\decode{ \fourstate {\nvsym  \tm}  {\proj_i \cons \stack}  \env  \ars }
			\end{array}$\end{center}
			
			\item $\fourstate {\nvsym  \tuple{\tm_1,\mydots,\tm_n}}  \stack  \env  \ars
 \tomachseathreenv 
\fourstate {\nvsym   \tm_n}  {\tuple{\nvsym\tm_1,\mydots,\machctxhole} \cons \stack}  \env  \ars$
 with  $\norm{\vv\tm} = n >0$. Then,
			\begin{center}$\begin{array}{lllll}
				\decode{ \fourstate {\nvsym  \tuple{\tm_1,\mydots,\tm_n}}  \stack  \env  \ars }
				&=& 
				\decodep\ars{\decode\stack\sigma_\env\ctxholep{ \tuple{\tm_1,\mydots, \tm_n}\sigma_\env} }
				\\[5pt] & = & 
				\decodep\ars{\decode{(\tuple{\nvsym\tm_1,\mydots,\machctxhole} \cons \stack)}\sigma_\env\ctxholep{ \tm_n\sigma_\env }}
				\\[5pt] & = & 
				\decode{ \fourstate {\nvsym   \tm_n}  {\tuple{\nvsym\tm_1,\mydots,\machctxhole} \cons \stack}  \env  \ars }
			\end{array}$\end{center}

			\item $\fourstate {\nvsym  \tuple{}}  \stack  \env  \ars
 \tomachseafournv 
\fourstate {\evsym   \tuple{}}  \stack  \env  \ars$. Then, 
			\begin{center}$\begin{array}{lllll}
				\decode{ \fourstate {\nvsym  \tuple{}}  \stack  \env  \ars } 
				&= 
				\decodep\ars{\decode\stack\sigma_\env\ctxholep{ \tuple{} }}
				=
				\decode{ \fourstate {\evsym   \tuple{}}  \stack  \env  \ars }
			\end{array}$\end{center}

			\item $\fourstate {\nvsym \pack{ \dabsv\vartwo\var\tm}\tuvartwo}  \stack  \env  \ars
 \tomachsubwev 
\fourstate { \evsym \pack{ \dabsv\vartwo\var\nsym\tm}{\vv{\env(\vartwo)}} } \stack  \env  \ars$
. Then, 
			\begin{center}$\begin{array}{lllll}
				\decode{ \fourstate {\nvsym \pack{ \dabsv\vartwo\var\tm}\tuvartwo}  \stack  \env  \ars }
				&=&
				\decodep\ars{\decode\stack\sigma_\env\ctxholep{ \pack{ \dabsv\vartwo\var\tm}\tuvartwo }\sigma_\env }
				\\[5pt] & = &
				\decodep\ars{\decode\stack\sigma_\env\ctxholep{ \pack{ \dabsv\vartwo\var\tm}{\vv{\decode{\env(\vartwo)}}} } }
				\\[5pt] & = &
				\decode{ \fourstate { \evsym \pack{ \dabsv\vartwo\var\nvsym\tm}{\vv{\env(\vartwo)}} } \stack  \env  \ars }
			\end{array}$\end{center}

			\item $ \fourstate{\nvsym  \var}  \stack \env \ars
\tomachsubvev
\fourstate {\env(\var)}  \stack  \env  \ars$. Then,
			\begin{center}$\begin{array}{llllllllll}
				\decode{ \fourstate{\nvsym  \var}  \stack \env \ars }
				&=&
				\decodep\ars{\decode\stack\sigma_\env\ctxholep{\var\sigma_\env} }
				&= & 
				\decodep\ars{\decode\stack\sigma_\env\ctxholep{\decode{\env(\var)}}}
				&=&
				\decode{ \fourstate {\env(\var)}  \stack  \env  \ars }
			\end{array}$\end{center}

			\item $\fourstate {\evsym  \val} {\nvlab\tm \cons\stack}  \env  \ars
\tomachseaoneev
\fourstate{\nvsym  \tm}  {\evlab\val \cons\stack}  \env  \ars$. Then, 
			\begin{center}$\begin{array}{lllll}
				\decode{ \fourstate {\evsym  \val} {\nvlab\tm \cons\stack}  \env  \ars }
				&=&
				\decodep\ars{\decode{\nvlab\tm \cons \stack}\sigma_\env\ctxholep{\decode\evval\sigma_\env}}
				\\[5pt] &=&
				\decodep\ars{\decode\stack\sigma_\env\ctxholep{\nvlab\tm\sigma_\env \, \decode\evval\sigma_\env}}
				\\[5pt] & = &
				\decodep\ars{\decode{\evsym\val \cons \stack}\sigma_\env\ctxholep{\nvlab\tm\sigma_\env}}
				\\[5pt] & = &
				\decode{ \fourstate{\nvsym  \tm}  {\evlab\val \cons\stack}  \env  \ars }
			\end{array}$\end{center}

			\item $\fourstate {\evsym  \val}  {\tuple{\mydots,\nvsym\tm,\machctxhole,\mydots} \cons\stack}  \env  \ars
\tomachseatupev
\fourstate {\nvsym  \tm}  {\tuple{\mydots,\machctxhole,\evlab\val,\mydots}\cons\stack}  \env  \ars$. Then, 
			\begin{center}$\begin{array}{lllll}
				\decode{ \fourstate		{\evsym  \val}  {\tuple{\mydots,\nvsym\tm,\machctxhole,\mydots} \cons\stack}  \decode\evval  \ars }
				&= &
				\decodep\ars{\decode{\tuple{\mydots,\nvsym\tm,\machctxhole,\mydots}\cons\stack}\sigma_\env\ctxholep{ \decode\evval\sigma_\env }}
				\\[5pt] &= &
				\decodep\ars{\decode\stack\sigma_\env\ctxholep{ \tuple{\mydots,\tm\sigma_\env,\decode\evval\sigma_\env,\mydots} }}
				\\[5pt] &= &
				\decodep\ars{\decode{\tuple{\mydots,\machctxhole,\evval,\mydots} \cons \stack}\sigma_\env\ctxholep{ \tm\sigma_\env }}
				\\[5pt] &= &
				\decode{ \fourstate {\nvsym  \tm}  {\tuple{\mydots,\machctxhole,\evlab\val,\mydots}\cons\stack}  \env  \ars }
			\end{array}$\end{center}

			\item $\fourstate {\evsym  \val}  {\tuple{\machctxhole,\mydots}  \cons\stack}  \env  \ars
\tomachseathreeev
\fourstate {\evsym  \tuple{\evlab\val,\mydots}}  \stack  \env  \ars$. Then, 
			\begin{center}$\begin{array}{lllll}
				\decode{ \fourstate {\evsym  \val}  {\tuple{\machctxhole,\mydots}  \cons\stack}  \env  \ars }
				&=&
				\decodep\ars{\decode{\tuple{\machctxhole,\mydots}\cons\stack}\sigma_\env\ctxholep{\tuple{\decode\evval\sigma_\env,\mydots}}}
				\\[5pt] &=&\decodep\ars{\decode\stack\sigma_\env\ctxholep{\tuple{\decode\evval\sigma_\env,\mydots}}}
				\\[5pt] &=&
				\decode{\fourstate {\evsym  \tuple{\evlab\val,\mydots}}  \stack  \env  \ars }
			\end{array}$\end{center}
			
			\item $\fourstate {\evsym  \val} \emptystack \env  {\pair\stack\envtwo\cons \ars}
 			\tomachseaactev
			\fourstate {\evsym  \val} \stack \envtwo \ars$. Then,
			\begin{center}$\begin{array}{lllll}
				\decode{ \fourstate {\evsym  \val} \emptystack \env  {\pair\stack\envtwo\cons \ars} }
				&=&
				\decodep\ars{\decode\stack\sigma_\env\ctxholep{\decode\evval\sigma_\env}}
				\\[5pt] &=^*&
				\decodep\ars{\decode\stack\sigma_\env\ctxholep{\decode\evval}}
				\\[5pt] &=^*&
				\decodep\ars{\decode\stack\sigma_\env\ctxholep{\decode\evval\sigma_\envtwo}}
				\\[5pt] &=&
				\decode{ \fourstate {\evsym  \val} \stack \envtwo \ars }
			\end{array}$\end{center}
Both $=^*$-steps use the closed values invariant (\reflemma{LTAM-invariants}).
		\end{itemize}
		
		\item The proof of this point is omitted here, as it is a consequence of the complexity analysis of  \refsect{int-complexity}. See \refprop{LTAM-number-of-trans}.			
		\item 
		Let $\fourstate{\flag\tm}\stack\env\ars$ be a final state. Two cases for $\flag\tm$:
		\begin{itemize}
		\item \emph{Non-evaluated focus} $\nvlab\tm$. Then the only possibilities are that $\tm=\var$ with $\var\notin\dom\env$ or $\tm=\pack{ \dabsv\vartwo\var\nvsym\tm}\tuvartwo$ with $\vartwo_i \notin\dom\env$ for some $i\leq\norm\tuvartwo$. By the closure invariant (\reflemma{LTAM-invariants}), this is impossible. 
		
		\item \emph{Evaluated focus} $\evlab\val$. Then consider the stack $\stack$. 
				\begin{itemize}
				\item \emph{$\stack$ is empty}. Then consider the activation stack $\ars$:
				\begin{itemize}
				\item \emph{$\ars$ is empty}. Then $\state$ reads back to a value $\decode{\fourstate{\evlab\val}\emptystack\env\emptystack} = \decode\evval\sigma_\env$.  By the closed values invariant (\reflemma{LTAM-invariants}), $\decode\evval\sigma_\env = \decode\evval$, which by the read back properties (\reflemma{LTAM-decoding-invariants}) is a value, which by harmony (\reflemma{intermediate-harmony}) is a $\toint$-normal form.
				\item \emph{$\ars$ is non-empty}. Then transition $\tomachseaactev$ applies, and the state was not final, absurd.
				\end{itemize}
				\item \emph{$\stack$ has a non-evaluated term $\nvlab\tm$ on top}, that is, $\stack = \nvlab\tm\cons\stacktwo$. Then transition $\tomachseaoneev$ applies, and the state was not final, absurd.
				\item \emph{$\stack$ has a tuple $\tuple{\mydots,\nvsym\tm,\machctxhole,\mydots}$ on top}, that is, $\stack = \tuple{\mydots,\nvsym\tm,\machctxhole,\mydots}\cons\stacktwo$. Then transition $\tomachseasixev$ applies, and the state was not final, absurd.
				\item \emph{$\stack$ has a tuple $\tuple{\machctxhole,\mydots} $ on top}, that is, $\stack = \tuple{\machctxhole,\mydots} \cons\stacktwo$. Then transition $\tomachseathreeev$ applies, and the state was not final, absurd.
				\item \emph{$\stack$ has an evaluated value $\evlab\valtwo$ on top}, that is, $\stack = \evlab\valtwo\cons\stacktwo$. Cases of $\evsym\val$:
				\begin{itemize}
				\item \emph{$\evsym\val$ is a tuple}. Then $\state$ is a clash state (clashing tuple) and its read back $\decode\state = 				\decodep\ars{\decode{\evsym\valtwo\cons\stack}\sigma_\env\ctxholep{\decode\evval}} = 	\decodep\ars{\decode\stack\sigma_\env\ctxholep{\decode\evval\,\decode\evvaltwo}}$ is the corresponding kind of term clash.
				\item  \emph{$\evsym\val$ is a \wrapt $\evsym\pack{ \dabsv\vartwo\var\nvsym\tm}{\vv{\evsym\valthree}}$}. Two sub-cases.
				\begin{itemize}
				\item  \emph{$\evsym\valtwo$ is not a tuple}. Then $\state$ is a clash state (clashing \wrapt with wrong argument) and its read back  is the corresponding kind of term clash.
				\item  \emph{$\evsym\valtwo$ is a tuple}. Two sub-cases:
				\begin{enumerate}
				\item  $\norm{\tuvar}\neq \norm{\evsym\valtwo}$. Then $\state$ is a clash state (clashing \wrapt with arity mismatch) and its read back  is the corresponding kind of term clash.
				\item  $\norm{\tuvar}= \norm{\evsym\valtwo}$. This case is impossible because then $\state$ can do a $\tomachbetaev$ transition, which is absurd.
				\end{enumerate}
				\end{itemize}
				\end{itemize}
				\item \emph{$\stack$ has a projection $\proj_i$ on top}, that is, $\stack = \proj_i\cons\stacktwo$. Cases of $\evsym\val$:
				\begin{itemize}
				\item \emph{$\evsym\val$ is a tuple}. Two sub-cases.
				\begin{itemize}
				\item  $\norm{\tuvar}\neq \norm{\evsym\val}$. Then $\state$ is a clash state (clashing projection with arity mismatch) and its read back  is the corresponding case of term clash.
				\item  $\norm{\tuvar}= \norm{\evsym\val}$. This case is impossible because then $\state$ can do a $\tomachprojev$ transition, which is absurd.
				\end{itemize}
				\item  \emph{$\evsym\val$ is a \wrapt $\evsym\pack{ \dabsv\vartwo\var\nvsym\tm}{\vv{\evsym\valthree}}$}. Then $\state$ is a clash state (clashing projection with wrong argument) and its read back  is the corresponding case of term clash.\qedhere
				\end{itemize}
				\end{itemize}
		\end{itemize}
	\end{enumerate}
\end{proof}

\gettoappendix{thm:int-ltam-implementation}
\begin{proof}
By \refprop{LTAM-impl-requirements}, \LTAM and $\toint$ form an implementation system on prime terms of $\intcal$ (because a requirement for the initial states of the \LTAM is that the term is prime), thus \LTAM implements $\toint$ by the abstract theorem \refthm{abs-impl}.\qedhere
\end{proof}
 \section{Proofs and Auxiliary Notions of \refsect{TTAM} (A Machine for $\tarcal$: the \TTAM)}
 \label{sect:app-TTAM}

\begin{definition}[Clashes]
A state $\state$ is a \emph{clash} if it has one of the following forms:
\begin{itemize}
\item \emph{Clashing projection}: $\state = \fourstate{\evsym\val} {\proj_i\cons\stack} \env \ars $ and ($\val$ is not a tuple or it is but $\norm\val < i$);
\item \emph{Clashing \wrapt}: $\state = \fourstate{\evsym \pack{\nvsym\tm}{\vv{\evlab\val}}} {\evvaltwo\cons\stack} \env \ars$ and $\evvaltwo$ not a tuple;
\item \emph{Clashing tuple}: $\state = \fourstate{\evsym\vv\val} {\evlab\valtwo\cons\stack} \env\ars$.
\end{itemize}
\end{definition}

\paragraph{Invariants.} The basic invariants of the \TTAM are given by the next proposition. To express the closure invariant, we need to extend the notion of norm of a term of $\tarcal$ to stacks, as follows:
\begin{center}
	$\begin{array}{llllll}
			\multicolumn{3}{c}{\textsc{\Lifted/source norms for stacks } (\vvar\in\set{\lvar,\svar})}
			\\
			\prnorm\stack &\defeq& \max\set{i\in\nat \,|\, \proj_i\vvar\mbox{ appears in $\stack$ out of \wrapts bodies}}
	\end{array}$
\end{center}
Moreover, for a tupled environment $\env = \vv{\evval_\lvar};\vv{\evval_\svar}$ we set $\prnorm\env \defeq \norm{\vv{\evval_\vvar}}$.

\begin{proposition}[Invariants]
\label{prop:TTAM-invariants}
Let $\state=\fourstate{\flag\tm}\stack\env\ars$ be a \TTAM reachable state and $\vvar\in\set{\lvar,\svar}$.
\begin{enumerate}
	\item \emph{Well-formedness}: all \wrapts in $\state$ are well-formed, that is, such that $\lnorm\tm=\norm\bag$ for $\nvsym\pack\tm\bag$ or $\evsym\pack{\nvsym\tm}{\evsym\bag}$.
	\item \emph{Closed values}: every value $\evlab\val$ in $\state$ is closed, that is, has no projected variable out of \wrapt bodies.
	\item \emph{Closure}: $\prnorm\tm \leq \prnorm\env$ and $\prnorm\stack \leq \prnorm\env$, and  $\prnorm\stacktwo \leq \prnorm\envtwo$ for every entry $\pair\stacktwo\envtwo$ of the activation stack $\ars$.
\end{enumerate}
\end{proposition}
\begin{proof}
\hfill
\begin{enumerate}
\item By induction on the length of the execution ending on $\state$. For the initial state, it follows from the requirement that compiled terms are prime (which is an extension of well-formed). For the inductive case, one considers each transition, and the statement for the target state immediately follows from the \ih about the source state.

\item By induction on the length of the execution ending on $\state$. For each transition, it immediately follows from the \ih

\item The two parts: 
	\begin{enumerate}
		\item[(a)] $\prnorm\tm \leq \prnorm\env$ and $\prnorm\stack \leq \prnorm\env$;
		\item[(b)] $\prnorm\stacktwo \leq \prnorm\envtwo$ for every entry $\pair\stacktwo\envtwo$ of the activation stack $\ars$.
	\end{enumerate}
are proved at the same time, by induction on the length of the execution ending on $\state$. For most transitions, it immediately follows from the \ih The relevant cases are (they show in particular that the two parts have to be proved simultaneously):
\begin{itemize}
\item In $\tomachsubwev$ and $\tomachsubvev$, we need the closed values invariant (Point 2) on the source state to prove part $a$ about $\env(\prvar)$ for the target state.
\item In $\tomachseasevenev$, we need part $b$ and the closed values invariant (Point 2) on the source state to prove part $a$ for the target state.
\item In $\tomachbetaev$, we need part $a$ on the source state to prove part $b$ on the target state. Part $a$ for the target state holds because of the well-formed invariant (Point 1).\qedhere
\end{itemize}
\end{enumerate}
\end{proof}

\begin{definition}[Projecting substitution on evaluation contexts]
Meta-level projecting (simultaneous) substitution $\evctx  \itsubp\tuvartwo\vvvalnone\tuvar\vvvalntwo$ on evaluation contexts $\ctx$ of $\tarcal$ is defined as follows:
\begin{center}
$\begin{array}{rlllll}
\ctxhole\itsubp\tuvartwo\vvvalnone\tuvar\vvvalntwo & \defeq & \ctxhole
\\
(\tm\, \evctx) \itsubp\tuvartwo\vvvalnone\tuvar\vvvalntwo  & \defeq & \tm\itsubp\tuvartwo\vvvalnone\tuvar\vvvalntwo \,\evctx \itsubp\tuvartwo\vvvalnone\tuvar\vvvalntwo
\\
 (\evctx\,\val) \itsubp\tuvartwo\vvvalnone\tuvar\vvvalntwo  & \defeq & \evctx\itsubp\tuvartwo\vvvalnone\tuvar\vvvalntwo \, \val \itsubp\tuvartwo\vvvalnone\tuvar\vvvalntwo
 \\
( \proj_i \evctx)\itsubp\tuvartwo\vvvalnone\tuvar\vvvalntwo  & \defeq & \proj_i \evctx\itsubp\tuvartwo\vvvalnone\tuvar\vvvalntwo
\\
\tuple{\tuv\tm, \evctx, \tuv\val}\itsubp\tuvartwo\vvvalnone\tuvar\vvvalntwo  & \defeq &  \tuple{\tuv{\tm\itsubp\tuvartwo\vvvalnone\tuvar\vvvalntwo},\machctxhole,\tuv{\val\itsubp\tuvartwo\vvvalnone\tuvar\vvvalntwo}}
\end{array}$
\end{center}
\end{definition}

\begin{lemma}[Read back properties]
\label{l:TTAM-decoding-invariants}
\hfill
	\begin{enumerate}
		\item
		\emph{Values}: $\decode{\evval}$ is a value of $\tarcal$  for every $\evval$ of the \TTAM.
		\item
		\emph{Evaluation contexts}: $\decode{\stack}$ and $\decode\ars$ are evaluation contexts of $\tarcal$ for every \TTAM constructor and activation stacks $\stack$ and $\ars$.
	\end{enumerate}
\end{lemma}
\begin{proof}
As for the \LTAM, see the proof of \reflemma{LTAM-decoding-invariants} at page \pageref{lab:LTAM-read-back-properties-proof}.\qedhere
\end{proof}

\begin{proposition}[\TTAM/$\toint$ implementation system]
\label{prop:TTAM-impl-requirements}
Let $\state$ be a \TTAM reachable state.
\begin{enumerate}
	\item \emph{Principal projection}: if $\state \tomachbeta \statetwo$ then $\decode\state \tollbv \decode\statetwo$ and if $\state \tomachproj \statetwo$ then $\decode\state \tollproj \decode\statetwo$.
	\item \emph{Overhead transparency}: if $\state \tomacho \statetwo$ then $\decode\state = \decode\statetwo$.
	\item \emph{Overhead termination}: $\tomacho$ terminates.
	\item \emph{Halt}: successful states read back to $\toint$-normal terms, and clash states read back to clashes of $\tarcal$.
\end{enumerate}
\end{proposition}

\begin{proof}
\hfill
\begin{enumerate}
		\item We only deal with $\tomachbetaev$, as for $\tomachprojev$ the proof is as for the \LTAM.
		\begin{itemize}
			\item \emph{Beta}: $\fourstate  {\evsym\pack{ \nvsym\tm}{\vv{\evval_1}}} {\vv{\evval_2} \cons\stack}  {\env}  {\ars}
\tomachbetaev
\fourstate{\nvsym  \tm}  {\emptystack}  {\vv{\evval_2};\vv{\evval_1}}  {\pair\stack\env \cons \ars}$.
Then,
			\begin{center}$\begin{array}{lclll}
				\decode{\fourstate  {\evsym\pack{ \nvsym\tm}{\vv{\evval_1}}} {\vv{\evval_2} \cons\stack}  {\env}  {\ars}}
				&= &
				\decodep\ars{\decode{\vv{\evval_2} \cons\stack}\sigma_\env\ctxholep{ \pack{ \nvsym\tm}{\vv{\decode{\evval_1}}} \sigma_\env }}
				\\[5pt]&= &
				\decodep\ars{\decode\stack\sigma_\env\ctxholep{ \pack{ \nvsym\tm}{\vv{\decode{\evval_1}}} \sigma_\env \, \vv{\decode{\evval_2}}\sigma_\env }}
				\\[5pt] &=^* & 
				\decodep\ars{\decode\stack\sigma_\env\ctxholep{ \pack{ \nvsym\tm}{\vv{\decode{\evval_1}}} \, \vv{\decode{\evval_2}}}}
				\\[5pt] &\totbv & 
				\decodep\ars{\decode\stack\sigma_\env\ctxholep{ \tm \itsubp{\tuvartwo}{\vv{\decode{\evval_1}}}{\tuvar}{\vv{\decode{\evval_2}}}}}
				\\[5pt] &= & 
				\decodep\ars{\decode\stack\sigma_\env\ctxholep{ \tm \sigma_{\vv{\evval_2};\vv{\evval_1}} }}
				\\[5pt] & = &				
				\decode{ \fourstate{\nvsym  \tm}  {\emptystack}  {\vv{\evval_2};\vv{\evval_1}}  {\pair\stack\env \cons \ars} } 
			\end{array}$\end{center}
		The $=^*$-step uses the closed values invariant (\refprop{TTAM-invariants}). The $\tollbv$ step is correct because  $\decodep\ars{\decode\stack\sigma_\env}$ is an evaluation context by the read back properties (\reflemma{TTAM-decoding-invariants}).
		\end{itemize}

		\item 
		\begin{itemize}
			\item $\fourstate {\nvsym \pack{ \tm}{\vv\prvar}}  \stack  \env  \ars
 \tomachsubwev 
\fourstate { \evsym \pack{ \nvsym\tm}{\vv{\env(\prvar)}} } \stack  \env  \ars$
. Then, 
			\begin{center}$\begin{array}{lllll}
				\decode{ \fourstate {\nvsym \pack{ \tm}{\vv\prvar}}  \stack  \env  \ars }
				&=&
				\decodep\ars{\decode\stack\sigma_\env\ctxholep{ \pack{ \tm}{\vv\prvar} }\sigma_\env }
				\\[5pt] & = &
				\decodep\ars{\decode\stack\sigma_\env\ctxholep{ \pack{ \tm}{\vv{\decode{\env(\prvar)}}} } }
				\\[5pt] & = &
				\decode{ \fourstate { \evsym \pack{ \nvsym\tm}{\vv{\env(\prvar)}} } \stack  \env  \ars }
			\end{array}$\end{center}

			\item $ \fourstate{\nvsym  \var}  \stack \env \ars
\tomachsubvev
\fourstate {\env(\prvar)}  \stack  \env  \ars$. Then,
			\begin{center}$\begin{array}{llllllllll}
				\decode{ \fourstate{\nvsym  \prvar}  \stack \env \ars }
				&=&
				\decodep\ars{\decode\stack\sigma_\env\ctxholep{\prvar\sigma_\env} }
				&= & 
				\decodep\ars{\decode\stack\sigma_\env\ctxholep{\decode{\env(\prvar)}}}
				&=&
				\decode{ \fourstate {\env(\prvar)}  \stack  \env  \ars }
			\end{array}$\end{center}
		\end{itemize}
		
		\item As for the \LTAM.
			
		\item 
		Let $\fourstate{\flag\tm}\stack\env\ars$ be a final state. Two cases for $\flag\tm$:
		\begin{itemize}
		\item \emph{Non-evaluated focus} $\nvlab\tm$. Then the only possibilities are that $\tm=\var$ with $\var\notin\dom\env$ or $\tm=\pack{\nvsym\tm}\tuvartwo$ with $\vartwo_i \notin\dom\env$ for some $i\leq\norm\tuvartwo$. By the closure invariant (\refprop{TTAM-invariants}), this is impossible. 
		
		\item \emph{Evaluated focus} $\evlab\val$. The only case which is (only slightly) different from the same case analysis done for the \LTAM is the following one.
				\begin{itemize}
				\item \emph{$\stack$ has an evaluated value $\evlab\valtwo$ on top}, that is, $\stack = \evlab\valtwo\cons\stacktwo$. Cases of $\evsym\val$:
				\begin{itemize}
				\item \emph{$\evsym\val$ is a tuple}. Then $\state$ is a clash state (clashing tuple) and its read back  $\decode\state = 				\decodep\ars{\decode{\evsym\valtwo\cons\stack}\sigma_\env\ctxholep{\decode\evval}} = 	\decodep\ars{\decode\stack\sigma_\env\ctxholep{\decode\evval\,\decode\evvaltwo}}$ is the corresponding kind of term clash.
				\item  \emph{$\evsym\val$ is a \wrapt $\evsym\pack{\nvsym\tm}{\vv{\evsym\valthree}}$}. Two sub-cases.
				\begin{itemize}
				\item  \emph{$\evsym\valtwo$ is not a tuple}. Then $\state$ is a clash state (clashing \wrapt with wrong argument) and its read back   is the corresponding kind of term clash.
				\item  \emph{$\evsym\valtwo$ is a tuple}. This case is impossible because then $\state$ can do a $\tomachbetaev$ transition, which is absurd.\qedhere
				\end{itemize}
				\end{itemize}
				\end{itemize}
		\end{itemize}
	\end{enumerate}
\end{proof}

\gettoappendix{thm:tar-ttam-implementation}
\begin{proof}
By \refprop{TTAM-impl-requirements}, \TTAM and $\totar$ form an implementation system on prime terms of $\tarcal$ (because a requirement for the initial states of the \TTAM is that the term is prime), thus \TTAM implements $\totar$ by the abstract theorem \refthm{abs-impl}.\qedhere
\end{proof}
\section{Proofs of \refsect{sharing} (Sharing, Size Explosion, and Tuples)}
\label{sect:app-sharing}

\gettoappendix{prop:size-explosion-plotkin}

\begin{proof}
By induction on $n$. The base case $n=0$ holds, because $\tm_0 = \Id = \tmtwo_0$ and $\tobv^0$ is the 
identity. The inductive case: $\tm_{n+1} = \pi \tm_n \tobv^n \pi \tmtwo_n = (\la\var\la\vartwo\vartwo 
\var\var) \tmtwo_n \tobv \la\vartwo\vartwo\tmtwo_n \tmtwo_n = \tmtwo_{n+1}$, where the first sequence 
is obtained by the \ih, and the last step by the fact that $\tmtwo_n$ is a value by \ih The bounds on the 
sizes are immediate, as well as the fact that $\tmtwo_{n+1}$ is a value.
\end{proof}

\gettoappendix{prop:size-explosion-tuples}

\begin{proof}
By induction on $n$. The base case $n=0$ holds, because $\tmthree_0 = \Id = \tmfour_0$ and $\tobv^0$ is the 
identity. The inductive case: $\tmthree_{n+1} = \tau \tuple{\tmthree_{n}} \tobv^{n} \tau \tuple{\tmfour_{n}} =(\la\var\tuple{\var,\var}) \tuple{\tmfour_{n}} \tobv \tuple{\tmfour_{n},\tmfour_{n}} = \tmfour_{n+1}$, where the first sequence 
is obtained by the \ih, and the last step  by the fact that $\tmfour_n$ is a value by \ih The bounds on the 
sizes are immediate, as well as the fact that $\tmfour_{n+1}$ is a value.
\end{proof}

\section{Proofs and Auxiliary Notions of \refsect{source-complexity} (Complexity Analysis of the \TAM)}
\label{sect:app-source-complexity}

\gettoappendix{l:TAM-sub-term-invariant}
\begin{proof}
	By induction on the length $\length{\exec}$ of the execution $\exec \colon \compil{\tm} \totam^{*} \state$. 
	If $\length{\exec} = 0$ then the statement trivially holds, since $\state = \compil{\tm} = \pairstate{\nvsym\pair\tm\emptyenv}\emptystack$.
	
	Otherwise $\exec \colon \compil{\tm} \totam^{*} \state$ is the concatenation of $\exectwo \colon \compil{\tm} \totam^{*} \statetwo$ and $\statetwo \totam \state$.
	By \ih, the statement holds for $\exectwo$. A straightforward analysis of the transition rules of the machine proves that the statement is preserved by the transition $\statetwo \totam \state$.
	\qedhere
\end{proof}

\gettoappendix{l:TAM-blue-bound}
\begin{proof}
Every transition $\tomachseaoneev$ in $\exec$ must be preceded by a transition $\tomachseaonenv$ (otherwise there would not be an entry of shape $\nvsym\bareclos$ in the   stack), thus $\sizep\exec{\evsym sea_{1}}   \leq \sizep\exec{\nvsym sea_{1}}$.
	
	Every transition $\tomachseathreeev$ in $\exec$ must be preceded by a transition $\tomachseathreenv$ (otherwise there would not be an entry of shape $\tuple{\machctxhole,\mydots}$ in the stack), thus $\sizep\exec{\evsym sea_{3}}   \leq \sizep\exec{\nvsym sea_{3}}$.
	
	Every transition $\tomachprojev$ in $\exec$ must be preceded by a transition $\tomachseatwonv$ (otherwise there would not be an entry of shape $\proj_{i}$ in the  stack), thus $\sizep\exec{
	\projsym}   \leq \sizep\exec{\nvsym sea_{2}}$.	\qedhere
\end{proof}

\gettoappendix{l:TAM-overhead}
\begin{proof}
\label{proof:ov-meas}\hfill
	\begin{enumerate}
		\item \emph{Beta}: $\state = \pairstate{\evsym (\la{\tuvar}\tmtwo,  \env)}{\vv\clos \cons\stack}
		\tomachbeta \pairstate{
			\nvsym (\tmtwo, \esub\tuvar{\vv\clos} \cons \env)} {\stack} = \statetwo$. 
		Then, $\omeas\statetwo \leq \omeas\state + \size{\tm}$ because:
		\begin{align*}
			\omeas{\state} &= \omeas{\stack} 
			& 
			\omeas{\statetwo} &= \size\tmtwo + \omeas{\stack}  \overset{\reflemmaeq{TAM-sub-term-invariant}}{\leq} \size{\tm} + \omeas{\stack}
		\end{align*}
	
		\item 
		\begin{enumerate}
			\item $\state = \pairstate{\nvsym (\tm\tmtwo, \env)}{\stack }
			\tomachseaonenv
			\pairstate{\nvsym (\tmtwo, \env)} { \nvlab{\pair\tm\env} \cons\stack } =\statetwo$.
			Then, $\omeas{\statetwo} < \omeas{\state}$ because:
			\begin{align*}
				\omeas{\state} &= 1 + \size\tm + \size\tmtwo + \omeas{\stack} 
				& 
				\omeas{\statetwo} &= \size\tm + \size\tmtwo + \omeas{\stack} 
			\end{align*}
			\item $\state=\pairstate{\nvsym (\proj_i \tm , \env) }{ \stack} 
			\tomachseatwonv 
			\pairstate{\nvsym (\tm,\env)}{ \proj_i \cons \stack }=\statetwo$. 
			Then, $\omeas{\statetwo} < \omeas{\state}$ because:
			\begin{align*}
				\omeas{\state} &= 1 + \size\tm +  \omeas{\stack} 
				& 
				\omeas{\statetwo} &= \size\tm + \omeas{\stack} 
			\end{align*}
			
			\item $\state=\pairstate{\nvsym (\tuple{\mydots,\tm_n}, \env)} {\stack} 
			\tomachseathreenv 
			\pairstate{\nvsym   (\tm_n, \env)  }{\pair{\tuple{\mydots,\machctxhole}}\env \cons \stack}=\statetwo$ with  $ n >0$. 
			Then, $\omeas{\statetwo} < \omeas{\state}$ because:
			\begin{align*}
				\omeas{\state} &= n + \sum_{i=1}^n\size{\tm_i} +  \omeas{\stack} 
				& 
				\omeas{\statetwo} &= n-1 +\sum_{i=1}^n\size{\tm_i} + \omeas{\stack} 
			\end{align*}
			
			\item $\state=\pairstate{\nvsym (\tuple{}, \env) }{\stack}  \tomachseafournv 
			\pairstate{\evsym (\tuple{}, \emptyenv) } {\stack} =\statetwo$. 
			Then, $\omeas{\statetwo} < \omeas{\state}$ because:
			\begin{align*}
				\omeas{\state} &= 1 +   \omeas{\stack} 
				& 
				\omeas{\statetwo} &= \omeas{\stack} 
			\end{align*}
			
			\item $\state=\pairstate{\nvsym  (\la\tuvar\tm, \env)} {\stack} 
			\tomachseafivenv 
			\pairstate{\evsym (\la\tuvar\tm , \env)} {\stack }=\statetwo$. 
			Then, $\omeas{\statetwo} < \omeas{\state}$ because:
			\begin{align*}
				\omeas{\state} &= \size{\la\tuvar\tm} +   \omeas{\stack} 
				& 
				\omeas{\statetwo} &= \omeas{\stack} 
			\end{align*}
			
			\item $\state=\pairstate{\nvsym (\var, \env)} {\stack} 
			\tomachsub
			\pairstate{\env(\var)} {\stack} =\statetwo$.
			Then, $\omeas{\statetwo} < \omeas{\state}$ because $\env(\var) = \clos$ and:
			\begin{align*}
				\omeas{\state} &= 1 +   \omeas{\stack} 
				& 
				\omeas{\statetwo} &= \omeas{\stack} 
			\end{align*}
			
			\item $\state=\pairstate{\clos} {\pair{\tuple{\mydots,\tm_n,\machctxhole,\mydots}}\env \cons\stack }
			\tomachseasixev
			\pairstate{\nvsym (\tm_n ,\env) }  {\pair{\tuple{\mydots,\tm_{n-1},\machctxhole,\clos,\mydots}}\env \cons\stack }=\statetwo$ with $n > 0$.
			Then, $\omeas{\statetwo} < \omeas{\state}$ because $n > 0$ and:
			\begin{align*}
				\omeas{\state} &= n + \sum_{i=1}^n\omeas{\tm_i} + \omeas{\stack} 
				& 
				\omeas{\statetwo} &= n-1 + \sum_{i=1}^n\omeas{\tm_i} + \omeas{\stack} 
			\end{align*}
		\end{enumerate}
	
		\item 
		\begin{enumerate}
			\item $\state = \pairstate{\vv\clos} {\proj_i \cons \stack}
			\tomachproj 
			\pairstate{\clos_i}{\stack} = \statetwo $ with $i \leq \norm{\vv\clos}$.
			Then, $\omeas{\statetwo} = \omeas{\stack} = \omeas{\state}$.
			\item $\state = \pairstate{\clos}{ \nvlab{\pair\tm\env} \cons\stack} 
			\tomachseaoneev
			\pairstate{\nvsym (\tm,\env)} {\clos \cons\stack} = \statetwo$. 
			Then, $\omeas{\statetwo} = \size\tm + \omeas{\stack} = \omeas{\state}$.
			\item $\state = \pairstate{\clos}{ \pair{\tuple{\machctxhole,\mydots}}\env \cons\stack} 
			\tomachseathreeev
			\pairstate{\tuple{\clos,\mydots}} {\stack} = \statetwo$. 
			Then, $\omeas{\statetwo} = \omeas{\stack} = \omeas{\state}$.\qedhere
		\end{enumerate}

	\end{enumerate}
\end{proof}

\gettoappendix{prop:TAM-number-of-trans}
\begin{proof}
First, we prove $\omeas\state \leq \omeas{\compil\tm} + \sizep\exec{\evsym\betav}\cdot\size\tm-\sizep\exec{\nvsym\subsym,\nvsym sea_{1-5},\evsym sea_{6}}$, and then we deduce the statement. The proof is by induction on the length $\size\exec$ of $\exec$. If $\exec$ is empty then the inequality becomes $\omeas\state= \omeas{\compil\tm} \leq \omeas{\compil\tm}$, which is true. If $\exec$ is non-empty then one considers the last transitions of $\exec$, that is, let $\exectwo:\compil\tm \totam^{\size\exec-1} \statetwo$ be $\exec$ without its last transition $\statetwo \totam \state$. By \ih,  $\omeas\state \leq \omeas{\compil\tm} + \sizep\exectwo{\evsym\betav}\cdot\size\tm-\sizep\exectwo{\nvsym\subsym,\nvsym sea_{1-5},\evsym sea_{6}}$. Case of this last transition:
\begin{itemize}
\item $\tomachbetaev$: by \reflemma{TAM-overhead}.1, the left-hand side of the inequality grows of at most $\size\tm$, which is also the increase of the right-hand side.
\item $\tomachhole{\nvsym\subsym,\nvsym sea_{1-5},\evsym sea_{6}}$: by \reflemma{TAM-overhead}.2, the left-hand side of the inequality decreases of at least $1$, and the right-hand side decreases of 1.
\item $\tomachhole{\evsym\proj,\evsym sea_{1-3}}$: by \reflemma{TAM-overhead}.3, both the left-hand and the right-hand side of the inequality do not change.
\end{itemize}
Thus, we proved the inequality. It follows that $\sizep\exec{\nvsym\subsym,\nvsym sea_{1-5},\evsym sea_{6}} \leq \omeas{\compil\tm} + \sizep\exec{\evsym\betav}\cdot\size\tm$. By definition, $\omeas{\compil{\tm}} = \omeas{\nvsym(\tm,\emptyenv)} = \size\tm$. Then, 
\[\omeas{\compil\tm} + \sizep\exec{\evsym\betav}\cdot\size\tm \ \ \leq \ \ \size\tm + \sizep\exec{\evsym\betav}\cdot\size\tm \ \ =\ \  (\sizep\exec{\evsym\betav}+ 1)\cdot\size\tm.\]
Summing up, the inequality is:
\[\sizep\exec{\nvsym\subsym,\nvsym sea_{1-5},\evsym sea_{6}} \leq (\sizep\exec{\evsym\betav}+ 1)\cdot\size\tm\]
Now, we finally obtain our bound of $\size\exec$ by applying twice the derived inequality:
\begin{center}
$\begin{array}{llll}
	\size{\exec} &= & \sizep\exec{\nvsym\subsym,\nvsym sea_{1-5}} + \sizep\exec{\evsym \betav}+ \sizep\exec{\evsym\proj,\evsym sea_{1,3,6}}
	\\
	&\leq_{\reflemmaeq{TAM-blue-bound}} &\sizep\exec{\nvsym\subsym,\nvsym sea_{1-5}} + \sizep\exec{\evsym \betav}+ \sizep\exec{\nvsym sea_{1-3},\evsym sea_{6}}
	\\
	&\leq & \sizep\exec{\evsym \betav} + 2\cdot(\sizep\exec{\evsym\betav}+ 1)\cdot\size\tm
	\\
	&\leq & 3\cdot(\sizep\exec{\evsym\betav}+ 1)\cdot\size\tm
\end{array}$
\end{center}
Where the last inequality is justified by the fact that $\size{\tm} \geq 1$.
\qedhere
\end{proof}

\section{Proofs and Auxiliary Notions of \refsect{int-complexity} (Complexity Analysis of the \TTAM)}
\label{sect:app-int-complexity}

\gettoappendix{l:LTAM-sub-term-invariant}
\begin{proof}
	By induction on the length $\length{\exec}$ of the execution $\exec \colon \compil{\tm} \tottam^{*} \state$. 
	If $\length{\exec} = 0$ then the statement trivially holds, since $\state = \compil{\tm} = \fourstate{\nvsym\tm} \emptyenv\emptyenv\emptyenv$.
	
	Otherwise $\exec \colon \compil{\tm} \tottam^{*} \state$ is the concatenation of $\exectwo \colon \compil{\tm} \tottam^{*} \statetwo$ and $\statetwo \tottam \state$.
	By \ih, the statement holds for $\exectwo$. A straightforward analysis of the transition rules of the machine proves that the statement is preserved by the transition $\statetwo \tottam \state$.
	\qedhere
\end{proof}

\paragraph{Number of (Overhead) Transitions} We follow the schema used for the \TAM. The new transition $\tomachseasevenev$ is part of the transitions that are factored out, since each $\tomachseasevenev$ transition is enabled by a $\tomachbetaev$ transition adding an entry to the activation stack.

\begin{lemma}[Transitions match]
	\label{l:LTAM-blue-bound}
Let $\exec \colon \state \tottam^{*} \statetwo$ be an execution. Then $\sizep\exec{\proj,\evsym sea_{1,3,7}}  \leq \sizep\exec{\nvsym sea_{1,2,3},\evsym\betav}$.
\end{lemma}
\begin{proof}
	Every transition $\tomachseaoneev$ in $\exec$ must be preceded by a transition $\tomachseaonenv$ (otherwise there would not be an entry of shape $\nvsym\tmtwo$ in the constructor  stack), thus $\sizep\exec{\evsym sea_{1}}   \leq \sizep\exec{\nvsym sea_{1}}$.
	
	Every transition $\tomachseathreeev$ in $\exec$ must be preceded by a transition $\tomachseathreenv$ (otherwise there would not be an entry of shape $\tuple{\machctxhole,\mydots}$ in the constructor stack), thus $\sizep\exec{\evsym sea_{3}}   \leq \sizep\exec{\nvsym sea_{3}}$.
	
	Every transition $\tomachprojev$ in $\exec$ must be preceded by a transition $\tomachseatwonv$ (otherwise there would not be an entry of shape $\proj_{i}$ in the constructor stack), thus $\sizep\exec{
	\projsym}   \leq \sizep\exec{\nvsym sea_{2}}$.

	Every transition $\tomachseasevenev$ in $\exec$ must be preceded by a transition $\tomachbetaev$ (otherwise there would not be an entry on the activation stack), thus $\sizep\exec{\evsym sea_7}   \leq \sizep\exec{\evsym\betav}$.	\qedhere
\end{proof}

The overhead measure of the \TTAM is defined in \reffig{ovh-measure-ltam}. It is a direct adaptation to the \TTAM of the measure given for the \TAM (see \reffig{ovh-measure-tam}, page \pageref{fig:ovh-measure-tam} of the paper). Note indeed that the measure ignores environments, which are the main difference between the two machines.

\begin{figure}[t!]
\centering
\fbox{
\setlength{\arraycolsep}{3pt}
\begin{tabular}{c}
$\begin{array}{r@{\hspace{.5cm}} rll@{\hspace{.5cm}} rlllrllrllllllll}
\textsc{Terms}
&
\omeas{\proj_i \efvar} & \defeq & 1
&
\omeas{\proj_i \envar} & \defeq & 1
\\
 &
\omeas{\packnm\tm\bag } & \defeq & \omeas\tm + n + m
&
\omeas{ \tm\,\tmtwo} & \defeq & \omeas{ \tm} + \omeas{ \tmtwo} + 1
\\
&&&&\omeas{\tuple{\tm_{1},\mydots,\tm_{n}}}
 &\defeq&  n + \sum_{i=1}^n\size{\tm_{i}}

\end{array}$
\\\hline
$\begin{array}{r@{\hspace{.5cm}} rll@{\hspace{.5cm}} rllrllrllllllll}
\textsc{Con. stack entries $\stacke$}
&
\omeas{\evval} & \defeq & 0
&
\omeas{\nvsym\tm} & \defeq & \size\tm
\\
 &
\omeas{\proj_{i}} & \defeq & 0
&
\omeas{\tuple{\nvsym\tm_{1},\mydots,\nvsym\tm_{n},\machctxhole,\tuv\evval}}
 &\defeq&  n + \sum_{i=1}^n\size{\tm_{i}}
\\
\textsc{Constructor stacks}
&\omeas\emptystack & \defeq & 0
& \omeas{\stacke\cons\stack} & \defeq & \omeas\stacke + \omeas \stack
\\
\textsc{Activation stacks}
&\omeas\emptystack & \defeq & 0
& \omeas{\pair\stack\env\cons\ars} & \defeq & \omeas\stack + \omeas \ars
\\
\textsc{States} &&&&
\omeas{\fourstate{\flag\tm} \stack \env \ars} & \defeq & \omeas{\flag\tm} + \omeas \stack + \omeas\ars
\end{array}$
\end{tabular}
}
\caption{Definition of the overhead measure for the \LTAM.}
\label{fig:ovh-measure-ltam}
\end{figure}

\begin{lemma}
	\label{l:LTAM-overhead}
Let $\compil\tm \tottam^{*} \state$ and $\state \tomachhole{a} \statetwo$.
\begin{enumerate}
\item
 \emph{$\beta$ increases the measure}: $\omeas\statetwo \leq \omeas\state + \size{\tm}$ if $a = \evsym\betav$;
\item
\emph{Unmatched transitions decrease the measure}: $\omeas\statetwo < \omeas\state$ if $a \in\set{\nsubwsym, \nsubvsym, \nvsym sea_{1-5}, \evsym sea_{6}}$;
\item 
\emph{Matched transitions do not change the measure}: $\omeas\statetwo = \omeas\state$ if $a \in\set{\evsym sea_{1,3,7}, \evsym \projsym}$.
\end{enumerate}
\end{lemma}
\begin{proof}
\hfill
	\begin{enumerate} 
			\item \emph{Beta}: $\state = 	\fourstate{\evsym  \pack{ \nvsym\tmtwo}{\vv{\evval_1}}}  {\vv{\evval_2} \cons\stack}  \env  \ars
\tomachbetaev
\fourstate{\nvsym  \tmtwo}   \emptystack {\vv{\evval_2};\vv{\evval_1}}  {\pair\stack\env \cons \ars}
 = \statetwo$. 
		Then, $\omeas\statetwo \leq \omeas\state + \size{\tm}$ because:
		\begin{align*}
			\omeas{\state} &= \omeas{\stack} +\omeas\ars
			& 
				\omeas{\statetwo} &= \size\tmtwo + \omeas{\stack}  \overset{\reflemmaeq{LTAM-sub-term-invariant}}{\leq} \size{\tm} + \omeas{\stack} +\omeas\ars
		\end{align*}

		\item Since the local environments do not play a role in the overhead measure for the \TAM, the reasoning is as for the \TAM (see \reflemma{TAM-overhead}, the proof of which is at page \pageref{proof:ov-meas}) for all transitions but the new specific ones of the \TTAM $\nsubwsym$ and $\nsubvsym$, and omitted. For $\nsubwsym$ and $\nsubvsym$:
		\begin{itemize}
		\item $\nsubwsym$: then $\state = \fourstate{\nvsym \pack{ \tmtwo}{\vv\prvar}}  \stack  \env  \ars
 \tomachsubwev 
\fourstate {\evsym \pack{ \nvsym\tmtwo}{\vv{\env(\prvar)}}}  \stack  \env  \ars
 = \statetwo$. 
		Then, $\omeas\state > \omeas\statetwo$ because the size of terms is positive (so that $\size{\nvsym \pack{ \tmtwo}{\vv\prvar}} \geq \size\tmtwo >0$) and:
		\begin{align*}
			\omeas{\state} &\geq \size\tmtwo + \omeas{\stack} +\omeas\ars
			& 
				\omeas{\statetwo} &= \omeas{\stack} +\omeas\ars
		\end{align*}

		\item $\nsubvsym$: then $\state = \fourstate{\nvsym  \prvar}  \stack  \env  \ars
 \tomachsubvev 
\fourstate {\env(\prvar)}  \stack  \env  \ars
 = \statetwo$. 
		Then, $\omeas\state> \omeas\statetwo$ because (since $\env(\prvar)$ is an evaluated value, thus of measure 0):
		\begin{align*}
			\omeas{\state} &= 1+ \omeas{\stack} +\omeas\ars
			& 
				\omeas{\statetwo} &= \omeas{\stack} +\omeas\ars
		\end{align*}

		\end{itemize}

		\item We only show the new transition $\tomachseasevenev$:
		\begin{enumerate}
			\item $\state = \fourstate{\evsym  \val}  \emptystack \env {\pair\stack\envtwo\cons \ars}
\tomachseasevenev
\fourstate{\evsym  \val} \stack \envtwo \ars = \statetwo$. Then, $\omeas\state = \omeas\stack + \omeas\ars = \omeas\statetwo$.
\qedhere
		\end{enumerate}

	\end{enumerate}
\end{proof}

\gettoappendix{prop:LTAM-number-of-trans}
\begin{proof}
The proof is essentially as for the \TAM. We detail it anyway.
First, we prove $\omeas\state \leq \omeas{\compil\tm} + \sizep\exec{\evsym\betav}\cdot\size\tm-\sizep\exec{\nsubwsym, \nsubvsym,\nvsym sea_{1-5},\evsym sea_{6}}$, and then we deduce the statement. The proof is by induction on the length $\size\exec$ of $\exec$. If $\exec$ is empty then the inequality becomes $\omeas\state= \omeas{\compil\tm} \leq \omeas{\compil\tm}$, which is true. If $\exec$ is non-empty then one considers the last transitions of $\exec$, that is, let $\exectwo:\compil\tm \tottam^{\size\exec-1} \statetwo$ be $\exec$ without its last transition $\statetwo \tottam \state$. By \ih,  $\omeas\state \leq \omeas{\compil\tm} + \sizep\exectwo{\evsym\betav}\cdot\size\tm-\sizep\exectwo{\nsubwsym, \nsubvsym,\nvsym sea_{1-5},\evsym sea_{6}}$. Case of this last transition:
\begin{itemize}
\item $\tomachbetaev$: by \reflemma{LTAM-overhead}.1, the left-hand side of the inequality grows of at most $\size\tm$, which is also the increase of the right-hand side.
\item $\tomachhole{\nsubwsym, \nsubvsym,\nvsym sea_{1-5},\evsym sea_{6}}$: by \reflemma{LTAM-overhead}.2, the left-hand side of the inequality decreases of at least $1$, and the right-hand side decreases of 1.
\item $\tomachhole{\evsym\proj,\evsym sea_{1,3,7}}$: by \reflemma{LTAM-overhead}.3, both the left-hand and the right-hand side of the inequality do not change.
\end{itemize}
Thus, we proved the inequality. It follows that $\sizep\exec{\nsubwsym, \nsubvsym,\nvsym sea_{1-5},\evsym sea_{6}} \leq \omeas{\compil\tm} + \sizep\exec{\evsym\betav}\cdot\size\tm$. By definition, $\omeas{\compil{\tm}} = \omeas{\fourstate{\nvsym\tm}\emptyenv\emptyenv\emptyenv} = \size{\tm}$. Then, 
\[\omeas{\compil\tm} + \sizep\exec{\evsym\betav}\cdot\size\tm \ \ \leq \ \ \size\tm + \sizep\exec{\evsym\betav}\cdot\size\tm \ \ =\ \  (\sizep\exec{\evsym\betav}+ 1)\cdot\size\tm.\]
Summing up, the inequality is:
\[\sizep\exec{\nsubwsym, \nsubvsym,\nvsym sea_{1-5},\evsym sea_{6}} \leq (\sizep\exec{\evsym\betav}+ 1)\cdot\size\tm\]
Now, we finally obtain our bound of $\size\exec$ by applying twice the derived inequality:
\begin{center}
$\begin{array}{llll}
	\size{\exec} &= & \sizep\exec{\nsubwsym, \nsubvsym,\nvsym sea_{1-5}} + \sizep\exec{\evsym \beta}+ \sizep\exec{\evsym\proj,\evsym sea_{1,3,6}}
	\\
	&\leq_{\reflemmaeq{LTAM-blue-bound}} &\sizep\exec{\nsubwsym, \nsubvsym,\nvsym sea_{1-5}} + \sizep\exec{\evsym \beta}+ \sizep\exec{\nvsym sea_{1-3},\evsym sea_{6}}
	\\
	&\leq & \sizep\exec{\evsym \beta} + 2\cdot(\sizep\exec{\evsym\betav}+ 1)\cdot\size\tm
	\\
	&\leq & 3\cdot(\sizep\exec{\evsym\betav}+ 1)\cdot\size\tm
\end{array}$
\end{center}
Where the last inequality is justified by the fact that $\size{\tm} \geq 1$.
\qedhere
\end{proof}

\gettoappendix{l:lifting-growth-bound}
\begin{proof}
	Point 1: the size increment in $\lali\tm$ is due to the \bagts introduced by 
	wrapping abstractions, and it is proportional to the number of free variables in the body of the abstraction, 
	bounded by $\hg\tm$; so $\size{\lali\tm}\in \bigo(\hg\tm\cdot\size\tm)$. 
	Point 2: 
	take $\tm_n \defeq \la{\var_{1}}\mydots\la{\var_{n}} \var_{1}\var_2\ldots\var_{n} \allowbreak\neq \la{\var_{1},\mydots,\var_{n}} \var_{1}\var_2\ldots\var_{n}$ (note $n$ abstractions in $\tm_n$, not~just~$1$).
\end{proof}
		
\section{Implementation in OCaml of the \TTAM}
\label{sect:app-implementation}

The submission comes with an implementation in OCaml of the \TTAM,
available on GitHub \cite{PPDP25ocaml}.

\paragraph{Motivations} The aims of the implementation are:
\begin{enumerate}
 \item Supporting the claim that all \TTAM{} transitions can be
   implemented in constant time, with the exception of $\tomachseathreenv$ and $\tomachsubwev$ whose complexity is linear in the width of the
    tuple/bag iterated over by the rule (in the paper $\tomachseathreenv$ and $\tomachsubwev$ are then considered $\bigo(1)$ operations because the linear cost is amortized over the cost of search in complete machine runs);
 \item Giving the user the possibility of observing \TTAM{} executions on terms of their own
   choice, to better grasp the way the machine works.
\end{enumerate}
The implementation is \emph{not} meant to be optimized, nor to represent
an instance of state-of-the-art OCaml code. In particular, the data
structures have been chosen to make the datatypes and code very close
to those in the paper, to allow the reader to verify that the
implementation reflects its specification.

\paragraph{User Interface and Provided Examples} Once compiled following the instructions in
the \texttt{README.md} file, the user can start a Read-Eval-Print-Loop
(REPL) by typing \texttt{dune exec main}. The loop requires the user to
enter a source term, which is then 
\begin{enumerate}
\item Wrapped to intermediate code,
\item Turned into target code by eliminating names, and then 
\item Run according to the transitions of  the \TTAM.
\end{enumerate}
Every intermediate machine transition is printed on the standard output.
The internal representation of the parsed source and intermediate terms
are also printed before starting reduction.

The exact syntax accepted by the REPL is printed when the command is run.
It is the user responsibility to enter closed and clash-free terms.
When one of the latter two conditions is violated, the REPL will abort.

The \texttt{TESTS} file includes a few examples of closed, clash-free
source terms that can be used as tests by typing them at the REPL prompt
or by feeding the whole file to the executable via \texttt{dune exec main < TESTS}.

\paragraph*{Code Structure} The code consists of the following files:
\begin{itemize}
  \item \texttt{lexer.mll/parser.mly} where the lexer and parser for
    source terms are implemented
  \item \texttt{term.ml} composed of three sub-modules
    \texttt{Source/Intermediate/Target} respectively for the source,
    intermediate and target languages. Each submodule defined the
    algebraic type \texttt{term} for the terms of the language.
    The wrapping and name elimination functions to turn source terms
    into intermediate terms and intermediate terms into target terms
    are implemented in the sub-modules of the function codomain.

    Note that, in the memory model of OCaml, every subterm is represented
    in memory as a pointer to its root. Therefore pattern matching over
    a term and calling a function on the pattern variables that hold
    the subterms is done in constant time. This is required to conclude
    that the implementation of most machine transitions requires constant
    time.
  \item \texttt{machine.ml} where the \TTAM{} is implemented.
   The module begins with the definition of
   algebraic data types for constructor and activation stacks,
   environments and states of a machine. Then a pretty-printing sub-module \texttt{PP} allows to turn all of
   the previous types to string. The code is not written with efficiency
   in mind. The functions are parameterized by a boolean stating if the
   term has already been evaluated or not. The boolean is propagated
   during recursion according to the machine invariants and it is used
   to decorate every subterm with $\nvsym/\evsym$.
   Outside the \texttt{PP} sub-module, the last two functions of the file
   are \texttt{run}, which implements the machine transitions and main
   loop, and \texttt{reduce}, that builds an initial machine state from
   a target term, computes its normal form using \texttt{run} and
   prints it.
\end{itemize}

\paragraph*{Differences Between the Code and the Paper} The definitions
in the paper and the code are almost in one-to-one correspondence.
There is no distinction in the code between values and terms.
Arrays are used to represent bags and tuples and to provide access in
constant time.

The only significant difference is in the stack entry
for tuples under evaluation $\tuple{\tuv{\nvsym\tm},\machctxhole,\tuv{\evsym\val}}$ that is represented in the code as the array of terms
$\tuv{\nvsym\tm}\tmtwo\tuv{\evsym\val}$ together with the index $i$
of $\machctxhole$, so that the $i$-th element of the array is
a non-meaningful term $\tmtwo$ that replaces the $\machctxhole$
placeholder. This choice allows us to use arrays of terms both for tuples
and for stack entries, simplifying the implementation in two ways:
it allows us to do a shallow copy of the tuple in input in the
$\tomachseathreenv$ transition, to obtain the needed stack item;
it also allows us to avoid a shallow copy of the array in the
$\tomachseathreeev$ and $\tomachseathreeev$ rules.

\end{document}